\documentclass[aps,pra,superscriptaddress,nofootinbib,onecolumn]{revtex4-2} 

\usepackage{soul}
\usepackage[pdftex]{graphicx}
\usepackage{dcolumn}
\usepackage{bm}
\usepackage{epsfig}
\usepackage{amsmath}
\usepackage{amssymb} 
\usepackage{amsfonts}
\usepackage{pifont} 
\usepackage{multirow}
\usepackage{subfigure}
\usepackage{float}
\usepackage[utf8]{inputenc} 
\usepackage{epstopdf}
\usepackage{dsfont}
\usepackage{MnSymbol}
\usepackage{xspace}
\usepackage{tcolorbox}
\usepackage{multirow}

\usepackage{xcolor}

\usepackage{amsthm}
\theoremstyle{definition}
\newtheorem{theorem}{Theorem}
\newtheorem*{theorem*}{Theorem}
\newtheorem{lemma}[theorem]{Lemma}
\newtheorem{corollary}[theorem]{Corollary}
\newtheorem*{corollary*}{Corollary}
\newtheorem{proposition}[theorem]{Proposition}
\newtheorem*{proposition*}{Proposition}
\newtheorem{definition}[theorem]{Definition}
\newtheorem*{definition*}{Definition}
\newtheorem{example}[theorem]{Example}

\def\g{\mathfrak{g}}
\def\gc{\mathfrak{c}}

\def\N{\mathbb{N}}
\def\C{\mathbb{C}}
\def\R{\mathbb{R}}

\newcommand{\Np}[1]{\N^{#1}_{\geq 0}}

\newcommand{\no}[1]{\mathcal{N}(#1)}
\newcommand{\mdeg}{\mathrm{mdeg}}
\newcommand{\abs}[1]{|#1|}

\newcommand{\gspan}{\mathrm{span}}
\newcommand{\vspan}[1]{\mathrm{span}\{ #1 \}{}}
\newcommand{\lie}[1]{\langle #1 \rangle_{\mathrm{Lie}}{}}

\newcommand{\inner}[2]{\langle #1 | #2 \rangle{}}
\newcommand{\gens}{\text{gens}}
\newcommand{\PP}[2]{\mathcal{P}_{#2}(#1)}
\newcommand{\ad}[1]{(a^\dag)^{#1}}
\newcommand{\bd}[1]{(b^\dag)^{#1}}

\usepackage{bbm}

\newcommand{\oneshalf}{\iota}
\newcommand{\ones}{\tau}

\newcommand{\multiset}[1]{\langle #1 \rangle}
\newcommand{\multisetnarrow}[1]{\multiset{#1}}
\newcommand{\emptymultiset}{\multiset{\,}}
\newcommand{\multy}[2]{{#1}^{\times #2}}
\newcommand{\up}{\simeq}
\newcommand{\upto}[1]{\simeq_{#1}}

\usepackage[bookmarks=false,pdfstartview={FitH}]{hyperref}

\usepackage{tikz,xcolor}
\definecolor{lime}{HTML}{A6CE39}
\DeclareRobustCommand{\orcidicon}{%
	\begin{tikzpicture}
	\draw[lime, fill=lime] (0,0) 
	circle [radius=0.16] 
	node[white] {{\fontfamily{qag}\selectfont \tiny ID}};
	\draw[white, fill=white] (-0.0625,0.095) 
	circle [radius=0.007];
	\end{tikzpicture}
	\hspace{-2mm}
}
\foreach \x in {A, ..., Z}{%
	\expandafter\xdef\csname orcid\x\endcsname{\noexpand\href{https://orcid.org/\csname orcidauthor\x\endcsname}{\noexpand\orcidicon}}
}

\begin{document}

\title{Deciding finiteness of bosonic dynamics with tunable interactions}

\author{David Edward Bruschi\,\orcidB{}}
\email{david.edward.bruschi@posteo.net}	
\email{d.e.bruschi@fz-juelich.de}	
\affiliation{Quantum Computing Analytics (PGI-12), Forschungszentrum J\"ulich, 52425 J\"ulich, Germany}
\affiliation{Theoretical Physics, Universit\"at des Saarlandes, 66123 Saarbr\"ucken, Germany}
\author{Andr\'{e} Xuereb\,\orcidC{}}
\email{andre.xuereb@um.edu.mt}
\affiliation{Department of Physics, University of Malta, Msida MSD 2080, Malta}
\author{Robert Zeier\,\orcidA{}}
\email{r.zeier@fz-juelich.de}
\affiliation{Forschungszentrum Jülich GmbH, Peter Grünberg Institute, Quantum Control (PGI-8), 52425 Jülich, Germany}

\begin{abstract}
In this work we are motivated by factorization
of bosonic quantum dynamics and
we study the corresponding Lie algebras,
which can potentially be infinite dimensional.
To characterize such factorization,
we identify conditions for these Lie algebras
to be finite dimensional.
We consider cases where each free Hamiltonian term
is itself an element of the generated Lie algebra.
In our approach, we develop new tools to systematically divide skew-hermitian bosonic operators
into appropriate subspaces, and construct specific sequences of
skew-hermitian operators that are used to gauge the dimensionality of the Lie algebras themselves.
The significance of our result relies on conditions that constrain only the 
independently controlled generators
in a particular Hamiltonian, thereby providing an effective algorithm for verifying the finiteness of the generated
Lie algebra.
In addition, our results are tightly connected to 
mathematical work where the polynomials of creation and annihilation operators
are known as the Weyl algebra.
Our work paves the way for better understanding
factorization of bosonic dynamics relevant to quantum control and quantum technology.
\end{abstract}


\date{February 14, 2024}

\maketitle

\section*{Introduction}
\emph{When are elementary operations sufficient to exactly reproduce the dynamics of complex physical systems?}
This fundamental question is at the core of most areas of physical sciences, since the state of a physical system encodes all
the information of interest at any point in time. Ideally, given an arbitrary set of (differential) equations that govern the time evolution
of the state of the system, an analytical solution can be found and therefore full control over the dynamics is achieved.
Complex many-body systems,
however, are in general expected to be impervious to analytical approaches with a degree of impermeability that increases with complexity
and size. It follows that specific and often ad-hoc methods need to be developed to successfully tackle the dynamics, at least partially.
This includes the use of numerical methods conjunctly with, or instead of, analytical ones, depending on the
specific case at hand. Therefore, it is only natural to conclude that
a priori knowledge on the ability to completely control or predict the dynamics of a physical system, namely when
can an exact solution be obtained
or when must it be approximated, is ultimately of fundamental importance.

The rise of quantum information theory and the quest for quantum technologies have sparked the development of many approaches to
study and characterize quantum systems and their dynamics. Areas of research
where this endeavor is important range from
quantum computing \cite{Ladd:Jelezko:2010} and quantum simulation
\cite{Buluta:Nori:2009,Bloch:2012ty,Georgescu:Ashhab:2013} to quantum annealing
\cite{Apolloni:Carvalho:1989,Kadowaki:Nishimori:1998,Ohzeki:Nishimori:2011,Hauke:Katzgraber:2020} and quantum control
\cite{JS:1972,Huang:Tarn:1983,Altafini:2002,Dong:Petersen:2010,Chen:Wang:2021,SZCM:2023,Elliott:2009,
Glaser:2015tm,schulte2018,Koch2022,dalessandro:2022},
including foundational studies of the theory itself
\cite{Amaku:Coutinho:Coutinho:2017}. A common
underlying issue can be found in all of these areas: the ability to compute relevant quantities at any time,
such as the eigenstates, energy eigenvalues, or entropy of a reduced subsystem to name a few, dramatically depends
on the ability to obtain the evolution of the state as a function of time. Ideally, an analytical expression for the
state can be obtained, and therefore all pertinent information can be extracted. In practice, however,
the full and exact time evolution of a quantum system is impossible to obtain except for a limited set of cases.
The main reason is that the Hamiltonian, whether
time-dependent or not, usually contains many terms that do not commute, therefore requiring the use of ad-hoc or
approximate methods, such as
perturbative expansions \cite{Trotter:1959,Baye:Goldstein:2003,Goldstein:Baye:2004,Chee:2012}, the
Campbell-Baker-Hausdorff formula
\cite{Poincare:1899,Magnus:1954,Jacobson:1966}, as well as other approaches
\cite{Vinitsky:Gerdt:2006,Agostini:Gross:2021}. Prominent examples of complex Hamiltonians that have
important applications and require specific methodologies are
Bose-Hubbard Hamiltonians \cite{Hubbard:Flowers:1963,Dutta:Gajda:2015}, the full nonlinear phononic Hamiltonian
of a Bose-Einstein Condensate (BEC) \cite{Dalfovo:Giorgini:1999,Combescot:Combescot:2017}, or Hamiltonians
with cross-Kerr interactions \cite{Clausen:Knoell:2002,Kounalakis:Dickel:2018}. 

One solution to the problem of analytically evolving a state in time is given by the process of \emph{factorization}.
Factorization ideally provides an alternative \emph{equivalent} expression for the time-evolution operator that is
composed of a sequence of operations each one simpler than the full
time-evolution operator itself. Once such a factorization sequence
has been obtained, it is then easier to apply each term sequentially to the initial state,
thus dramatically reducing the complexity of the problem. Furthermore, knowledge of which terms appear in the factorization also informs
on the physical interactions that de-facto play a role in the process of interest, as well as the magnitude of their effects. The number of
terms that appear also plays a crucial role. If it is possible to know a priori that a factorization with finite terms is to be expected, this might
encourage the quest for a factorization of the evolution. The problem, therefore, shifts from attempting to employ
the full time-evolution operators to finding
which sequences of operators guarantee a successful factorization, and when are such sequences finite. It is
already well-known that Hamiltonians that are
quadratic in the creation and annihilation operators of the system always admit a finite factorization
\cite{Bruanstein:2005,Adesso:Ragy:2014}. Nevertheless, the vast majority of (bosonic) Hamiltonians of 
interest will contain \emph{nonlinear} terms, where nonlinearity means that they are composed of combinations of operators of
which at least one is at least cubic in the creation and annihilation operators
\cite{Bose:Jacobs:1997,Bruschi:Xuereb:2017,Qvarfort:Serafini:2019,Bruschi:2019,Bruschi:2020,Qvarfort:Serafini:2020}.
Brute-force methods can be employed in such situations,
but obtaining meaningful results is highly case-dependent. An approach that can tackle this problem more systematically
is therefore desirable.

Seminal work has proposed a strategy based on Lie algebras to tackle factorization of quantum dynamics
\cite{Wei:Norman:1963,Wei:Norman:1964} in the case of finite-dimensional systems, which has inspired subsequent
efforts to expand the program 
for particular infinite-dimensional cases such as \cite{Bose:Jacobs:1997,Ibarra-Sierra:Sandoval-Santana:2015}.
The core idea is to use methods from symplectic geometry,
Lie groups, and Lie algebras \cite{Jacobson:1966,hall2015,Bump:2013}
to determine elementary operators that must appear in the factorized expression for the time-evolution
operator. This expression is then complemented with 
a set of coupled, first-order, nonlinear differential equations whose solution provides the time-dependent coefficients
that drive each term in the
factorization \cite{Wei:Norman:1963,Wei:Norman:1964}. The whole approach is straightforward and
therefore in principle very powerful. However, this approach has its challenges that are related to 
infinite-dimensional Lie algebras and unbounded operators
when it is extended to infinite-dimensional
(bosonic) systems \cite{Huang:Tarn:1983,KZSH:2014,KEGGZ:2023,Keyl:2019},

To date, only partial results have been obtained for particular cases of infinite-dimensional systems (i.e., bosons) with
finite-dimensional associated Lie algebras. Prominent examples are: linear
(i.e., quadratic in the creation and annihilation operators) Hamiltonians of $N$ harmonic oscillators, with the finite-dimensional
skew-hermitian Lie algebra $\mathfrak{sp}(2N,\mathbb{R})$ known as the \emph{symplectic algebra} \cite{Adesso:Ragy:2014};
\emph{optomechanics}, namely the field that studies light interacting with matter via radiation-pressure
\cite{Aspelmeyer:Kippenberg:2014}, which enjoy a finite Lie algebra with nonlinear terms
\cite{Bose:Jacobs:1997,Bruschi:Lee:2013,Bruschi:Xuereb:2017,Qvarfort:Serafini:2019,Bruschi:2019,Bruschi:2020,Qvarfort:Serafini:2020}.
Furthermore, Hamiltonians with a finite number of terms that all commute also enjoy finite abelian Lie algebras, and thus a finite factorization.
One option to improve the current state of the art is to construct a classification of Lie algebras from an algebraic perspective that can provide
intuition on how to systematically tackle the problem. Ideally, one would obtain an algorithm that could
take the Hamiltonian as the input and then provide an answer to the fundamental question:
\emph{does the Hamiltonian of interest generate dynamics associated to a finite-dimensional Lie algebra?}
Obtaining such algorithm would therefore dramatically simplify the problems at hand. Regardless of all efforts to
date, the question posed here has remained unanswered.

Bosonic Lie algebras appear as part of the infinite-dimensional Weyl algebra $A_n$
\cite{Dixmier:1968,Dixmier:1977,Bjoerk:1979,Woit:2017,Coutinho:1995,Goodearl:Warfield:2004}
(\emph{vide infra}) which encompasses all complex polynomials in creation and annihilation operators of $n$ modes.
For the one-mode case, all finite-dimensional (complex) Lie subalgebras of the Weyl algebra have been recently
classified \cite{TST:2006} while building on a large body of earlier work
motivated by pioneering work of Dixmier \cite{Dixmier:1968} as well 
as Gelfand and Kirillov \cite{GK:1966}.
Properties of commutators and how
the degree of the commutator changes from its elements have also already been studied for the Weyl algebra 
\cite{Dixmier:1968,Igusa:1981a,Joseph:1974a}. This, in part, lead to conditions for two elements in the Weyl algebra
to generate an infinite Lie algebra \cite{Igusa:1981a} which is closely related to our work.
Finite-dimensional Lie algebras realized as part of the Weyl algebra are
widely studied in the mathematics community \cite{Dixmier:1968,SZ:1969,Joseph:1970,Joseph:1972,Joseph:1974b,Joseph:1974c}.
Moreover, bosonic realizations of Lie algebras are considered even in nuclear physics \cite{KM:1991},
and they are closely related to filter functions in control theory based on independent work of Brockett and Mitter 
\cite{HM:1982}. Motivated by hermitian Hamiltonians, we will restrict ourselves to skew-hermitian
elements of the Weyl algebra which we denote as the \emph{skew-hermitian Weyl algebra} $\hat{A}_n$
in the case of $n$ modes.
Most interestingly, compact semisimple Lie algebras \cite{hall2015,Bump:2013} are mostly ruled out as 
finite-dimensional Lie subalgebras of the skew-hermitian Weyl algebra \cite{Joseph:1972} which can lead to 
interesting noncompact Lie algebras \cite{hall2015,Bump:2013}.

In this work we provide a first partial answer to the question of finite factorization of quantum dynamics.
We approach the problem with the aim of finding conditions on the Hamiltonian for a finite associated Lie algebra.
To achieve our goal we develop
new techniques that allow us divide the vector space of the skew-hermitian Weyl algebra into meaningful subspaces.
We then consider the commutator of any two elements of the Weyl algebra and study its degree in relation to that of the
two initial elements. The degree of a commutator has a finite value depending on the elements that are are being commuted
and their degrees.
We find a fundamental properties the two elements in the commutator that, when satisfied,
guarantees that their commutator has the maximally possible degree. This, in turn, is used to construct sequences of nested
commutators called \emph{commutator chains} obtained starting from two initially given \emph{seed} elements. 
We show which conditions must be satisfied by the seed elements in order for each element
in the chain to be nonvanishing with an monotonically increasing degree. 
Finally, we use these tools to obtain conditions that need to be satisfied by a set of generators in order for the
generated Lie algebra to be finite. 

The present work relies on a key assumption: all operators that
appear in the absence of interaction for each harmonic oscillator (i.e., quantized bosonic degrees of freedom), are
themselves independently controlled generators of the Lie algebra. Such a situation occurs physically when the frequencies of each
harmonic oscillator present in the system can be varied independently as a function time. In the language of
quantum control \cite{Huang:Tarn:1983,Wu:Tarn:2006,Bagarello:2007,Dong:Petersen:2010,Bliss:Burgarth:2014,KZSH:2014,ABFH:2018,KEGGZ:2023}, we
can equivalently say that we are studying the case of Hamiltonians without \textit{drift} since each term in the Hamiltonian is driven by
a time-dependent coefficient that can be experimentally modulated, or tuned.
Examples of systems or physical processes with
these properties, i.e., Hamiltonians without drift, include spontaneous parametric down-conversion modelled as a nonlinear process on
optical fields \cite{Couteau:2018}, quenches of electromagnetic traps \cite{Rajabpour:Sotiriadis:2014},
modulated optomechanics \cite{Qvarfort:Serafini:2020,Qvarfort:Plato:2021,Partanen:Tulkki:2022}, or networks of
mechanical oscillators \cite{Xuereb:Imparato:2015} to name a few.
The complementary case occurs with a drift term that can be formed, for example, by linear combinations of 
operators that appear in the absence of interaction for each harmonic oscillator.
This case is not settled in the present work.
Our main result provides a first step into the categorization of Lie algebras in terms of their dimensionality, with a
specific focus on providing a simple algorithm for verifying
the finiteness of the dimension with generators of the whole algebra as sole input.

This work is organized as follows. In Section~\ref{work:scope:appendix} we discuss the scope 
and motivation of the present work.
Section~\ref{algebras:section} recalls the Weyl algebra and introduces the skew-hermitian Weyl algebra
consisting of skew-hermitian elements in the Weyl algebra. In
Section~\ref{algebra:properties:section} key properties of the elements of Weyl algebra and skew-hermitian
Weyl algebra are described. We introduce the concept of a commutator chain and
develop its basic properties for two important cases in Section~\ref{commutator:chain:section}.
In Section~\ref{main:results:section} we present our main result. Finally, before concluding,
we discuss our results in
Section~\ref{discussion:outlook:section} and outline future directions.

\section{Scope of the work}\label{work:scope:appendix}

This work is motivated by the question of how to factor the 
dynamics given by Hamiltonians with tunable interactions in the context of coupled bosonic quantum systems. 
Bosonic systems are infinite-dimensional systems \cite{Messiah:1999}, and the Lie algebras generated by a 
bosonic Hamiltonian have in general also  an infinite dimension. Paramount bosonic systems are quantum harmonic
oscillators, e.g., cavity field modes \cite{Walther:Varcoe:2006} or the modes of
free quantum fields \cite{Srednicki:2007}. Recall that in the case of fermions, such as two-level systems
or other $d$-dimensional systems, the Lie algebras have always a finite dimension as we have the finite-dimensional
Hilbert space. The question of factoring quantum dynamics in this framework is
quite involved as extensively discussed in the literature, in particular from the point of view of quantum control
\cite{Huang:Tarn:1983,Wu:Tarn:2006,Bagarello:2007,Bliss:Burgarth:2014,KZSH:2014,ABFH:2018}.
To set up and motivate our work, we start by shortly recalling  the well-understood
finite-dimensional unitary case.

\subsection{Factorization in the finite dimensions\label{sub:finite}}

In finite dimensions, the concept of factoring a unitary operator is quite straightforward \cite{Messiah:1999}.
Let $U(t)$ be a one-parameter family of unitary operators such that $\frac{d}{dt}U(t)= -i\,H(t)\,U(t)$,
where $H(t)=\sum_{j\in\mathcal{J}} A_j(t)$ is the \textit{Hamiltonian} and the $A_j(t)$ form a set of hermitian operators.
Here, an operator $A$ is hermitian if $A^\dag=A$, it is skew-hermitian if $A^\dag=-A$,
$U(t)$ denotes the time-evolution operator, $H(t)$ is a generally time-dependent
Hamiltonian, and $\mathcal{J}$ is an appropriate finite subset of
$\mathbb{N}$. Henceforth we assume $\hbar=1$ for notational convenience.
We can write
\begin{equation}\label{time:evolution:operator}
U(t)=\overset{\leftarrow}{\mathcal{T}}e^{-i\int_0^t dt' H(t')}
\end{equation}
as the solution of the differential equation, where $\overset{\leftarrow}{\mathcal{T}}$ is known as the time-ordering operator.
We consider a decomposition $-iH(t)=\sum_{j\in\mathcal{J}} u_j(t) G_j$ into skew-hermitian
generators $G_j$ that do not depend on the time $t$
and $u_j(t)$ are tunable, time-dependent, and
real \textit{control pulses} \cite{Elliott:2009,dalessandro:2022}. We then collect the generators
$G_j$ into the finite set $\mathcal{G}:=\{G_j : j\in\mathcal{J}\}$.
A finite \textit{factorization} is then given by
\begin{equation}\label{operator:factorization}
U(t)=\prod_{k\in\mathcal{K}} e^{f_k(t) G_{j_k}} \;\text{ for certain }\; j_k\in \mathcal{J} \;\text{ and real values }\;  f_k(t),
\end{equation}
where $\mathcal{K}$ is an appropriate finite subset of $\mathbb{N}$. Note that, in general,
$\mathcal{K}\neq\mathcal{J}$ and it is often the case that $|\mathcal{J}|<|\mathcal{K}|$, i.e., there are (many)
more linearly-independent terms in the factorization \eqref{operator:factorization} than in the original Hamiltonian $H(t)$.
An example is the single-qubit Hamiltonian $H=\omega\sigma_{z}+g\sigma_{x}$, where $\omega$ is the
energy, $g$ is a coupling constant, and $\sigma_j$ are the standard Pauli matrices that enjoy the commutators
$[\sigma_j,\sigma_k]=i\epsilon_{jkl}\sigma_l$, see  \cite{Messiah:1999}. Here, $\epsilon_{jkl}$ is the totally
antisymmetric Levi-Civita symbol and we use the Einstein summation convention on repeated indices. We can define $G_0:=i\sigma_{z}$ and $G_1:=i\sigma_{x}$ and it is
well known that $\{G_0,G_1\}$ generate the Lie algebra $\mathfrak{su}(2)$, which is $3$-dimensional.

In the present case, the Lie algebra $\g=\langle \mathcal{G}\rangle_{\text{Lie}}$ associated to the Hamiltonian of
finite-dimensional systems has always a finite dimension and it can therefore be used to obtain the
factorization \eqref{operator:factorization} explicitly \cite{Wei:Norman:1963,Wei:Norman:1964}. We wish to
emphasize the key observation here that there is a direct connection between the factorization \eqref{operator:factorization}
and the Lie algebra associated to the Hamiltonian \cite{Wei:Norman:1963,Wei:Norman:1964,Wu:Tarn:2006}.

\subsection{Skew-hermitian bosonic operators and constraints from physics}
In this work we are interested in finite factorizations of $U(t)$ as introduced for the finite-dimensional
case \cite{Wei:Norman:1963,Wei:Norman:1964},
while going beyond such finite-dimensional setting. Before discussing the infinite-dimensional case, let us
make a few remarks below.

We consider an ordered set $\mathcal{G}:=\{G_j : 0\leq j \leq p\}$ of skew-hermitian operators $G_j$ that are built from bosonic
creation and annihilation operators $a_k^{\dagger}$ and $a_k$ (which we discuss more in detail in
Section~\ref{subsection:weyl} and Section~\ref{subsection:skew:weyl}). Here $\mathcal{J}:=\{1,\ldots,p\}$.
The time-dependent Hamiltonian is then given by
\begin{equation*}
H(t) = H = H_0 + H_{\text{I}}(t) \;\text{ with }\; iH_0 = G_{0} \;\text{ and }\;
iH_{\text{I}}(t)=iH_{\text{I}}=\sum_{k=1}^{p} u_k(t) G_{k}
\end{equation*}
for $G_k \in \mathcal{G}$ and real \textit{control functions} $u_k(t)$.
Henceforth, $H_0$ will denote the so-called \emph{drift} Hamiltonian which
collects the fixed or not controllable terms,
while $H_{\text{I}}(t)$ collects the terms that can be externally tuned
in a time-dependent fashion
and denotes the \emph{interaction} (or \emph{control}) part
of the Hamiltonian \cite{Elliott:2009,dalessandro:2022,Wu:Tarn:2006,Dong:Petersen:2010,Chen:Wang:2021}. 
Note that the Hamiltonian is hermitian (i.e., $H^\dag(t)=H(t)$), hence the factor of $i$ in the expression above. 
The skew-hermitian $G_j$ generate a Lie algebra $\g$ of possibly infinite dimension.
We here give particular emphasis to the following notation:

\begin{definition*}
The linear combination $\sum_k \omega_k a_k^\dag a_k$ is called the \textit{free Hamiltonian}.
The real coefficients $\omega_k$, also known in the literature as the \textit{frequencies}, can be time dependent. 
\end{definition*}
It should be clear that,
when the coefficients $\omega_k$ are time independent, then the free Hamiltonian constitutes part of the drift.
It can very well be the case that the free Hamiltonian \textit{coincides} with the drift. This occurs when the
frequencies $\omega_k$ are the only time-independent coefficients in the Hamiltonian.

To analyze a concrete example for the drift, one can consider the case of $H_0=\sum_k \omega_k a_k^\dag a_k$
for real and constant \textit{frequencies} $\omega_k$ of the harmonic oscillators.
Therefore, the drift Hamiltonian $H_0$ is composed of linearly independent operators of the form $a^\dag_k a_k$ that appear in $H_0$ 
as a fixed linear combination, and it coincides with the free Hamiltonian.
Thus, if $iH_0$ is an element of a specific Lie algebra $\g$, 
the individual $ia_k^\dag a_k$ will not necessarily be contained in $\g$.
In the case without drift, \text{all} coefficients of the generators $G_k$ are time-dependent, tunable controls $u_k(t)$ and
therefore the different operators $a^\dag_k a_k$ can be combined to show that each
one is also an element of the generated Lie algebra. An example would be if at 
$t=t_0$ one has $H_{\text{a}}=\omega_1 a^\dag_1 a_1 +\omega_2 a^\dag_2 a _2$,
while at $t=t_1$ one has $H_{\text{b}}=\omega_1 a^\dag_1 a_1 +5\omega_2 a^\dag_2 a _2$. It is clear that one can obtain
$a^\dag_1 a_1$ and $a^\dag_2 a_2$ as $a^\dag_1 a_1=\frac{1}{2}(H_{\text{a}}{+}H_{\text{b}})$
and $a^\dag_2 a_2=\frac{1}{4}(H_{\text{a}}{-}H_{\text{b}})$, which implies that if
$iH_{\text{a}},iH_{\text{b}}\in\g$ then also $ia^\dag_1 a_1,ia^\dag_2 a_2\in\g$. 
The following examples highlight that coefficients 
in the Hamiltonian either can be fixed and, in particular, constant
or they can be externally tuned in a time-dependent fashion,
which also means that the same operators will be either part of the drift or contained in the control part,
and this dramatically changes the outlook on the possibility of factorizing the quantum dynamics. 

\begin{tcolorbox}[colback=orange!3!white,colframe=orange!85!black,title=Example: Hamiltonians with and without drift]
\label{example:one}
Two examples to differentiate between Hamiltonians with different drifts.

\begin{example}[Hamiltonian with drift]
The first Hamiltonian that we consider is $H=\omega_1 a^\dag_1 a_1 +\omega_2 a^\dag_2 a _2+\omega_3 a^\dag_3 
a _3+g(t)(a^\dag_1a^\dag_2 a_3+a_1 a_2 a^\dag_3)$, where $\omega_1$ and $\omega_2$
are constant and $\omega_3\neq\omega_1+\omega_2$, while $g(t)$ is time dependent
and externally tunable.
In this case, the drift is $H_0=\omega_1 a^\dag_1 a_1 +\omega_2 a^\dag_2 a _2+\omega_3 a^\dag_3 a _3$
while the interaction Hamiltonian is $H_\text{I}=g(a^\dag_1a^\dag_2 a_3+a_1 a_2 a^\dag_3)$. 
It is immediate to verify that: (i) $[H_0,H_\text{I}]=0$ when $\omega_3=\omega_1+\omega_2$; (ii) $[H_0,H_\text{I}]\neq0$ otherwise.
\end{example}

\begin{example}[Hamiltonian without drift]
The second Hamiltonian is $H(t)=\omega_1(t) a^\dag_1 a_1 +\omega_2(t) a^\dag_2 a _2+\omega_3(t) a^\dag_3
a _3+g(a^\dag_1a^\dag_2 a_3+a_1 a_2 a^\dag_3)$, where the frequencies $\omega_k(t)$ are time dependent
and externally tunable
and there is no drift term.
Let $H_{\text{I}}=\omega_1(t) a^\dag_1 a_1 +\omega_2(t) a^\dag_2 a _2+\omega_3(t) a^\dag_3
a _3$ while $H_{\text{J}}=g(a^\dag_1a^\dag_2 a_3+a_1 a_2 a^\dag_3)$.
It is immediate to verify that $[H_{\text{I}},H_{\text{J}}]=0$ only for times $t_*$ for which
$\omega_3(t_*)=\omega_1(t_*)+\omega_2(t_*)$.
\end{example}
\end{tcolorbox}

We reiterate here that our formal results are applicable to cases without drift, which concretely leads us to assume in our
main result that each $ia^\dag_k a_k$ is an element
of the generated Lie algebra. Extensions to cases with drift, and other relevant considerations, are presented and discussed at the end of this work.

\subsection{Factorization in the infinite-dimensional case}
We now move to the case of  infinite-dimensional systems. Representative examples are quantum
harmonic oscillators, which are characterized by creation and annihilation operators $a_j^{\dagger}$ and
$a_j$ that observe the canonical commutation relations $[a_i, a_j^{\dagger}] = \delta_{ij}$ while
all others vanish \cite{Messiah:1999}.

It is well known that infinite-dimensional systems pose additional challenges compared to their finite-dimensional
counterparts. For our purposes of illustration, we can recall one major issue, which is the presence of
\textit{unbounded operators} \cite{Messiah:1999,meise1997,hall2013quantum,schmuedgen2012}.
Given a Hilbert space $\mathcal{H}$, we say that $A$ is a
bounded operator if and only if the exists $M\geq0$ such that $||A v||\leq M||v||$ for all $v\in\mathcal{H}$. In the
case of finite-dimensional systems all operators are bounded. However, the situation
changes dramatically when moving to infinite-dimensional Hilbert spaces. In this case, one immediately obtains
unbounded operators \cite{Huang:Tarn:1983,Lan:2005,KZSH:2014,Keyl:2019,ABFH:2018,KEGGZ:2023}, such as the annihilation and
creation operators $a,a^\dag$ for a harmonic oscillator \cite{Bagarello:2007,Bloch:Brockett:2010,Bliss:Burgarth:2014},
or its free Hamiltonian $\omega a^\dag a$. Let us discuss a brief example here:

\begin{tcolorbox}[colback=orange!3!white,colframe=orange!85!black,title=Example: Unbounded operators ]
\label{example:unbounded}
\begin{example}
We consider the Hilbert space $\mathcal{H}$  of a single quantum harmonic oscillator with annihilation and
creation operators $a,a^\dag$ that satisfy $[a,a^\dag]=1$. A convenient orthonormal and complete basis
$\{|n\rangle\}$ is defined by $a^\dag a|n\rangle=n|n\rangle$ for all $n\geq0$. This allows us to write the
free Hamiltonian as $\omega a^\dag a=\omega \sum_{n=0}^\infty n |n\rangle\langle n|$. 
Let us choose the state $|\psi\rangle=\sum_{n=1}^\infty n^{-3/2}|n\rangle\in\mathcal{H}$, which is properly
normalized to $\langle\psi|\psi\rangle=\zeta(3)<\infty$ where $\zeta(m)$ is the Riemann zeta-function.
However, $|\tilde{\psi}\rangle:=\omega a^\dag a|\psi\rangle=\omega\sum_{n=1}^\infty n^{-1/2}|n\rangle$ and thus
$||\omega a^\dag a|\psi\rangle||\equiv \langle\tilde{\psi}|\tilde{\psi}\rangle=\omega^2\sum_{n=1}^\infty 1/n$, which diverges.
\end{example}
\end{tcolorbox}

For completeness we note that the specific problem outlined here is not unsurmountable: 
One can start by restricting the domain of interest from the full Hilbert space to an appropriate dense
subset $\mathcal{D}(\mathcal{H})\subseteq\mathcal{H}$. This can be, for example, combined 
with analyticity assumptions and assuming a finite-dimensional Lie algebra
\cite{Huang:Tarn:1983,Lan:2005,Wu:Tarn:2006}
which however leads only to a restricted  and finite controllability.
One complementary approach, for example, relies on an abelian symmetry to treat all but one of the 
infinitesimal generators \cite{KZSH:2014}, while a relatively simple argument using the remaining symmetry-braking 
infinitesimal generator leads to a general class of controllability.

The issue mention above, as well as others, prevent us from using the important Lie-algebraic results obtained in
the case of finite-dimensional systems  \cite{Wei:Norman:1963,Wei:Norman:1964}, and therefore to export them 
directly to the infinite-dimensional case that is of interest to us. Nevertheless, it would still be of great significance
if it were possible to find a way to extend these ideas beyond the finite-dimensional framework.
Therefore, regardless of the concrete (yet important) differences between the finite- and infinite-dimensional systems,
we conclude this brief part by making a few crucial comments that contribute to motivate our present efforts. 

First, there is a vast literature in the context of quantum control that tackles very related questions that are directly
relevant \cite{Brockett:1972,JS:1972,Hirschorn:1975,Kunita:1979}. Important results have been proven for finite-dimensional
systems, where conditions for controllability (read finite-factorization in our case) in the case of no drift present have
been explicitly give (see Theorem 7.1 and 7.2 in \cite{JS:1972}).
Second, the role of the drift is clear for unitary dynamics (see Theorem 7.1 in \cite{JS:1972})
and can be also resolved in certain finite-dimensional cases beyond that (see Theorem 7.3 in \cite{JS:1972}).
It has already been shown that the presence of the drift can have
significant implications for the controllability of the dynamics \cite{Wu:Tarn:2006}.
Our work mostly does not include this case, which remains extant.
Finally, it is not necessary at present to have a direct link between Lie algebras and the ability to factorize the
time-evolution operator of bosonic Hamiltonians in order to study the Lie algebras associated to the dynamics.
In this sense, we can remove ourselves form the problem of determining which Hamiltonians can be factorized
and the set of states in the Hilbert space for which this is a well-defined operation, and we can just focus on
studying Lie algebras that are derived from the Hamiltonians of coupled quantum harmonic oscillators. It will be
part of future efforts to complete the link between the two. This is the spirit with which we approach the problem at hand
as outlined in the next section.

\subsection{Aim}
This work tackles a crucial property of skew-hermitian Lie algebras $\g$ that has implications both for
mathematics and physics. Namely, we are interested in the determining the finiteness of a Lie algebra
focusing only on the properties of its generators, and their commutators. In this work we restrict ourselves
to the set of Lie algebras $\mathfrak{g}$ that naturally arise in physics when considering the dynamics of
interacting bosonic systems.

From the perspective of understanding the properties of Lie algebras, our work aims at answering the
following question of fundamental importance:

\textbf{Q1}: \emph{Which skew-hermitian Lie algebras are finite-dimensional? In particular which conditions need to be satisfied by the
set of generators $\mathcal{G}$ in order to verify if the generated algebra $\g=\langle\mathcal{G}\rangle_{\text{Lie}}$ is finite?}

The question posed above, which is given in a purely mathematical context, can be rephrased from the point of view of physics. In
physics, the study of the dynamics of a physical system, whether classical or quantum, is the main task to be undertaken in order to obtain
information about a physical process at all times. In general, we are interested in exploiting the full extent of quantum mechanics, which is
based on noncommutative operators. The dynamics of any quantum system, given by a unitary
operator $U(t)$ with Hamiltonian $H$, is therefore
determined by the commutation properties of the operators within the Hamiltonian under consideration. In this context, it is natural to ask
the following question:

\textbf{Q2}: \emph{Which conditions need to be satisfied by the operators that determine
the Hamiltonian to guarantee that the time evolution of the system can be obtained through a finite set of controllable operations?}

The concrete ambition, therefore, is to provide a set of conditions that uniquely determine the classes of Hamiltonians
that have finite Lie algebras. In practice, given the Hamiltonian $H(t)=\sum_{j\in\mathcal{J}} u_j(t) G_j$ of the system,
with $\mathcal{G}=\{G_j\}$ for the finite index set $\mathcal{J}$, one should be required to perform a minimal number
of operations on the generators $G_j\in\mathcal{G}$ in order to conclude that the full Hamiltonian algebra
$\g=\langle\mathcal{G}\rangle_{\text{Lie}}$ is finite-dimensional. If these criteria can also be used in
this context to show that the dynamics does admit a finite factorization, this then would have a direct impact
on the ability to fully control the evolution of the quantum state, or to design algorithms for simulating complex
physical processes. We anticipate that we answer the questions posed above in the case of Hamiltonians without
drift. A pictorial representation of the discriminating process is given in Figure~\ref{figure:factorization:process}

\begin{figure}[t]
\includegraphics[width=0.9\linewidth]{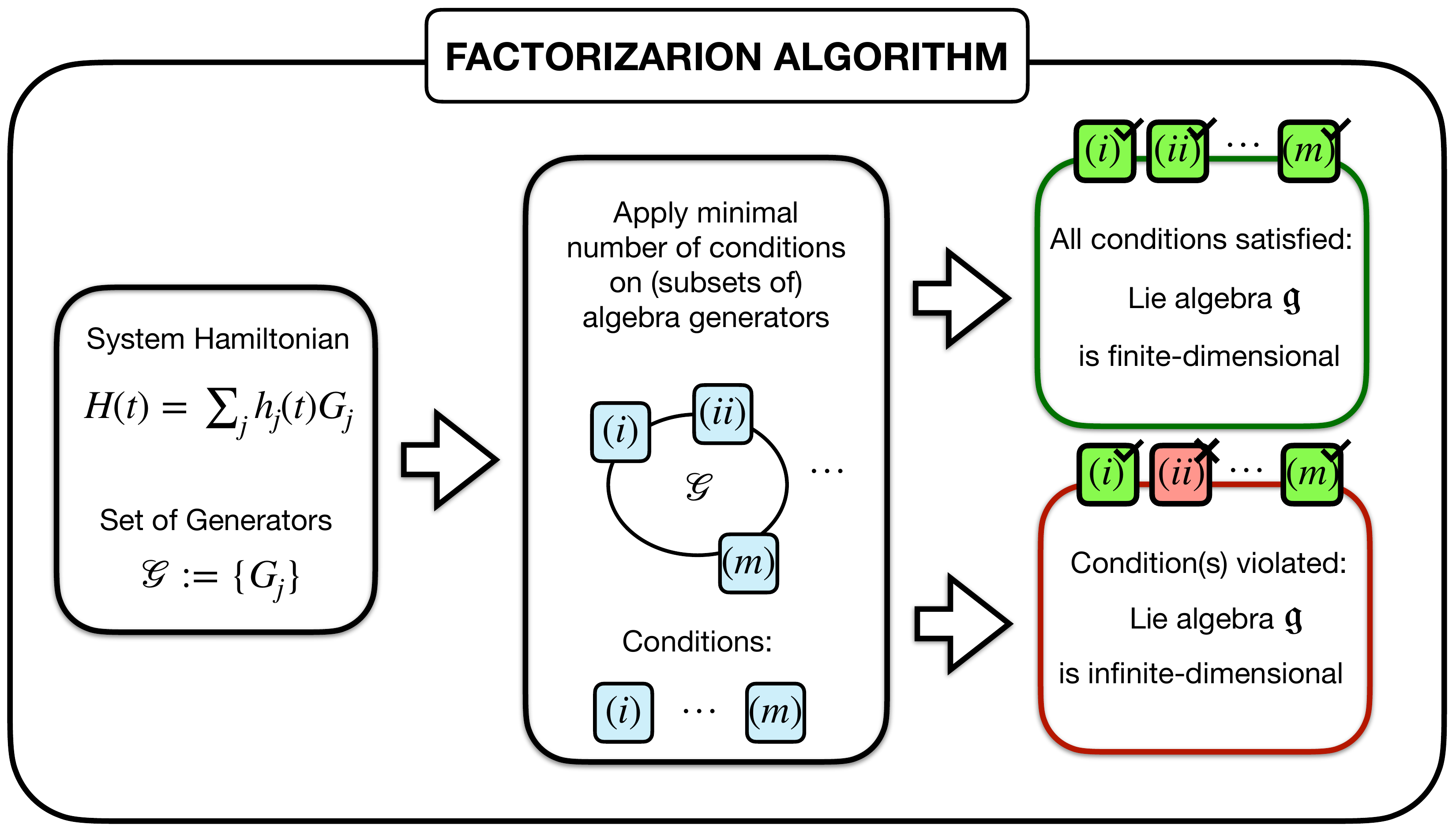}
\caption{\textbf{Factorization algorithm}: A pictorial representation of the general idea behind the approach to verifying
whether a given Hamiltonian $H(t)$ will lead to a finite factorization of the quantum dynamics. First, the individual components
of the Hamiltonian must be identified. These components will form the set $\mathcal{G}$ of \emph{generators} of the
Lie algebra $\g$.  Then, a set of conditions is applied to the generators $\mathcal{G}$, and to a subset of
their commutators $[\mathcal{G},\mathcal{G}]$. If all conditions are verified, the Lie algebra $\g$ has a finite dimension.
Violation of any condition guarantees an infinite-dimensional Lie algebra $\g$.}\label{figure:factorization:process}
\end{figure}

This parallel vision underlies our study of key properties of skew-hermitian Lie algebras. Our approach concretely
differentiates between the two questions, de facto making the answers to the first as prop{\ae}deutic tools to
obtaining the answers of the second.
In particular:
\begin{enumerate}
	\item[I.] We provide a useful decomposition and classification of the elements of the skew-hermitian
	Weyl algebra. This classification,
	which can seem arbitrary at first, divides the full vector space of operators into five complementary subvector spaces. We
	then show that there exist at least two different types of sequences of operators, called \emph{commutator chains}, 
	that are generated by two \textit{seed} elements by means of nested commutators. This uniquely associates (infinite)
	sequences of elements of the algebra to pairs of appropriately chosen seed elements.
	\item[II.] We show that, when the Hamiltonian contains no drift and it is determined by its independent
	\emph{generators} as discussed below,
	all generators of the Lie algebra except for those of the form $a^\dag_k a_k$ must either: (a) satisfy specific conditions for the
	algebra to be finite, or; (b) they can be used to construct infinite commutator chains of linearly independent operators that
	all belong to the Lie algebra. This latter case implies an infinite-dimensional Lie algebra.
\end{enumerate}
We can see from this approach that we \emph{first} develop tools to classify elements of the skew-hermitian
Weyl algebra and determine their properties, and only \emph{subsequently} apply this classification to Hamiltonians without drift.

\section{Introduction to the Weyl algebra and the skew-hermitian Weyl algebra}\label{algebras:section}
We describe now the general setting while recalling the the Weyl algebra (see Subsection~\ref{subsection:weyl}) and
introducing its skew-hermitian part as the skew-hermitian Weyl algebra (see Subsection~\ref{subsection:skew:weyl}).
Subsection~\ref{sec:decomp:weyl} describes the general structure of the skew-hermitian Weyl algebra
by highlighting subspaces which will be important in characterizing certain finite-dimensional Lie subalgebras
of the skew-hermitian Weyl algebra.

\subsection{The Weyl algebra\label{subsection:weyl}}
We consider an $n$-mode bosonic quantum system with creation and annihilation operators $a_j^{\dagger}$ 
and $a_j$ that observe the canonical commutation relations $[a_i, a_j^{\dagger}] = \delta_{ij}$ while
all others vanish.
These operators generate an infinite associative algebra $A_n$ over $\C$ with an identity element $1$, 
and $A_n$ is known as $n$th Weyl algebra
\cite{Dixmier:1968,Dixmier:1977,Bjoerk:1979,Woit:2017,Coutinho:1995,Goodearl:Warfield:2004}.
In the following, we use $\alpha,\beta \in \Np{n}$ and $\gamma = (\alpha,\beta) \in \Np{2n}$ as
multi-indices, where $\Np{n}$ denotes the set of nonnegative integer vectors of length $n$. We also introduce, for later convenience,
the vector $\ones_p \in \Np{2n}$ which has only two nonzero entries at the positions $p$ and $n+p$, where both entries are equal to $1$. 
The notation $a^{\gamma} = a^{(\alpha,\beta)}$ denotes the normal-ordered monomial
$a^{\gamma}=(a_1^{\dagger})^{\alpha_1} \cdots (a_n^{\dagger})^{\alpha_n} (a_1)^{\beta_1} \cdots (a_n)^{\beta_n}$,
where the \emph{normal order} is obtained in a product of creation and annihilation operators if all creation operators are to the left
of all annihilation operators \cite{Blasiak:2007}.

For a general (but possibly not normal-ordered) monomial
$m$ of creation and annihilation operators, its multi-degree $\mdeg(m) := \gamma = (\alpha,\beta) \in \Np{2n}$ 
specifies the frequencies $\alpha_j$ and $\beta_j$ of $a_j^{\dagger}$ and  $a_j$ in $m$.
In particular, $\mdeg(a^{\gamma})=\gamma$ and an element of this vector is denoted by $\gamma_j$.
The degree of a monomial $m$ is given by $\abs{\mdeg(m)}=|\gamma|$ where $\abs{\gamma}:=\sum_{j=1}^{2n} \gamma_j$.
The normal-ordered monomials $a^{\gamma}$ form a complex vector-space basis of the Weyl algebra $A_n$, see \cite[p.~9]{Coutinho:1995}.

Any finite element $g\in A_n$ is called a (non-commutative) polynomial and it
can be uniquely expanded into its normal-ordered (or canonical) form
$g=\no{g}:=\sum_{\gamma} c_{\gamma} a^{\gamma}$ with $c_{\gamma} \in \C$, while $\no{0}=0$. 
The degree $\deg(g)$ of any nonzero 
polynomial $g \in A_n$ is defined as the largest value of $\abs{\gamma}$ for any normal-ordered monomial 
$a^\gamma$ with $c_\gamma\neq 0$ appearing in 
$\no{g}=\sum_{\gamma} c_{\gamma} a^{\gamma}$, see \cite[p.~14]{Coutinho:1995}. We set $\deg(0):=-\infty$ for convention as also done in the literature. 
For any monomial $m$, the notions  $\abs{\mdeg(m)}$ and $\deg(m)$ are equal which follows from 
the canonical commutation relations. 

Moreover we introduce the linear anti-automorphism 
$(\cdot)^\dagger: g \mapsto g^{\dagger}$ for
polynomials $g\in A_n$ by setting $1^{\dagger}:=1$, $(a_j^{\dagger})^\dagger:=a_j$,  $(a_j)^{\dagger}:=a_j^{\dagger}$,
and $(m_1 m_2)^\dagger := (m_2)^\dagger (m_1)^\dagger $ for any monomials $m_1,m_2 \in A_n$. In particular,
$(a^{(\alpha,\beta)})^\dagger= a^{(\beta,\alpha)}$. Let us also define 
$\gamma^{\dagger}=(\alpha,\beta)^{\dagger}:=(\beta,\alpha)$.

\subsection{The skew-hermitian Weyl algebra\label{subsection:skew:weyl}}

Let us introduce the skew-hermitian Weyl algebra $\hat{A}_n$
as the real subalgebra of the Weyl algebra $A_n$ that consists of
all polynomials $g\in A_n$ such that $g^\dagger = -g$. We use the notation $\gamma = (\alpha,\beta)$.
A basis of $\hat{A}_n$ is given by
all \emph{basis elements} of the form $g_+^{\gamma} := g_+(a^{\gamma})= i ( (a^\gamma)^\dagger {+} a^\gamma)=
i (a^{(\beta,\alpha)} {+} a^{(\alpha,\beta)})$ with $\alpha \ge \beta$
[i.e., $(\alpha,\beta) \ge (\beta,\alpha) = (\alpha,\beta)^{\dagger}$]
and $g_-^{\gamma} := g_-(a^{\gamma}) = (a^\gamma)^\dagger - a^\gamma = a^{(\beta,\alpha)} - a^{(\alpha,\beta)}$
with $\alpha > \beta$,
where the notation $a^{\gamma}=a^{(\alpha,\beta)}$ has been defined in Section~\ref{subsection:weyl} and the relation $>$ 
is the usual order on vectors (where $\alpha > \beta$ if $\alpha_k > \beta_k$
with $k$ being the smallest index $j$ such that $\alpha_j \neq \beta_j$).
More generally, we can introduce the elements $g_+(m) := i ( m^\dagger {+} m)$ and $g_-(m) := m^\dagger {-} m$
for any monomial $m \in A_n$.  Moreover, any polynomial $g \in \hat{A}_n$
can be expanded as $g = \sum_\gamma c_{\gamma,\sigma} g_\sigma^\gamma$ with $\sigma \in \{+, - \}$ and
$c_{\gamma,\sigma}\in \R$. As stated above, we generally assume that 
$\alpha \ge \beta$ for $g_{+}^\gamma$ and $\alpha > \beta$ for $g_{-}^\gamma$ in order to ensure
that no $g_\sigma^\gamma$ is zero or equal up to a sign to a different element 
$g_\sigma^{\hat{\gamma}}$ with $\hat{\gamma} \neq \gamma$.
For example, note that $g_+(a_k) = g_+(a_k^\dagger)$.
The skew-hermitian Weyl algebra $\hat{A}_n$ is closed under commutators
as $([g_1, g_2])^\dagger = - [g_1, g_2]$ and $[g_1, g_2] \in \hat{A}_n$
for any polynomials $g_1, g_2 \in \hat{A}_n$.
We denote the operators $g_+^{(\alpha,\alpha)}=2i a^{(\alpha,\alpha)}$ 
as diagonal operators; $g_-^{(\alpha,\alpha)}=0$.
Using the notation $\ones_p$, we obtain in particular the free operators
$g_+^{\ones_p}=2i a^{\ones_p} = 2i a_p^\dagger a_p \in \hat{A}_n$.
Moreover, we find it convenient to define the scalar product $\inner{\,\cdot\,}{\,\cdot\,}$ through the expression 
$\inner{g_{\sigma}^{(\alpha,\beta)}}{g_{\tilde{\sigma}}^{(\tilde{\alpha},\tilde{\beta})}} \neq 0$ if 
$\sigma = \tilde{\sigma} \in \{+,-\}$, $\alpha = \tilde{\alpha}$, and $\beta = \tilde{\beta}$, while
$\inner{g_{\sigma}^{(\alpha,\beta)}}{g_{\tilde{\sigma}}^{(\tilde{\alpha},\tilde{\beta})}} = 0$ otherwise.

\subsection{Decomposition of the skew-hermitian Weyl algebra\label{sec:decomp:weyl}}

We discuss now the decomposition into meaningful subspaces of the
skew-hermitian Weyl algebra $\hat{A}_n$. The importance of these subspaces will become evident once the main claim and its proof are presented. 

\begin{definition}\label{definition:vector:spaces:A}
The vector spaces $\hat{A}_n^K$ for $K\in \{0,1,2,{=},\text{om}\}$ are defined as
\begin{align*}
\hat{A}_n^{0} &:=\vspan{2i \text{ and } 2ia_{k}^{\dagger} a_k \text{ for } 1\leq k \leq n}
= \vspan{g_{+}^{\tau_k}, g_{-}^{\tau_k}
\text{ such that } \alpha = \beta
\text{ and } \abs{\gamma}=|\alpha|+|\beta| \in \{0,2\}},\\
\hat{A}_n^{1} &:=\vspan{g_{+}(a^{(\alpha,\beta)}), g_{-}(a^{(\alpha,\beta)})
\text{ such that } \abs{\gamma}=|\alpha|+|\beta|=1},\\
\hat{A}_n^{2} &:= \vspan{g_{+}(a^{(\alpha,\beta)}), g_{-}(a^{(\alpha,\beta)}) \text{ such that }
\abs{\gamma}=|\alpha|+|\beta|=2 \text{ and } \alpha_k \neq \beta_k \text{ for one index $k$ or for two indices } k,k'},\\
\hat{A}_n^{=} &:=\vspan{g_{+}(a^{(\alpha,\beta)}), g_{-}(a^{(\alpha,\beta)})
\text{ such that } \alpha = \beta
\text{ and } \abs{\gamma}=\abs{\alpha}+\abs{\beta} \geq 4},\\
\hat{A}_n^{\text{om}} &:=\gspan\{g_{+}(a^{(\alpha,\beta)}), g_{-}(a^{(\alpha,\beta)}) \text{ such that }
\abs{\gamma}=|\alpha|+|\beta| \geq 3,\,
\alpha_k + \beta_k = 1 \text{ for a unique index $k$}, \alpha_j = \beta_j \text{ for } j\neq k\}.
\end{align*}
\end{definition}
This enables us to decompose $\hat{A}_n=\hat{A}_n^{0} \oplus \hat{A}_n^{1} \oplus \hat{A}_n^{2} \oplus \hat{A}_n^{=} \oplus
\hat{A}_n^{\text{om}}  \oplus \hat{A}_n^{\perp}$ into a direct sum of vector spaces. Here,
$\hat{A}_n^{\perp}$ denotes the vector spaces complement of
$\hat{A}_n^{\text{core}} := \hat{A}_n^{0} \oplus \hat{A}_n^{1} \oplus \hat{A}_n^{2} \oplus \hat{A}_n^{=} \oplus
\hat{A}_n^{\text{om}}$ with respect to the scalar product $\inner{\,\cdot\,}{\,\cdot\,}$ introduced above. We note here that the label \textit{om} stands for \textit{optomechanics}, and the reason becomes clearer when perusing the discussion in Section~\ref{discussion:outlook:section} after the main result. We say that an element $v \in \hat{A}_n$
(and similarly for any element of a Hilbert space)
has nonzero support on a vector subspace $W \subseteq \hat{A}_n$ if there exists an element $w\in W$ such that $\inner{v}{w} \neq 0$.
We extend these definitions to Lie subalgebras $\g$ of $\hat{A}_n$ by defining
$\g_0 := \PP{\g}{\hat{A}_n^{0}}$, $\g_1 := \PP{\g}{\hat{A}_n^{1}}$, $\g_2 := \PP{\g}{\hat{A}_n^{2}}$,
$\g_= := \PP{\g}{\hat{A}_n^{=}}$
$\g_{\text{om}} := \PP{\g}{\hat{A}_n^{\text{om}}}$
and  where $\PP{a}{V}$ is the projection of $a$ onto the vector space $V$
(we will use this notation in particular in Theorem~\ref{main:theorem_final}(iii)).
With respect to these projections, we obtain a
vector-space decomposition
$\g= \g_{\text{core}} \oplus \g_{\perp}$ where $\g_{\perp}$ is the vector-space complement
of $\g_{\text{core}}:=\g_{0} \oplus \g_{1} \oplus \g_{2} \oplus \g_{=} \oplus
\g_{\text{om}} $ with respect to the introduced scalar product.
We denote $\gc$ as the center of $\g$, i.e., the subalgebra of all elements that commute with all elements of $\g$. 
Here we find it convenient to explicitly prove a set of equivalent properties of elements in $\hat{A}_n^{\perp}$.

\begin{lemma}\label{lemma:A_purp}
Let $\gamma = (\alpha,\beta) \in \Np{2n}$ with $\alpha \geq \beta$, $\sigma \in \{+,-\}$, and
$\alpha > \beta$ for $\sigma = -$. 
Thus $g_{\sigma}^{\gamma} \neq 0$, and $g_{\sigma}^{\gamma} \in \hat{A}_n^{\perp}$
implies $\abs{\gamma} \geq 3$.
We consider the conditions\\
(A) $g_{\sigma}^{\gamma} \in \hat{A}_n^{\perp}$;\\
(B) $\alpha \neq \beta$;\\
(C) there is no index $k$ with $\alpha_k {+} \beta_k = 1$ and $\alpha_j = \beta_j$ for all $j\neq k$;\\
(D) $\alpha_k {+} \beta_k \ge 2$ and $\alpha_k \neq \beta_k$
for at least one index $k$;\\
(E) $\alpha_k {+} \beta_k = 1$ for indices $k \in \mathcal{I}$ with $\abs{\mathcal{I}} \geq 2$;\\
(F) $\alpha_k {+} \beta_k = 1$ for indices $k \in \mathcal{I}$ with $\abs{\mathcal{I}} \geq 2$
and $\alpha_j = \beta_j$ for $j \notin \mathcal{I}$.\\
Assuming $\abs{\gamma} \geq 3$, we obtain that
\begin{equation*}
\text{ (A) holds } \Leftrightarrow
\text{ (B) and (C) holds } \Leftrightarrow
\text{ (D) or (E) holds } \Leftrightarrow
\text{ (D) or (F) holds}.
\end{equation*}
\end{lemma}
\begin{proof}
It follows from the definitions that (A) is equivalent to the condition that (B) and (C) holds.

We assume (D). Clearly, $g_{\sigma}^{\gamma}$ cannot be in
$\hat{A}_n^0$ or $\hat{A}_n^=$ as otherwise $\alpha_p=\beta_p$ for all $p$. It cannot be in
$\hat{A}_n^{\text{om}}$ as there would be a unique $p$ such that $\alpha_p+\beta_p=1$ while
$\alpha_q=\beta_q$ for all $q\neq p$ which conflicts with (D).
It cannot be in $\hat{A}_n^2$ due to our assumption
$\abs{\gamma}\geq3$ (as all elements of $\hat{A}_n^2$ have degree 2). Thus,
$g_{\sigma}^{\gamma}\in\hat{A}_n^{\perp}$ and (D) implies (A).

Let us assume (E). Again, $g_{\sigma}^{\gamma}$ cannot be in
$\hat{A}_n^0$ or $\hat{A}_n^=$ as otherwise $\alpha_p=\beta_p$ for all $p$. It cannot be in
$\hat{A}_n^{\text{om}}$ (which has a unique $p$ such that $\alpha_p+\beta_p=1$ while
$\alpha_q=\beta_q$ for all $q\neq p$) as there cannot be two indices $k,j$ with
$\alpha_k {+} \beta_k=\alpha_j {+} \beta_j=1$. It cannot be in $\hat{A}_n^2$ since
$\abs{\gamma}\geq3$.
Thus, $g_{\sigma}^{\gamma}\in\hat{A}_n^{\perp}$
and (E) implies (A). Similarly, (F) implies (A).

Let us assume that (B) and (C) holds. Thus 
(C) holds and there exists an index $k$ such that 
either
(i) $\alpha_k\neq \beta_k$ and $\alpha_k+ \beta_k=1$ holds or 
(ii) $\alpha_k\neq \beta_k$ and $\alpha_k+ \beta_k\geq 2$ holds.
Since (ii) implies (D) and (D) implies (C), we obtain that either (D) holds or 
(C) is observed while there exists an index $k$ such that  $\alpha_k\neq \beta_k$ and $\alpha_k+ \beta_k=1$.
The second possibility is equivalent to the existence of an index $k$
such that $\alpha_k+ \beta_k=1$ and not all $j\neq k$ observe $\alpha_j = \beta_j$.
We obtain three cases: (I) There is an index $j\neq k$
with $\alpha_j \neq \beta_j$ and $\alpha_j + \beta_j\geq 2$ (which implies (D)).
(II) There is an index $j\neq k$
with $\alpha_j + \beta_j= 1$ (which implies (E)).
(III) $n=1$, $k=1$, and $\alpha_1 + \beta_1= 1$ 
(which conflicts with $\abs{\gamma}\geq 3$).
Thus (D) or (E) holds.

Let us assume that (D) or (E) holds and we aim to prove that
(D) or (F) holds. To prove this by contradiction, we assume that (D) is not fulfilled, (E) holds
(i.e., $\alpha_k {+} \beta_k = 1$ for indices $k \in \mathcal{I}$ with $\abs{\mathcal{I}} \geq 2$)
and there exists an index $j \not\in \mathcal{I}$ with 
$\alpha_j \neq \beta_j$ and $\alpha_j + \beta_j\neq 1$. This implies
that $\alpha_j \neq \beta_j$ and $\alpha_j + \beta_j\geq 2$ which is impossible 
as (D) is not fulfilled.
\end{proof}

We now need to define an equivalence relation between the operators in the algebra, which will be useful later when proving important properties of specific sequences of operators for which only the highest degree terms in their expansions in terms of basis elements are important. The idea here will be that if $x,x'\in\hat{A}_n$, we will often need to compare them considering only components in their expansion above a certain degree, or we will often retain the components in the commutator $[x,x']$ only above a desired degree. Thus, the necessity to introduce the following:
\begin{definition}\label{equivalence:relation}
The equivalence relation $x \upto{d} y$ for $d \in \Np{}$ holds if
$\deg(x{-}y) < d$.
\end{definition}
Concretely,  $x \upto{d} y$ says that 
$x$ is equal to $y$ up to terms of degree less than $d$. Our main application of this notation
is to expand $x\in\hat{A}_n$ into its highest-degree terms as
$x \upto{d} \sum_{\gamma,\sigma} c_{\sigma}^{\gamma}\, g_{\sigma}^{\gamma}$
with $c_{\sigma}^{\gamma} \in \R$ and $\abs{\gamma}=d$. 
Moreover, $x=x_1+x_2\upto{d} x_1$ implies $[x,x']\upto{d} [x_1,x']$.
We sometimes also use the related equivalence 
$\up$ given by $x\up y$ iff  $\deg(x{-}y) < \deg(x) = \deg(y)$
or $x=y=0$, which does not depend on a particular degree but is impractical to apply in proofs.
Let us briefly illustrate the equivalence relation presented in Definition~\ref{equivalence:relation} with two examples.

\begin{tcolorbox}[colback=orange!3!white,colframe=orange!85!black,title=Example: equivalence up to degree $d$]\label{example:one:zero:one}
\begin{example}[Single elements]
Let $x=(a^\dag b c d-a b^\dag c^\dag d^\dag)+(a^\dag b^\dag-a b)$ and $x'=(a^\dag b c d-a b^\dag c^\dag d^\dag)+(a^\dag b^\dag-a b)+(a^\dag -a )$.
Then, $x\up x'$, $x\upto{4}x'$, $x\upto{3}x'$ and $x\upto{2}x'$. We also have that $x'-x\up a^\dag -a$.
\end{example}
\begin{example}[Commutator]
Let $x=((a^\dag)^3-a^3)+((a^\dag)^2-a^2)$ and $x'=a^\dag -a$.
Then, $[x,x']\upto{2}[(a^\dag)^3-a^3,x']$.
\end{example}
\end{tcolorbox}

We here provide explicit general expressions for the elements in the sets from Definition~\ref{definition:vector:spaces:A}:
$\hat{A}_n^0$, $\hat{A}_n^1$, $\hat{A}_n^2$, $\hat{A}_n^{=}$, and $\hat{A}_n^\text{om}$.
We start with $\hat{A}_n^0$. Among all elements in the skew-hermitian Weyl algebra $\hat{A}_n$ we find that those of the
form $i a^\dag a$ play a special role. Therefore, we choose to collect all $g_+(a^\dag_k a_k):=2ia^\dag_k a_k$ 
together with $2i$ into this
particular real subspace $\hat{A}_n^0$ of $\hat{A}_n$. 
Furthermore, given that the generators $g_\sigma^\gamma$ for $\sigma=\pm$ come always in pairs, it is natural to
ask if there is a map $\Phi:\hat{A}_n\rightarrow\hat{A}_n$ such that
$\Phi:g_\sigma^\gamma\mapsto \kappa_\sigma g_{-\sigma}^\gamma$ for appropriate constants $\kappa_\sigma$.
As discussed in Subsection~\ref{subsection:complementary}, we can use $i a_k^\dag a_k$ to define the desired maps $\Phi_k$ such that
$\Phi_k(g_\sigma^\gamma):=[i a_k^\dag a_k,g_\sigma^\gamma]=\sigma(\alpha_k{-}\beta_k)g_{-\sigma}^\gamma$. Note that
these maps are nothing else than the Lie derivative, i.e., $\Phi_k(g_\sigma^\gamma)=\mathcal{L}_{i a_k^\dag a_k}(g_\sigma^\gamma)$.
For $\hat{A}_n^1$, in contrast, an element $g_\sigma^\gamma\in\hat{A}_n^1$ is characterized by $|\gamma|=1$. 
These elements always have the general forms 
$g_+^{1}(k):=i(a^\dag_k{+}a_k)$ and $g_-^{1}(k):=a^\dag_k{-}a_k$.

For $\hat{A}_n^2$, 
an element $g_\sigma^\gamma\in\hat{A}_n^2$ is characterized by $|\gamma|=2$ and $\alpha\neq\beta$.
The elements of this set are well known in quantum optics as those that model linear optical
transformations \cite{Gerry:Knight:2004}, and they contain single-mode squeezing, two-mode squeezing,
and mode-mixing (or beam splitting): For single-mode squeezing of mode $k$, we introduce the operators
$g^{\textrm{S}}_{+}(k):=i[(a_k^\dag)^2 {+} a_k^2]$ and
$g^{\textrm{S}}_{-}(k):=(a_k^\dag)^2 {-} a_k^2$.
For two-mode squeezing with respect to the 
modes $k\neq j$, we have the operators
$g^{\textrm{T}}_{+}(k,j):=i(a_k^\dag a_j^\dag{+}a_k a_j)$ and
$g^{\textrm{T}}_{-}(k,j):=a_k^\dag a_j^\dag{-}a_k a_j$.
For two-mode mixing of the modes $k\neq j$, we get 
$g^{\textrm{B}}_{+}(k,j):=i(a_k^\dag a_j{+}a_k a_j^\dag)$ and
$g^{\textrm{B}}_{-}(k,j):=a_k^\dag a_j{-}a_k a_j^\dag$.

For $\hat{A}_n^{=}$, an element $g_\sigma^{\tilde{\gamma}}\in\hat{A}_n^{=}$, with
$\tilde{\gamma}=(\tilde{\alpha},\tilde{\beta})$ is characterized by $\tilde{\alpha}=\tilde{\beta}$ and
$|\tilde{\gamma}|\geq4$. Thus, we have that the elements always have the generic form
$g_+^{\tilde{\gamma}}=2ia^{\tilde{\gamma}}$ and
$g_-^{\tilde{\gamma}}=0$ due to Proposition~\ref{diagonal:element:form}. 
For $\hat{A}_n^\text{om}$, an element $g_\sigma^\gamma\in\hat{A}_n^\text{om}$, with
$\gamma=(\alpha,\beta)$, is characterized with one $j$ such that
$\alpha_j+\beta_j=1$ while $\alpha_k=\beta_k$ for all $k\neq j$. We can therefore write
$g_+^\gamma=i a^{\tilde{\gamma}}(a_j^\dag{+}a_j)$ and $g_-^\gamma= a^{\tilde{\gamma}}(a_j^\dag{-}a_j)$, with
$\tilde{\gamma}=(\tilde{\alpha},\tilde{\alpha})$, $\tilde{\alpha}_j=0$, and $|\tilde{\gamma}|=|\gamma|-1$.
Note that, from now on, we use the tilde-notation only for multidegree vectors $\tilde{\gamma}$ of the form $\tilde{\gamma}=(\tilde{\alpha},\tilde{\alpha})$. We do not use this notation for multidegree vectors $\gamma$ of the form $\gamma=(\alpha,\beta)$, with $\alpha\neq\beta$. Thus, all elements of $\hat{A}_n^{=}$ have multidegree $\tilde{\gamma}$.

We now make an important observation. It is immediate to see that the operators with odd degree cannot have $\alpha=\beta$: in fact, the degree of an
element of the algebra is given by $|\gamma|=|\alpha|+|\beta|$. Thus, if $\alpha=\beta$ we have $|\gamma|=2|\alpha|$, which
cannot be odd. This implies that elements of $\hat{A}^0_n$, $\hat{A}_n^{=}$ and $\hat{A}^2_n$ cannot have odd degree. It is also
immediate to verify elements in $\hat{A}^{\text{om}}_n$ have always odd degree. Thus, elements with odd degree can
only belong to $\hat{A}^{1}_n$, $\hat{A}^{\text{om}}_n$ or $\hat{A}^\perp_n$.

We present now two simple examples in the skew-hermitian Weyl algebra $\hat{A}_2$
to focus the following discussion and 
to elucidate how the decomposition of the vector space $\hat{A}_n$ proposed above is directly related to the
dimensionality of a Lie algebra generated by a given set $\mathcal{G}$ of generators.
The intuitive idea will be, as seen in the main result towards the end of this work, that the generators must
obey simple constrains (in terms of having support in certain subvector spaces) for the algebra generated to be finite.
The first example considers a finite-dimensional Lie algebra, while the second example considers an infinite-dimensional Lie algebra.

\begin{tcolorbox}[colback=orange!3!white,colframe=orange!85!black,title=Example: Finite-dimensional Lie algebra]\label{example:one:one:a}
\begin{example}
\label{exA}
The Lie algebra $\g=\lie{ \gens }\subseteq \hat{A}_2$ is generated from the generators
$\gens :=
[ ia_1^\dag a_1, ia_2^\dag a_2, i(a_2^{\dag})^2+ia_2^2, (a_2^{\dag})^2-a_2^2, 
ia_1^\dag a_1(a_2^\dag+a_2), a_1^\dag a_1(a_2^\dag- a_2), i(a_1^\dag a_1)^2 ]$ where $\gens_j$ denotes the $j$th generator.
One can easily verify by direct computations that $\g$ has a finite dimension.
We have $\gens_1, \gens_2 \in  \hat{A}_2^0$, $\gens_3, \gens_4 \in  \hat{A}_2^2$, and $\gens_5, \gens_6 \in  \hat{A}_2^{\text{om}}$.
Moreover, 
$\gens_7$ has nonzero support on $\hat{A}_2^=$ but commutes with all generators. Note that $[\gens_5, \gens_6]
= 2i (a_1^{\dagger})^2 a_1^2 + 2i a_1^{\dagger} a_1 \in \hat{A}_2^{0} \oplus \hat{A}_2^=$
has zero support in $\hat{A}_2^\perp$.
\end{example}
\end{tcolorbox}

\begin{tcolorbox}[colback=orange!3!white,colframe=orange!85!black,title=Example: Infinite-dimensional Lie algebra]\label{example:one:one:b}
\begin{example}
\label{exB}
The Lie algebra $\g=\lie{ \gens }\subseteq \hat{A}_2$ is generated from the set of generators
$\gens := [ i a_1^\dag a_1, i a_2^\dag a_2, i a_1^\dag a_1(a_2^\dag{+}a_2), i a_2^\dag a_2(a_1^\dag{+} a_1) ]$ 
where $\gens_j$ denotes the $j$th generator. 
We anticipate that $\g$ is infinite.
We have $\gens_1, \gens_2 \in  \hat{A}_2^0$ and $\gens_3, \gens_4 \in  \hat{A}^{\text{om}}_2$.
Moreover, $[\gens_3, \gens_4 ]= 
a_2^{\dagger} a_1 a_2^2 + (a_2^{\dagger})^2 a_1 a_2 - a_1^{\dagger} a_1^2 a_2 - a_1^{\dagger} a_2^{\dagger}  a_2^2 + 
a_1^{\dagger}  a_2^{\dagger} a_1^2 - a_1^{\dagger}  (a_2^{\dagger})^2 a_2 - (a_1^{\dagger})^2 a_1 a_2 + (a_1^{\dagger})^2 a_2^{\dagger} a_1 + 
a_2^{\dagger} a_1 - a_1^{\dagger} a_2$
has nonzero support in $\hat{A}_2^\perp$. The condition $\PP{[\gens_3, \gens_4]}{\hat{A}_n^\perp}\neq0$ for $\gens_3, \gens_4 \in  \hat{A}^{\text{om}}_2$ will be sufficient for the algebra to be infinite.
\end{example}
\end{tcolorbox}

We now briefly provide a few comments on the commutation properties of the vector spaces that we have introduced here.
It is not guaranteed that the commutators of two arbitrary elements of a vector space $\hat{A}^V_n$ are themselves elements of the vector
space, i.e., $[\hat{A}^V_n,\hat{A}^V_n]\subseteq\hat{A}^V_n$ is not always true. 
Here we list a series of properties that provide intuition about the results presented in this work.
They can be verified by direct inspection using the properties bounding the maximum degree of a
commutator, as well as other theorems listed in the following sections. In particular,
\begin{subequations}
\begin{align}
	[\hat{A}^0_n,\hat{A}^0_n] &= [\hat{A}^{=}_n,\hat{A}^{=}_n] = [\hat{A}^0_n,\hat{A}^{=}_n]=0;\\
	[\hat{A}^0_n,\hat{A}^K_n] &=\hat{A}^K_n, \text{ where }
	K\in \{1,2,\text{om},\perp\} \text{ (see Lemma~\ref{commutator:no:change:vector:space});}\label{3:b}\\
	[\hat{A}^1_n,\hat{A}^1_n] &\subseteq\hat{A}^0_n
	\text{ (see Lemma~\ref{degree:commutator:conjugate:basis:operators});}\\
	[\hat{A}^2_n,\hat{A}^2_n] &=\hat{A}^2_n, \text{ where } 
	\hat{A}^2_n \text{ is the Lie algebra } \mathfrak{sp}(2n,\mathbb{R}) \text{ of the symplectic group } \mathrm{Sp}(2n,\mathbb{R})\; \text{\cite{Adesso:Ragy:2014}};\\
	[\hat{A}^1_n,\hat{A}^2_n] &=\hat{A}^1_n
	\text{ (see Lemma~\ref{degree:commutator:conjugate:basis:operators}).}
\end{align}
\end{subequations}
Note that \eqref{3:b} is true since the commutator of an element in $\hat{A}^0_n$ with an element $g_\sigma^\gamma\in \hat{A}^K_n$ gives either $0$ or $g_{-\sigma}^\gamma$, both of which are in $\hat{A}^K_n$. The same is true for linear combinations.
 
This provides a first insight into the Lie-algebra structure in $\hat{A}_n$, and helps to better visualize why certain results 
in Sections~\ref{algebra:properties:section}-\ref{main:results:section} do apply.

We have introduced the decomposition of the algebra $\hat{A}_n$ in the vector spaces $\hat{A}_n^{\text{K}}$
for $K\in \{0,1,2,{=},\text{om},\perp\}$.
Let us exploit this decomposition to introduce the following intersections between a set of Lie-algebra generators
$\mathcal{G}$ and the spaces $\hat{A}_n^{\text{K}}$. We have:
\begin{equation*}
\mathcal{G}^{\text{K}}:=\mathcal{G}\cap \hat{A}_n^{\text{K}}.
\end{equation*}
In addition, we also have that $\mathcal{G}^\perp:=\mathcal{G}\cap \hat{A}_n^\perp$ to
complement the definitions provided here. Notice that these sets might inherit some of the properties of the algebras from
which they are constructed. For example, one has that $[\mathcal{G}^2,\mathcal{G}^2]\subseteq\hat{A}_n^2$, i.e.,
the commutator of two elements of $\mathcal{G}^2$
will always be in $\hat{A}_n^2$, but might not be in $\mathcal{G}^2$.

\section{Properties of the Weyl algebra and the skew-hermitian Weyl algebra}\label{algebra:properties:section}
We now continue studying the properties of monomials, basis elements, as well as their commutators. 
Here we will show that there exists a specific condition on the multi-degrees of two monomials,
or two basis elements, that allows us to predict the highest possible degree
of their commutator. This condition plays an important role in this work.

\subsection{Basic properties of commutators of algebra elements: monomials}
We start with providing useful properties of elements of the skew-hermitian Lie-algebra $\hat{A}_n$ using normal-ordered monomials.

\begin{proposition} \label{app:polynomials}
(a) Consider two polynomials $g$ and $\hat{g}$ in the 
 $n$th Weyl algebra $A_n$.
(a1) We have a unique expansion $g= \sum_{\gamma} c_{\gamma} a^{\gamma}$ with multi-indices $\gamma  \in \Np{2n}$ and $c_{\gamma} \neq 0$. 
(a2) $\deg(g{+}\hat{g}) \leq \max\{\deg(g), \deg(\hat{g})\}$, (a3)  $\deg(g \hat{g}) = \deg(g) + \deg(\hat{g})$,
(a4) $\deg([g,\hat{g}]) \leq \deg(g) + \deg(\hat{g}) -2$. 
(b) Given a  monomial $m \in A_n$ 
with multi-degree $\gamma= (\alpha,\beta) \in \Np{2n}$.
We have a unique expansion $m= \sum_{\hat{\gamma}} c_{m,\hat{\gamma}} a^{\hat{\gamma}}$ with multi-indices
$\hat{\gamma}  \in \Np{2n}$ and $c_{m,\hat{\gamma}} \neq 0$. The 
parities of $\abs{\hat{\gamma}}$ and $\deg(m)$ agree
and $\hat{\gamma} = (\hat{\alpha},\hat{\beta}) = (\alpha {-} \oneshalf,\beta {-} \oneshalf)$ for nonnegative integer vectors $\oneshalf \in \Np{n}$.
There is a single $\hat{\gamma}_*$ such that  $\abs{\hat{\gamma}_*} = \deg(m)$, and $c_{m,\hat{\gamma}_*} = 1$ and
$\hat{\gamma}_*=(\alpha,\beta)$ for this particular $\hat{\gamma}_*$.
\end{proposition}

\begin{proof}
The unique expansion in
part (a1) follows as the normal-ordered monomials form a basis of $A_n$, see \cite[p.~9]{Coutinho:1995}.  Theorem~1.1 on p.~14
of \cite{Coutinho:1995} proves (a2), (a3), and (a4).
After applying (a1),
the canonical commutation relations
imply the rest of (b). 
\end{proof}

\begin{proposition} \label{app:polynomials:same:order} 
Consider a monomial $m \in A_n$ 
with multi-degree $(\alpha,\beta) \in \Np{2n}$. We have $m=m^\dag$
if and only if $\alpha=\beta$.
\end{proposition}

\begin{proof}
We have $m=\sum_{\oneshalf} c_{m,(\alpha {-} \oneshalf,\beta {-} \oneshalf)} a^{(\alpha {-} \oneshalf,\beta {-} \oneshalf)}$
for nonnegative integer vectors $\oneshalf$ due to Proposition~\ref{app:polynomials}(b) and we obtain that 
$m^{\dagger}=\sum_{\oneshalf}
c_{m,(\alpha {-} \oneshalf,\beta {-} \oneshalf)} a^{(\beta {-} \oneshalf,\alpha {-} \oneshalf)}$. The if case follows immediately.
For the other direction, we assume that $m=m^\dag$. We consider the part of the 
normal-ordered expansions with $\oneshalf$ equal to the zero vector (or equivalently with
the normal-ordered monomial of highest-degree). This implies that $a^{(\alpha,\beta)}
=a^{(\beta,\alpha)}$, which completes the proof.
\end{proof}

\begin{proposition}\label{diagonal:element:form}
Consider the element $g_+(\widetilde{m})\in\hat{A}_n$
for a given monomial $\tilde{m} \in A_n$ with $\mdeg(\widetilde{m})=(\tilde{\alpha},\tilde{\alpha})$.
Then, $g_+(\widetilde{m})=2i\widetilde{m}$.
\end{proposition}
\begin{proof}
The claim is an immediate consequence of Proposition~\ref{app:polynomials:same:order}. In fact, we have
$\mdeg(\widetilde{m})=(\tilde{\alpha},\tilde{\alpha})$,
which implies $\widetilde{m}=\widetilde{m}^\dag$, and $g_+(\widetilde{m})=2i\widetilde{m}$.
\end{proof}

We here bring the attention to the fact that all elements in $\hat{A}_n^0$ and $\hat{A}_n^{=}$ have the generic form presented in
Proposition~\ref{diagonal:element:form}, as has also been discussed above. In the remaining part of this work, we retain the
convention that a monomial $\tilde{m}$ or a basis element $g^{\tilde{\gamma}}_{\tilde{\sigma}}$ that are elements of either
$\hat{A}_n^0$ and $\hat{A}_n^{=}$ are connected to objects marked by a \emph{tilde}. 

\begin{lemma} \label{app:lemma:powers}
(a) Let $a^{\dagger}$ and  $a$ denote the single-mode creation and annihilation operators. Assuming that $r$ and
$s$ are nonnegative integers, we obtain:
(a1) $[a^r, (a^{\dagger})^s] = \sum_{k=1}^{\min(r,s)} k! \binom{r}{k} \binom{s}{k} (a^{\dagger})^{s-k} a^{r-k}$;
(a2) $[a^r, (a^{\dagger})^s] = rs (a^{\dagger})^{s-1} a^{r-1} + A$ where $\deg(A) \leq r+s-4$;
(a3) $(a^{\dagger})^r a^r = \sum_{k=0}^r s(r,k) (a^{\dagger} a)^k$ where $s(r,k)$ are the Stirling numbers of the first kind
with $s(r,0) =0$ if $r>0$ and $s(r,r)=1$ if $r\geq 0$; 
(a4) $[(a^{\dagger})^r a^r, (a^{\dagger})^s a^s] = 0$ for all $r,s\geq0$.
(b) Given two monomials $\tilde{m}, \tilde{m}' \in A_1$ 
with multi-degrees  $\mdeg(\tilde{m}) = (\tilde{r},\tilde{r})$ and $\mdeg(\tilde{m}') = (\tilde{r}',\tilde{r}')$ for
$0 \leq  \tilde{r},\tilde{r}' \in \N$, we have $[\tilde{m},\tilde{m}']=0$.
(c)~Given two monomials $\tilde{m}, \tilde{m}' \in A_n$ 
with multi-degrees 
$\mdeg(\tilde{m}) = (\tilde{\alpha},\tilde{\alpha})$ and $\mdeg(\tilde{m}') = (\tilde{\alpha}',\tilde{\alpha}')$
for  $\tilde{\alpha}, \tilde{\alpha}' \in \Np{n}$, we have $[\tilde{m},\tilde{m}']=0$.
\end{lemma}

\begin{proof}
For the statement of (a1), refer to \cite[p.~873]{BJ:1925} or \cite[Eq.~(6.2)]{MS:2016} and (a2) is a direct consequence of (a1).
For (a3), we refer to \cite[Eq.~(6.10)]{MS:2016} and (a4) follows directly from
(a3). We use Proposition~\ref{app:polynomials}(b) to expand $\tilde{m}$ and $\tilde{m}'$ from (b) as
$\tilde{m} = \sum_{0 \leq \tilde{\oneshalf} \in \N}  c_{\tilde{m},\tilde{\oneshalf}} a^{(\tilde{r}-\tilde{\oneshalf},\tilde{r}-\tilde{\oneshalf})}$ and
$\tilde{m}' = \sum_{0 \leq \tilde{\oneshalf}' \in \N}  c_{\tilde{m}',\tilde{\oneshalf}'} a^{(\tilde{r}'-\tilde{\oneshalf}',\tilde{r}'-\tilde{\oneshalf}')}$.
All elements of the form $a^{(\tilde{s},\tilde{s})}$ can be expanded using (a3).
The statement  $[\tilde{m},\tilde{m}']=0$ of (b) then follows by applying (a4). The proof of (c) is similar.
\end{proof}

We have provided useful properties regarding the decomposition of monomials into normal-ordered elements of the  Weyl algebra.
Here we investigate properties of the degrees of two monomials in order for
them to commute. This will also affect the maximum achievable degree
of the commutator of two monomials.
Before we start, let us introduce the vector $\kappa_p(r,s) \in \Np{2n}$ which has only zero entries except for
the entry $r$ at position $p$ and the entry $s$ at position $n+p$. Also, in this notation we have that $\ones_p:=\kappa_p(1,1)$. Note that one-dimensional vectors $\alpha$ and $\beta$ are denoted by positive natural numbers $r$ and $s$ respectively to highlight this particular dimensionality.

\begin{corollary}\label{theorem:commutator:monomials:one:diagonal}
(a) Consider two monomials $m, \hat{m} \in A_1$ with
multi-degrees $ (r,s),
 (\hat{r},\hat{s}) \in \Np{2}$. (a1) If $[m,\hat{m}]=0$, then
$s\hat{r}=r\hat{s}$.
(a2) If $\hat{r}=\hat{s}\neq 0$ and $[m,\hat{m}]=0$ holds, then
 $r=s$.
(b) Consider two monomials $m, \hat{m} \in A_n$ 
with multi-degrees $ (\alpha,\beta),
 (\hat{\alpha},\hat{\beta}) \in \Np{2n}$.
(b1) If $[m,\hat{m}]= 0$, then $\beta_j \hat{\alpha}_j=\alpha_j \hat{\beta}_j$ for all $j \in \{1,\ldots,n\}$.
(b2) If  $\hat{\alpha}=\hat{\beta}$ and $[m,\hat{m}]=0$ holds, then
$\alpha_j=\beta_j$ for each index $j$ such that $\hat{\alpha}_j\neq0$.
(c) Given two monomials $m, \tilde{m} \in A_n$ 
with multi-degrees $ (\alpha,\beta),
 (\tilde{\alpha},\tilde{\alpha}) \in \Np{2n}$,
then $[m,\tilde{m}]= 0$ if and only if $\alpha_j=\beta_j$ for each index $j$ such that $\tilde{\alpha}_j\neq0$.
\end{corollary}

\begin{proof}
(a1) Using Proposition~\ref{app:polynomials}(b),
we expand $m=\sum_{k\geq 0} c_{m,k} (a^{\dagger})^{r-k} a^{s-k}$ with $c_{m,0}=1$ 
and $\hat{m}=\sum_{\ell\geq 0} c_{\hat{m},\ell} (a^{\dagger})^{\hat{r}-\ell} a^{\hat{s}-\ell}$ with
$c_{\hat{m},0}=1$. It follows from $[m,\hat{m}]=0$
that $0=(a^{\dagger})^{r} a^{s} (a^{\dagger})^{\hat{r}} a^{\hat{s}} - (a^{\dagger})^{\hat{r}} a^{\hat{s}} (a^{\dagger})^{r} a^{s}$.
Recall from Lemma~\ref{app:lemma:powers}(a2) that 
$a^v (a^\dagger)^w = (a^\dagger)^w a^v + v w (a^\dagger)^{w-1} a^{v-1} +A$
with $\deg(A) \leq v+w-4$. We obtain that 
$0= s\hat{r} (a^\dagger)^{r+\hat{r}-1} a^{s+\hat{s}-1} 
- r\hat{s} (a^\dagger)^{r+\hat{r}-1} a^{s+\hat{s}-1}$ which implies 
the condition $s\hat{r}=r\hat{s}$ of (a1). The proof of (a2) is an immediate consequence of (a1).
For the multi-mode case in (b), we consider for $m$ and $\hat{m}$
the normal-ordered monomials 
 $\prod_{p=1}^n m_p$  and 
$ \prod_{\hat{p}=1}^n \hat{m}_{\hat{p}}$ of highest order
where $m_p:=a^{\kappa_p(\alpha_p,\beta_p)}$ and 
$\hat{m}_{p}:=a^{\kappa_p(\hat{\alpha}_p,\hat{\beta}_p)}$. We
apply the formula $[AB,CD]=A[B,C]D+[A,C]BD+CA[B,D]+C[A,D]B$ for 
$A=m_1$, $B=\prod_{p=2}^n m_p$, $C= \hat{m}_1$, and $D=\prod_{\hat{p}=2}^n \hat{m}_{\hat{p}}$. 
It follows that $[\prod_{p=1}^n m_p,\prod_{\hat{p}=1}^n \hat{m}_{\hat{p}}]= [m_1, \hat{m}_1]  
\prod_{q=2}^n m_q \hat{m}_q + \hat{m}_1 m_1 
[  \prod_{p=2}^n m_p, \prod_{\hat{p}=2}^n \hat{m}_{\hat{p}} ]$. By induction, we obtain
$[\prod_{p=1}^n m_p,\prod_{\hat{p}=1}^n \hat{m}_{\hat{p}}]= \sum_{q=1}^n (\prod_{p=1}^{q-1} \hat{m}_p m_p ) 
[m_q, \hat{m}_q] (\prod_{\hat{p}=q+1}^n m_{\hat{p}} \hat{m}_{\hat{p}})$. 
The contribution of highest order is then given by $
\sum_{q=1}^n (\beta_q \hat{\alpha}_q - \alpha_q \hat{\beta}_q)  
a^{\gamma+\hat{\gamma} - \ones_q}$. This completes the proof of (b1) as $[m,\hat{m}]= 0$
implies that $\beta_j \hat{\alpha}_j = \alpha_j \hat{\beta}_j$ for all $j$ since all $a^{\gamma+\hat{\gamma} - \ones_q}$
for different $q$ are linearly independent operators. The proof of (b2) follows directly.
The only-if case of (c) is a consequence of (b2). To prove the if case of (c), we assume that $\alpha_j=\beta_j$
or $\tilde{\alpha}_j=0$ holds for each $j$. Assume that there exists a $j$ such that $\tilde{\alpha}_j=0$, 
then $[m,\hat{m}]=0$ if and only if $m_j[m_{\neg j},\tilde{m}]=0$, where $m_j$ contains the contributions of the $j$th mode, and $m_{\neg j}$ has been obtained from $m$ by
removing all contributions of the $j$th mode. This reduces the proof to the case where all 
$\tilde{\alpha}_j\neq 0$, which is resolved by 
Lemma~\ref{app:lemma:powers}(c).
\end{proof}

\begin{lemma}\label{A:one:commutative}
The vector space $\hat{A}_n^0\oplus\hat{A}_n^{=}$ together with the commutator operation is an abelian subalgebra of $\hat{A}_n$.
\end{lemma}
\begin{proof}
It is an immediate consequence of Corollary~\ref{theorem:commutator:monomials:one:diagonal}(c) that two elements
$g_+^{\tilde{\gamma}}$ and $g_+^{\tilde{\gamma}'}$ with $\gamma=(\alpha,\alpha)$ and $\tilde{\gamma}=(\tilde{\alpha},\tilde{\alpha})$ commute if they
belong to $\hat{A}_n^0\oplus\hat{A}_n^{=}$. Thus, any two
linear combinations $x=\sum_p C_p g_+^{\tilde{\gamma}_p}$ and
$x'=\sum_p C_p' g_+^{\tilde{\gamma}_p'}$ in $\hat{A}_n^0\oplus\hat{A}_n^{=}$, with $\tilde{\gamma}_p=(\tilde{\alpha}^{(p)},\tilde{\alpha}^{(p)})$ and $\tilde{\gamma}_p'=(\tilde{\alpha}^{\prime(p)},\tilde{\alpha}^{\prime(p)})$, also commute due to bi-linearity of the commutator.
\end{proof}

We continue by proving one central result in this work. As mentioned before, the commutator and the degree of an element
of the algebra are the key objects of this work. We here provide a condition that the multi-degrees of two monomials need
to satisfy in order for the commutator to have the highest degree possible (see Proposition~\ref{app:polynomials}).

\begin{theorem} \label{app:theorem:normal} 
Given two normal-ordered monomials $a^{\gamma}= (a^{\dagger})^r a^s$ and $a^{\hat{\gamma}} = (a^{\dagger})^{\hat{r}} a^{\hat{s}}$ with
$\gamma=(r,s), \hat{\gamma}=(\hat{r},\hat{s}) \in \Np{2}$ in the Weyl algebra $A_1$ of single-mode creation and annihilation operators. 
Assume that $[a^{\gamma}, a^{\hat{\gamma}}] \neq 0$. Then:
$\deg([a^{\gamma}, a^{\hat{\gamma}}]) =  \abs{\gamma} +  \abs{\hat{\gamma}} -4$ if $\hat{r}s=r\hat{s}$; or 
$\deg([a^{\gamma}, a^{\hat{\gamma}}]) =  \abs{\gamma} +  \abs{\hat{\gamma}} -2$ otherwise.
\end{theorem}

\begin{proof}
We have
$[a^{\gamma}, a^{\hat{\gamma}}] = (a^{\dagger})^r a^s (a^{\dagger})^{\hat{r}} a^{\hat{s}} -
(a^{\dagger})^{\hat{r}} a^{\hat{s}} (a^{\dagger})^r a^s
= (a^{\dagger})^r [ a^s, (a^{\dagger})^{\hat{r}}] a^{\hat{s}} - (a^{\dagger})^{\hat{r}}     [ a^{\hat{s}}, (a^{\dagger})^{r}]  a^s$.
We apply Lemma~\ref{app:lemma:powers}(a2)
and obtain  $[a^{\gamma}, a^{\hat{\gamma}}] =  s \hat{r} (a^{\dagger})^r  (a^{\dagger})^{\hat{r}-1} a^{s-1}  a^{\hat{s}} - 
r \hat{s}  (a^{\dagger})^{\hat{r}}  (a^{\dagger})^{r-1} a^{\hat{s}-1}   a^s + 
A = (s \hat{r}-r \hat{s})  (a^{\dagger})^{r+\hat{r}-1}  a^{s+\hat{s}-1} + A$ with $\deg(A) \leq \abs{\gamma} + \abs{\gamma} -4$.
This verifies that $\deg([a^{\gamma}, a^{\hat{\gamma}}]) \leq  \abs{\gamma} +  \abs{\hat{\gamma}} -4$ if $\hat{r}s=r\hat{s}$ and 
$\deg([a^{\gamma}, a^{\hat{\gamma}}]) =  \abs{\gamma} +  \abs{\hat{\gamma}} -2$ otherwise. It remains to prove that 
$\deg([a^{\gamma}, a^{\hat{\gamma}}]) =  \abs{\gamma} +  \abs{\hat{\gamma}} -4$ if $\hat{r}s=r\hat{s}$
(under our assumption that $[a^{\gamma}, a^{\hat{\gamma}}] \neq 0$).
Neither $r$, $s$, $\hat{r}$, or $\hat{s}$ can be zero, as $\hat{r}s=r\hat{s}$ would imply that $[a^{\gamma}, a^{\hat{\gamma}}] =0$.
In the following, we assume that $\abs{\gamma} \leq \abs{\hat{\gamma}}$ (and the other case is similar). From $\hat{r}s=r\hat{s}$, we obtain
$\hat{r}(r{+}s)=r(\hat{r}+\hat{s}) \geq r (r{+}s)$ which implies that $\hat{r} \geq r$. Similarly, we obtain that $\hat{s} \geq s$.
The equation $\hat{r}s=r\hat{s}$
yields $\hat{r}/r = \hat{s}/s$ which means that there is a rational number $p>1$ such that $\hat{r} = p r$ and $\hat{s} = p s$.
Note that $p=1$ would imply that $[a^{\gamma}, a^{\hat{\gamma}}]=0$.
We determine now the degree of  $[a^{\gamma}, a^{\hat{\gamma}}] = [a^{(r,s)}, a^{(pr,ps)}] = 
[(a^{\dagger})^r a^s, (a^{\dagger})^{pr} a^{ps}]= (a^{\dagger})^r [ a^s, (a^{\dagger})^{pr}] a^{ps} - (a^{\dagger})^{pr}     [ a^{ps}, (a^{\dagger})^{r}]  a^s$.
We apply Lemma~\ref{app:lemma:powers}(a1) and get
{\small
\begin{align*}
 [a^{\gamma}, a^{\hat{\gamma}}] = (a^{\dagger})^r
\left\{ prs (a^{\dagger})^{pr-1} a^{s-1} + 2 \binom{pr}{2} \binom{s}{2} (a^{\dagger})^{pr-2} a^{s-2} \right\}
 a^{ps} - (a^{\dagger})^{pr}    
\left\{ prs (a^{\dagger})^{r-1} a^{ps-1} + 2 \binom{ps}{2} \binom{r}{2} (a^{\dagger})^{r-2} a^{ps-2} \right\}
 a^s + \hat{A}
\end{align*}
}
with $\deg(\hat{A}) \leq  \abs{\gamma} + \abs{\gamma} -6$. We obtain
 $[a^{\gamma}, a^{\hat{\gamma}}] = \{ 2 \binom{pr}{2} \binom{s}{2}  - 2 \binom{ps}{2} \binom{r}{2} \}
 (a^{\dagger})^{r+pr-2}
 a^{s+ps-2}
 + \hat{A}$. Note that $2 \binom{pr}{2} \binom{s}{2}  - 2 \binom{ps}{2} \binom{r}{2} = 
 \{pr (pr{-}1) s (s{-}1) - ps (ps{-}1)  r(r{-}1)\}/2 = \{ (pr{-}1) (s{-}1) -(ps{-}1) (r{-}1)  \} prs/2 = (s{-}r) (p{-}1) p  rs /2$
 which is a nonzero integer for $r\neq s$. But $r=s$ would imply via Lemma~\ref{app:lemma:powers}(a4) that
 $[a^{\gamma}, a^{\hat{\gamma}}] =0$. This proves that
 $\deg([a^{\gamma}, a^{\hat{\gamma}}]) =  \abs{\gamma} +  \abs{\hat{\gamma}} -4$ for $\hat{r}s=r\hat{s}$. 
\end{proof}

We move on to prove the final result of this subsection. Here, we extend the result of Theorem~\ref{app:theorem:normal} to
arbitrary monomials that do not commute. We anticipate that the constraint found before will still play a role. Before
providing the claim and proof, we list a few examples to help visualize the idea behind the theorem.

\begin{tcolorbox}[colback=orange!3!white,colframe=orange!85!black,title=Example: relation between the degree of two
monomials and the degree of their commutator]\label{example:two}
We present two examples that elucidate the instances of Theorem~\ref{theorem:commutator:monomials}. The idea is that the
degree of the commutator of two monomials cannot be controlled a priori if the degree is not maximal. 

\begin{example}
Consider the elements $a^{\gamma}= (a^{\dagger})^r a^s$ with $\gamma=(2,6)$ and $r=2$ and $s=6$, as well as
$a^{\hat{\gamma}} = (a^{\dagger})^{\hat{r}} a^{\hat{s}}$ with $\hat{\gamma}=(3,9)$ and $\hat{r}=3$ and $\hat{s}=9$
in the Weyl algebra $A_1$. We have $[a^{\gamma}, a^{\hat{\gamma}}] \upto{16} 18 (a^{\dagger})^3  a^{13}$ and $\deg([a^{\gamma}, a^{\hat{\gamma}}] )=
\deg(a^{\gamma}) + \deg(a^{\hat{\gamma}}) -4$. We observe $\hat{r}s=r\hat{s}$ and there is a rational number
$p =  \hat{r}/r = \hat{s}/s=3/2$ (but $p$ is in general not an integer).
\end{example}

\begin{example}
The two monomials $m=a^{\dagger} a a^{\dagger} a a^{\dagger}=(a^{\dagger})^3 a^2 + 3 (a^{\dagger})^2 a + a^{\dagger}$
and $\hat{m}=a (a^{\dagger})^3 a =(a^{\dagger})^3 a^2 + 3 (a^{\dagger})^2 a$ have a commutator of
$[m,\hat{m}] = -2 (a^{\dagger})^3 a - 3 a^{\dagger}$.
We observe that $\deg([m,\hat{m}]) = \deg(m) + \deg(\hat{m}) -6$. Consequently,
the degree of a nonzero commutator $[m,\hat{m}] $ of \emph{monomials} can drop by more than $4$ when compared with $\deg(m) + \deg(\hat{m})$.
\end{example}
\end{tcolorbox}

We are now ready to proceed with the theorem.

\begin{theorem}\label{theorem:commutator:monomials}
Consider two monomials $m, \hat{m} \in A_n$ 
with multi-degrees $ (\alpha,\beta),
 (\hat{\alpha},\hat{\beta}) \in \Np{2n}$.
Assume that $[m,\hat{m}]\neq 0$ and 
$\deg(m)\leq \deg(\hat{m})$. If 
$ \hat{\alpha}_j \beta_j = \alpha_j \hat{\beta}_j$ holds
for all $1\leq j \leq n$, then $\deg([m,\hat{m}]) \leq \deg(m) + \deg(\hat{m}) -4$.
Otherwise, $\deg([m,\hat{m}]) = \deg(m) + \deg(\hat{m}) -2$.
\end{theorem}

\begin{proof}
Let us first prove the single-mode case with $m, \hat{m} \in A_1$. We set
$\mdeg(m) = (r,s), \mdeg(\hat{m}) = (\hat{r},\hat{s}) \in \Np{2}$. It follows from Proposition~\ref{app:polynomials}(b) that
$m = \sum_{0 \leq \oneshalf \in \N}  c_{m,\oneshalf} a^{(r-\oneshalf,s-\oneshalf)}$ and
$\hat{m} = \sum_{0 \leq \hat{\oneshalf} \in \N}  c_{\hat{m},\hat{\oneshalf}} a^{(\hat{r}-\hat{\oneshalf},\hat{s}-\hat{\oneshalf})}$ with
$c_{m,0} = c_{\hat{m},0} = 1$. We apply Proposition~\ref{app:polynomials}(a4) and obtain that 
\begin{align*}
[m, \hat{m}] = [a^{(r,s)}, a^{(\hat{r},\hat{s})}]+ c_{\hat{m},1} [a^{(r,s)}, a^{(\hat{r}-1,\hat{s}-1)}] + c_{m,1} [a^{(r-1,s-1)}, a^{(\hat{r},\hat{s})}] + \hat{A}
\end{align*}
with 
$\deg([a^{(r,s)}, a^{(\hat{r},\hat{s})}]) \leq \deg(m) +\deg(\hat{m}) -2$, $\deg([a^{(r,s)}, a^{(\hat{r}-1,\hat{s}-1)}]) \leq \deg(m) +\deg(\hat{m}) -4$, 
$\deg([a^{(r-1,s-1)}, a^{(\hat{r},\hat{s})}]) \leq \deg(m) +\deg(\hat{m}) -4$, and 
$\deg(\hat{A}) \leq \deg(m) +\deg(\hat{m}) -6$. Now the single-mode case follows by applying Theorem~\ref{app:theorem:normal} to
$[a^{(r,s)}, a^{(\hat{r},\hat{s})}]$.
For the multi-mode case, we set $m = \prod_{p=1}^n m_p$  and $\hat{m} = \prod_{\hat{p}=1}^n \hat{m}_{\hat{p}}$ where
$m_p$ and $\hat{p}$ denote the contribution form
the $p$th mode. We
apply the formula $[AB,CD]=A[B,C]D+[A,C]BD+CA[B,D]+C[A,D]B$ for $A=m_1$, $B=\prod_{p=2}^n m_p$, $C= \hat{m}_1$, and
$D=\prod_{\hat{p}=2}^n \hat{m}_{\hat{p}}$. It follows that
$[m,\hat{m}]= [m_1, \hat{m}_1]  \prod_{q=2}^n m_q \hat{m}_q + \hat{m}_1 m_1 [  \prod_{p=2}^n m_p, \prod_{\hat{p}=2}^n \hat{m}_{\hat{p}} ]$.
By induction, we obtain
\begin{align*}
[m,\hat{m}]= \sum_{q=1}^n \left(\prod_{p=1}^{q-1} \hat{m}_p m_p \right) [m_q, \hat{m}_q]
\left(\prod_{\hat{p}=q+1}^n m_{\hat{p}} \hat{m}_{\hat{p}}\right).
\end{align*}
The multi-mode case now follows by 
simultaneously applying the single-mode statement to $ [m_q, \hat{m}_q]$ for all $q$.
\end{proof}

\subsection{Commutators with elements of $\hat{A}_n^0$}

We have already noted when defining the subspaces of the whole algebra $\hat{A}_n$ that there exists a special class of
operators conveniently collected in the set $\hat{A}_n^0$. We also noted that generators $g_\sigma^\gamma$ naturally come
in pairs labelled by $\sigma=\pm$. We now use the elements of $\hat{A}_n^0$ to construct natural relations between $g_+^\gamma$ and
$g_-^\gamma$ for arbitrary $\gamma=(\alpha,\beta)$, which is of relevance when $\alpha\neq\beta$. In this sense, what
we discuss below is of particular importance to all elements of $\hat{A}_n$ excluding those in $\hat{A}_n^{=}$, since
$[\hat{A}_n^0,\hat{A}_n^{=}]=0$ applies as an immediate consequence of Corollary~\ref{theorem:commutator:monomials:one:diagonal}.
We hereby proceed to prove the necessary results for this case.

\begin{lemma}\label{commutator:no:change:vector:space}
Let $\hat{A}_n^K$ be any of the subspaces listed in Definition~\ref{definition:vector:spaces:A}. Then,
$[ia_s^\dag a_s,\hat{A}_n^K]\in\hat{A}_n^K$ for all $s\in \{1,\ldots,n\}$ and all $K\in \{0,1,2,{=},\text{om},\perp\}$.
\end{lemma}

\begin{proof}
It is immediate to see that $[ia_s^\dag a_s,\hat{A}_n^{0}]=0\in \hat{A}_n^{0}$, $ [ia_s^\dag a_s,\hat{A}_n^{=}]=0\in \hat{A}_n^{=}$ and $[ia_s^\dag a_s,\hat{A}_n^{1}]\in \hat{A}_n^{1}$.
Also, let $g_\sigma^\gamma\in\hat{A}_n^{2}$ with $\gamma=(\alpha,\beta)$. Then,
$[ia_s^\dag a_s,g_\sigma^\gamma]=0$ or
$\deg([ia_s^\dag a_s,g_\sigma^\gamma])=2$,
i.e.\ $[ia_s^\dag a_s,g_\sigma^\gamma]=\sigma (\alpha_s{-}\beta_s)g_{-\sigma}^\gamma\neq0$, 
and this means that $[ia_s^\dag a_s,g_\sigma^\gamma]\in\hat{A}_n^{2}$. This is true for all $s$ and all
$g_\sigma^\gamma\in\hat{A}_n^{2}$. Let $g_\sigma^\gamma\in\hat{A}_n^{\text{om}}$ and let
there be a unique $j$ such that $\alpha_j+\beta_j=1$ while $\alpha_k=\beta_k$ for all $k\neq j$. Then, it is immediate
to verify that $[ia_s^\dag a_s,g_\sigma^\gamma]=\sigma g_{-\sigma}^\gamma\,\delta_{sj}\in\hat{A}_n^{\text{om}}$
for all $s$ and  $g_\sigma^\gamma\in\hat{A}_n^{\text{om}}$. Finally, let $g_\sigma^\gamma\in\hat{A}^\perp_n$.
Then, it is immediate to verify that there is an $s$ such that 
$\mdeg([ia_s^\dag a_s,g_\sigma^\gamma])=\mdeg(g_\sigma^\gamma)$ if $[ia_s^\dag a_s,g_\sigma^\gamma]\neq0$.
This means that $[ia_s^\dag a_s,g_\sigma^\gamma]$ must be in $\hat{A}^\perp_n$.
\end{proof}

\begin{lemma}\label{helper:lemma}
We observe that  (a) $[a_k^\dagger a_k, m] = (\alpha_k {-} \beta_k) m$
and (b) $[a_k^\dagger a_k, m^\dagger] = (\beta_k {-} \alpha_k) m^\dagger$
for any monomial $m \in A_n$ with $\mdeg(m)=\gamma=(\alpha,\beta)$.
\end{lemma}
\begin{proof}
Recall from Lemma~\ref{app:lemma:powers}(a1)
that $[a, (a^\dagger)^r] = r (a^\dagger)^{r-1}$ and $[a^s, a^\dagger] = s a^{s-1}$.
We obtain $[a^{(1,1)},a^{(r,s)}]=a^\dagger a (a^\dagger)^r a^s - (a^\dagger)^r a^s a^\dagger a
= a^\dagger [a, (a^\dagger)^r] a^s - (a^\dagger)^r [a^s, a^\dagger] a = (r{-}s) a^{(r,s)}$ for any normal-ordered monomial $a^{(r,s)}$
and $n=1$. Essentially the same proof applies to any normal-ordered monomial and arbitrary $n$. Following Proposition~\ref{app:polynomials},
any monomial can be written as $m= \sum_{\oneshalf} c_{m,\oneshalf} a^{(\alpha {-} \oneshalf,\beta {-} \oneshalf)}$
 for nonnegative integer vectors $\oneshalf \in \Np{n}$.
The proof of (a) follows now as 
\begin{align*}
[a_k^\dagger a_k, m] = [a_k^\dagger a_k, \sum_{\oneshalf} c_{m,\oneshalf} a^{(\alpha {-} \oneshalf,\beta {-} \oneshalf)}]
= \sum_{\oneshalf} (\alpha_k-\oneshalf_k {-} \beta_k+\oneshalf_k) c_{m,\oneshalf} a^{(\alpha {-} \oneshalf,\beta {-} \oneshalf)}=
\sum_{\oneshalf} (\alpha_k {-} \beta_k) c_{m,\oneshalf} a^{(\alpha {-} \oneshalf,\beta {-} \oneshalf)} = (\alpha_k {-} \beta_k) m. 
\end{align*}
It is clear that
$\mdeg(m^\dagger)=(\beta,\alpha)$ and (b) is immediate from (a).
\end{proof}

We extend the results from Lemma~\ref{helper:lemma} to the skew-hermitian Weyl algebra:

\begin{theorem}\label{important:one:element:and:diagonal:implies:other:two:mode:new}
Consider a Lie subalgebra $\mathfrak{g}$ of the skew-hermitian Weyl algebra $\hat{A}_n$ which contains
$g_{\sigma}(m)$ for some monomial $m \in A_n$,
with $\mdeg(m)=\gamma=(\alpha,\beta)$, as well as
$i a^\dag_k a_k\in\mathfrak{g}$ for all $k\in \{1,\ldots,n\}$. 
Then: (a1)~$[i a^\dag_k a_k, g_\sigma(m)]=\sigma(\alpha_k {-} \beta_k) g_{-\sigma}(m)$.
(a2)~$[i a^\dag_k a_k,[i a^\dag_k a_k, g_\sigma(m)]]=-(\alpha_k {-} \beta_k)^2 g_\sigma(m)$.
(b)~If there is a $k$ such that $\alpha_k \neq \beta_k$, then $[i a^\dag_k a_k, g_\sigma(m)]\neq 0$ and $g_{-\sigma}(m)\in\mathfrak{g}$.
\end{theorem}
\begin{proof}
Applying Lemma~\ref{helper:lemma}, we obtain 
\begin{align*}
[i a^\dag_k a_k, g_+(m)]
&= [i a^\dag_k a_k, i(m^\dagger {+} m)]
= i^2 ( [ a^\dag_k a_k, m^\dagger] + [ a^\dag_k a_k, m ])
= (\alpha_k {-} \beta_k) m^\dagger  + (\beta_k {-} \alpha_k) m \\
&= (\alpha_k {-} \beta_k) (m^\dagger {-}m)
=  (\alpha_k {-} \beta_k) g_-(m).
\end{align*}
This verifies the case of $\sigma=+$ and the case of $\sigma=-$ 
can be immediately repeated as well giving $[i a^\dag_k a_k, g_-(m)]=-(\alpha_k {-} \beta_k) g_+(m)$. This proves (a1). The claim (a2) can be immediately obtained by applying twice the result from (a1) as
\begin{align*}
[i a^\dag_k a_k,[i a^\dag_k a_k, g_\sigma(m)]]=\sigma(\alpha_k {-} \beta_k)[i a^\dag_k a_k, g_{-\sigma}(m)]=-\sigma^2(\alpha_k {-} \beta_k)^2g_\sigma(m)=-(\alpha_k {-} \beta_k)^2g_\sigma(m).
\end{align*}
Finally, it is clear that we can compute
$[i a^\dag_k a_k, g_\sigma(m)]=\sigma(\alpha_k {-} \beta_k) g_{-\sigma}(m)$. Since the commutator of two elements
of an algebra is also in the algebra, if $\alpha_k \neq \beta_k$ we then conclude that
$[i a^\dag_k a_k, g_\sigma(m)]\neq0$ and therefore $g_{-\sigma}(m)\in\mathfrak{g}$.
\end{proof}

\subsection{Complementary operators\label{subsection:complementary}}
At this point we introduce the following notion of \emph{complementary operators}:

\begin{definition}\label{definition:complementary:elements}
Let $x_+,x_-$ be two elements of $\hat{A}_n$ not in $\hat{A}_n^0\oplus\hat{A}_n^=$.
We say that $x_+,x_-$ are \emph{complementary} if for each $k$ either:
(a)  $[i a^\dag_ka_k,x_\sigma]=0$ or (b) $[i a^\dag_ka_k,x_\sigma]=\sigma\,\mu_k\, x_{-\sigma}$
for appropriate $\mu_k\neq0$, where $\sigma\in\{+,-\}$, and there is at least one index $k$
that satisfies (b). The coefficients $\mu_k$ are called the \emph{reciprocity constants}.
\end{definition}

Note that for any arbitrary two $x,x'\notin \hat{A}_n^0\oplus\hat{A}_n^=$,
Theorem~\ref{important:one:element:and:diagonal:implies:other:two:mode:new}
guarantees that there are always exist $k_1,k_1'$ such that $[i a^\dag_{k_1} a_{k_1},x]\neq0$
and $[i a^\dag_{k_1'} a_{k_1'},x']\neq0$. Furthermore, there can be additional finite index sets
$\mathcal{J}:=\{k_2,k_3,\ldots\}$ and $\mathcal{J}':=\{k_2',k_3',\ldots\}$ such that $[i a^\dag_{k_p} a_{k_p},x]\neq0$ and
$[i a^\dag_{k_q'} a_{k_q'},x']\neq0$ for all $k_p\in\mathcal{J}$ and $k_q'\in\mathcal{J}'$. However it can occur that
$\mathcal{J}\cap\mathcal{J}'=\emptyset$, and thus there is no $k$ for which $[i a^\dag_k a_k,x]\neq0$
and $[i a^\dag_{k} a_{k},x']\neq0$ simultaneously. This is one of the required conditions for $x,x'$ to be complementary.

For complementary operators, 
Theorem~\ref{important:one:element:and:diagonal:implies:other:two:mode:new}
immediately implies that all basis elements $g_+(m)$ and $g_-(m)$
are always complementary if they are not in
$\hat{A}_n^0\oplus\hat{A}_n^=$.
This special case leads to an important observation.
Theorem~\ref{important:one:element:and:diagonal:implies:other:two:mode:new}
also tells us that $[i a^\dag_k a_k,[i a^\dag_k a_k, g_\sigma(m)]]=-(\alpha_k{-}\beta_k)^2g_\sigma(m)$.
It can therefore be immediately verified that the following holds:
\begin{lemma}
Let $x_+,x_-$ be two complementary operators in $\hat{A}_n$. Then, for each $k$, either:
(a) $[i a^\dag_k a_k, x_\sigma]=0$, or (b) $[i a^\dag_k a_k,[i a^\dag_k a_k, x_\sigma]]=-\mu_k^2\,x_\sigma$
for appropriate $\mu_k\neq0$, and there is at least one $k$ that satisfies (b).
\end{lemma}
\begin{proof}
The proof is a direct verification of the claim by using the two commutators that appear in Definition~\ref{definition:complementary:elements}. 
\end{proof}

We now proceed to study the commutator of complementary operators in general
and, in particular, for complementary elements $g_+(m)$ and $g_-(m)$. 

\begin{lemma}\label{pre:chain:lemma}
Consider the complementary operators $x_+,x_- \in \hat{A}_n$. Then (a)~$[i a^\dag_k a_k, [x_+,x_-]]=0$ for all
$k$ and (b)~$[x_+,x_-]=\sum_p C_p g_+^{(\tilde{\alpha}_p,\tilde{\alpha}_p)}$ for appropriate coefficients
$C_p$ and basis operators $g_+^{(\tilde{\alpha}_p,\tilde{\alpha}_p)}$. 
\end{lemma}

\begin{proof}
By definition of complementarity we have at least one $k$ such that
$[i a^\dag_k a_k,x_\sigma]=\sigma\,\mu_k\,x_{-\sigma}$ for appropriate $\mu_k\neq0$.
Then, we  use Jacobi and compute 
\begin{align*}
[i a^\dag_k a_k, [x_+,x_-]]=-[x_+, [x_-,i a^\dag_k a_k]]-[x_-, [i a^\dag_k a_k,x_+]]=-\mu_k[x_+,x_+]-\mu_k[x_-,x_-]\equiv0.
\end{align*}
This proves (a).
We then note that we can always write $[x_+,x_-]=\sum_{p,\sigma_p} C_{p,\sigma_p} g_{\sigma_p}^{(\alpha_p,\beta_p)}$
for appropriate coefficients $C_{p,\sigma_p}$, where the elements $g_{\sigma_p}^{(\alpha_p,\beta_p)}$ in the expansion are
all linearly independent. However, we have 
\begin{align*}
0=[i a^\dag_k a_k, [x_+,x_-]]=\sum_{p,\sigma_p} C_{p,\sigma_p} \bigl[i a^\dag_k a_k, g_{\sigma_p}^{(\alpha_p,\beta_p)}\bigr]
=\sum_{p,\sigma_p} C_{p,\sigma_p} \sigma_p (\alpha^{(p)}_k-\beta^{(p)}_k) g_{-\sigma_p}^{(\alpha_p,\beta_p)},
\end{align*}
where we have used the bi-linearity of the commutator and also
Theorem~\ref{important:one:element:and:diagonal:implies:other:two:mode:new}.
Since all of the $g_{-\sigma_p}^{(\alpha_p,\beta_p)}$ are linearly independent it follows that $\alpha^{(p)}_k=\beta^{(p)}_k$
must be satisfied for all $p,\sigma_p$. However, this is true for all $k$, and thus the only elements allowed in the expansion of
the commutator are of the form $g_+^{(\tilde{\alpha}_p,\tilde{\alpha}_p)}$. It is worth recalling that that
$g_-^{(\tilde{\alpha}_p,\tilde{\alpha}_p)}=0$ identically. This proves (b).
\end{proof}

Notice that, since the operators $g_+(m),g_-(m)$ with $\mdeg(m)=\gamma=(\alpha,\beta)$ are complementary for $\alpha\neq\beta$ it follows that Lemma~\ref{pre:chain:lemma} applies and therefore (a) $[i a^\dag_k a_k, [g_{+}(m),g_{-}(m)]]=0$ for all
$k$; and (b) $[g_{+}(m),g_{-}(m)]=\sum_p C_p g_+^{(\tilde{\alpha}_p,\tilde{\alpha}_p)}$ for appropriate coefficients
$C_p$ and basis operators $g_+^{(\tilde{\alpha}_p,\tilde{\alpha}_p)}$.

\begin{lemma}\label{degree:commutator:conjugate:basis:operators}
Consider the polynomials $g_{+}(m), g_{-}(m) \in \hat{A}_n$ for a monomial $m \in  A_n$
with $\mdeg(m)=(\alpha,\beta)$. Then (i)~$g_{+}(m)$ and $g_{-}(m)$ commute if and only if $\alpha=\beta$; 
(ii)~If $g_{+}(m)$ and $g_{-}(m)$ do not commute, then $\deg([g_{+}(m),g_{-}(m)])=\deg(g_{+}(m))+\deg(g_{-}(m))-2=:d$;
(iii)~$[g_{+}(m),g_{-}(m)]\upto{d} -2i\sum_k(\alpha_k^2{-}\beta_k^2)a^{(\alpha+\beta-\oneshalf_k,\alpha+\beta-\oneshalf_k)}$.
\end{lemma}

\begin{proof}
If $\alpha=\beta$, then Proposition~\ref{app:polynomials:same:order} implies that $m = m^\dagger$, $g_{-\sigma}(m)=0$, and 
$[g_{\sigma}(m), g_{-\sigma}(m)]=0$. For the converse, we compute $0=[g_{+}(m), g_{-}(m)]=i [m^\dagger {+} m, m^\dagger {-} m] = -2i [m^\dagger, m]$.
Corollary~\ref{theorem:commutator:monomials:one:diagonal}(b1) shows that $\beta_j^2 = \alpha_j^2$ for all $j \in \{1, \ldots, n\}$.
Since all values of $\alpha$ and $\beta$ are non-negative integers, we obtain $\alpha=\beta$ which proves (i).
To prove (ii), we know that $[m^\dagger, m]\neq 0$. We conclude that $\alpha\neq \beta$
via (i) and, in particular, $m^\dagger \neq m$ via Proposition~\ref{app:polynomials:same:order}.
Theorem~\ref{theorem:commutator:monomials} implies that $\deg([m^\dagger,m])=\deg(m^\dagger)+\deg(m)-2$.
This immediately verifies (ii).
We obtain (iii) via  
	\begin{align*}
	[g_{+}(m), g_{-}(m)]=-2i [m^\dagger, m]\upto{d} 2i[a^{(\alpha,\beta)},a^{(\beta,\alpha)}]
	\upto{d} -2i\sum_k(\alpha_k^2-\beta_k^2)a^{(\alpha+\beta-\oneshalf_k,\alpha+\beta-\oneshalf_k)}.
	\end{align*}
\end{proof}

\begin{corollary}\label{commutator:comutator:corollary}
Given the polynomials $g_{+}(m), g_{-}(m) \in \hat{A}_n$ for a monomial $m \in  A_n$
with $\mdeg(m)=(\alpha,\beta)$, we assume that
$\alpha\neq\beta$. It follows that the commutator $[g_{+}(m),g_{-}(m)]$ has degree 
$\deg(g_{+}(m))+\deg(g_{-}(m))-2$ and its expansion into basis elements contains one basis element of this degree with a nonzero coefficient.
\end{corollary}

\begin{proof}
Lemma~\ref{degree:commutator:conjugate:basis:operators}(i)
shows that $[g_{+}(m), g_{-}(m)]=0$ if and only if $\alpha=\beta$. The condition of the corollary thus implies that $[g_{+}(m), g_{-}(m)]\neq0$. 
Then Lemma~\ref{degree:commutator:conjugate:basis:operators}(ii) concludes the proof.
\end{proof}

\subsection{Generators that commute with all other generators}
We are interested in characterizing 
a subset of the generators that commute with all other generators 
and are consequently contained in the center of the generated Lie algebra.

\begin{definition}\label{center:basis:definition}
Consider a finite set 
$\mathcal{G}:=\{ g_{\sigma_p}^{\gamma_p} \text{ for } p \in \mathcal{R} \text{ with } \abs{\mathcal{R}} < \infty \text{ and } \sigma_p\in \{+,-\} \}
\subseteq\hat{A}_n$ of generators and define 
$\mathfrak{g}:=\lie{\mathcal{G}}$ with $\gc$ its centre.
We define $\gc_{\mathcal{G}}$ as the maximal subset of $\mathcal{G}$
with $[\gc_{\mathcal{G}},\mathcal{G}]=0$.
\end{definition}

We can now proceed, and prove the following:

\begin{lemma}\label{center:basis:elements}
Let $\mathfrak{g}:=\lie{\mathcal{G}}$ be an algebra generated by the finite set 
$\mathcal{G}:=\{ g_{\sigma_p}^{\gamma_p} \text{ for } p \in \mathcal{R} \text{ with } \abs{\mathcal{R}} < \infty \text{ and } \sigma_p\in \{+,-\} \}
\subseteq\hat{A}_n$ of generators, and $\gc$ is the centre of $\g$.
 Then, $\gc_{\mathcal{G}}\subseteq\gc$.
\end{lemma}

\begin{proof}
We know that each element $\hat{g}_{\hat{\sigma}_k}^{\hat{\gamma}_k}\in \gc_{\mathcal{G}}$ commutes with
all elements in $\mathcal{G}$, and therefore with any linear combination of these elements. We can use the Jacobi identity to compute
$[\hat{g}_{\hat{\sigma}_k}^{\hat{\gamma}_k},
[g_{\sigma_p}^{\gamma_p},g_{\sigma_{p'}}^{\gamma_{p'}}]]=
-[g_{\sigma_p}^{\gamma_p},[g_{\sigma_{p'}}^{\gamma_{p'}},
\hat{g}_{\hat{\sigma}_k}^{\hat{\gamma}_k}]]-[g_{\sigma_{p'}}^{\gamma_{p'}},
[\hat{g}_{\hat{\sigma}_k}^{\hat{\gamma}_k},g_{\sigma_p}^{\gamma_p}]]=0$.
Therefore, $\hat{g}_{\hat{\sigma}_k}^{\hat{\gamma}_k}$ commutes with any commutator
of operators $g_{\sigma_p}^{\gamma_p},g_{\sigma_{p'}}^{\gamma_{p'}}\in\mathcal{G}$, which is itself an element of $\g$.
By induction, we can prove that $\gc_{\mathcal{G}}$ commutes with any element in $\g$, i.e., it is contained in $\gc$.
\end{proof}

The usefulness of the concept of $\gc_{\mathcal{G}}$ will be come evident in the proof of our main result
(see Theorem~\ref{main:theorem_final}). The idea is
that this set will allow us to provide one of the criteria for the verification of the finiteness of the subalgebra that can be verified
using the elements of the basis alone. 

\begin{lemma}\label{algebra:center:elements}
Let $\mathcal{G}:=\{ i{a}^\dag_k {a}_k   \text{ for } k\in \{1,\ldots,n\} \text{ and } g_{\sigma_p}^{(\alpha_p,\beta_p)} \text{ for } p \in \mathcal{R}
\text{ with } \abs{\mathcal{R}} < \infty \text{ and } \sigma_p\in \{+,-\} \}$ generate the Lie subalgebra $\mathfrak{g}=\lie{ \mathcal{G}}$ of
$\hat{A}_n$. If $g_{\sigma_p}^{(\alpha_p,\beta_p)}$ is contained in the center $\gc$ of $\g$, then $\alpha_p=\beta_p$, and therefore $g_{\sigma_p}^{(\alpha_p,\beta_p)}\in\hat{A}_n^=\cup\hat{A}^0_n$ with 
$g_{\sigma_p}^{(\alpha_p,\beta_p)}= g_+^{(\tilde{\alpha}_p,\tilde{\alpha}_p)}$.
\end{lemma}

\begin{proof}
Let us assume that $g_{\sigma_p}^{(\alpha_p,\beta_p)}\in\mathcal{G}$ is contained in the center $\gc$ of $\g$. In particular we must have 
$[g_{\sigma_p}^{(\alpha_p,\beta_p)},i{a}^\dag_k{a}_k]=0$ for all $k$.
An explicit computation from Theorem~\ref{important:one:element:and:diagonal:implies:other:two:mode:new} shows then
that 
$$[g_{\sigma_p}(a^{(\alpha_p,\beta_p)}),i{a}^\dag_k{a}_k]=((\alpha_p)_k {-} (\beta_p)_k)g_{-\sigma_p}(a^{(\alpha_p,\beta_p)})=0.$$
Consequently, $(\alpha_p)_k=(\beta_p)_k$ for all $k$ and $\alpha_p=\beta_p$, and thus $g_{\sigma_p}^{(\alpha_p,\beta_p)}\in\hat{A}_n^=\cup\hat{A}^0_n$. This means that 
$g_{\sigma_p}^{(\alpha_p,\beta_p)}= g_+^{(\tilde{\alpha}_p,\tilde{\alpha}_p)}$ by definition.
\end{proof}

\subsection{General commutators of generators}
We move to the final expressions for this section, which tackle the generic form of the commutator of two basis elements $g_\sigma^\gamma$ and $g_{\sigma'}^{\gamma'}$. They can be used as a omnibus lemma for calculations in the subsequent parts of this work.
Let us first introduce useful definitions that will be crucial for the following section. We recall that $\gamma^{\dagger} \equiv (\alpha,\beta)^{\dagger} := (\beta,\alpha)$.
\begin{definition}\label{definition:useful:notation}
We introduce the functions
\begin{gather*}
\Theta(\gamma) := 
\begin{cases}
\gamma & \text{if $\gamma \geq \gamma^{\dagger}$}\\
\gamma^{\dagger} & \text{otherwise}
\end{cases}
\qquad
\text{ and }
\qquad
E(\gamma) := 
\begin{cases}
+1 & \text{if $\gamma \geq \gamma^{\dagger}$}\\
-1 & \text{otherwise}
\end{cases}
\end{gather*}
which enable us to uniformly treat a few special cases.
Furthermore, we introduce the multisets
$S=
\multiset{\multy{1}{S_1}, \ldots, \multy{n}{S_n}}$
containing elements $k\in\{1,\ldots,n\}$ that will describe certain combinatoric aspects,
where $\multy{k}{S_k}$ denotes that $k$ appears with multiplicity $S_k$ in $S$ and
$k = \multy{k}{1}$. Moreover, $\abs{S}=\sum_{k=1}^n S_k$.
One example is
$\multiset{\multy{1}{3}, 2, \multy{3}{2}}$
which contains
$1$ three times, $2$ once, and $3$ twice, and all other elements do not
appear in $S$ (i.e.\ have multiplicity zero). The empty multiset is denoted by $\emptymultiset$
and, e.g., $\multiset{k,k}=\multiset{\multy{k}{2}}$. The union $S{\cup}R$ of two multisets $S$ and $R$
is specified via $(S{\cup}R)_k=S_k+R_k$.
We extend the notation $\ones_k$ to
multisets $S$ such that $\ones_S = \ones(S) := \sum_{k=1}^n S_k\, \ones_k$.
\end{definition}

This allows us to immediately prove the following:

\begin{lemma}\label{lem:omnibus}
Let $\gamma = (\alpha,\beta), \lambda = (\mu,\nu) \in \Np{2n}$ with $\alpha \geq \beta$ and $\mu \geq \nu$,
and $\sigma \in \{+,-\}$.
Moreover, $u,v \in \{1,\ldots,n\}$, $R$ and $S$ are multisets,
$d=\abs{\gamma}+\abs{\lambda}-2$, $e=\abs{\gamma}+\abs{\lambda}-2\abs{R}-2\abs{S}-2$,
and $f=2\abs{\gamma}-2\abs{R}-2\abs{S}-2$.
All sums $\sum_v c_v g_{\sigma}^{\lambda(v)}$ below are restricted to terms with  $\lambda(v) \in \Np{2n}$.
We obtain the following relations:
\begingroup
\allowdisplaybreaks
\begin{flalign*}
\begin{alignedat}{10}
& \text{(a)}\; && [g_{+}^{\gamma}, g_{-}^{\lambda}] && \upto{d}\; && 
+
[
\sum\nolimits_{v} (\alpha_v \nu_v {-} \beta_v \mu_v)\, 
g_{+}^{\gamma+\lambda-\ones_v}
]
+
[
\sum\nolimits_{v} (\beta_v \nu_v {-} \alpha_v \mu_v)\,
g_{+}^{\Theta(\gamma+\lambda^{\dagger})-\ones_v}
]; &&
\\
& \text{(b)}\; && [g_{+}^{\gamma}, g_{+}^{\lambda}] && \upto{d}\; &&
- [ \sum\nolimits_{v} (\alpha_v \nu_v {-} \beta_v \mu_v)\,
g_{-}^{\gamma+\lambda-\ones_v}
]
+
[
\sum\nolimits_{v} (\beta_v \nu_v {-} \alpha_v \mu_v)\,
E(\gamma{+}\lambda^{\dagger})\,
g_{-}^{\Theta(\gamma+\lambda^{\dagger})-\ones_v}
]; &&
\\
& \text{(c)}\; && [g_{-}^{\gamma}, g_{+}^{\lambda}] && \upto{d}\; &&
+
[
\sum\nolimits_{v} (\alpha_v \nu_v {-} \beta_v \mu_v)\,
g_{+}^{\gamma+\lambda-\ones_v}
]
-
[
\sum\nolimits_{v} (\beta_v \nu_v {-} \alpha_v \mu_v)\,
g_{+}^{\Theta(\gamma+\lambda^{\dagger})-\ones_v}
]; &&
\\
& \text{(d)}\; && [g_{-}^{\gamma}, g_{-}^{\lambda}] && \upto{d}\; &&
+ [ \sum\nolimits_{v} (\alpha_v \nu_v {-} \beta_v \mu_v)\,
g_{-}^{\gamma+\lambda-\ones_v}
]
+
[
\sum\nolimits_{v} (\beta_v \nu_v {-} \alpha_v \mu_v)\,
E(\gamma{+}\lambda^{\dagger})\,
g_{-}^{\Theta(\gamma+\lambda^{\dagger})-\ones_v}
]; &&
\\
& \text{(e)}\; && [g_{+}^{\gamma-\ones(R)}, g_{-}^{\lambda-\ones(S)}] && \upto{e}\; &&
+
\{
\sum\nolimits_{v} [(\alpha_v{-}R_v) (\nu_v{-}S_v) {-} (\beta_v{-}R_v) (\mu_v{-}S_v)]\,
g_{+}^{\gamma+\lambda-\ones(R\cup S \cup \multisetnarrow{v})}
\} &&
\\[-1mm]
&&&&&&&
+
\{
\sum\nolimits_{v} [(\beta_v{-}R_v) (\nu_v{-}S_v) {-} (\alpha_v{-}R_v) (\mu_v{-}S_v)]\,
g_{+}^{\Theta(\gamma+\lambda^{\dagger})-\ones(R\cup S \cup \multisetnarrow{v})}
\};
\\
& \text{(f)}\; && [g_{+}^{\gamma-\ones(R)}, g_{+}^{\lambda-\ones(S)}] && \upto{e}\; &&
-
\{
\sum\nolimits_{v} [(\alpha_v{-}R_v) (\nu_v{-}S_v) {-} (\beta_v{-}R_v) (\mu_v{-}S_v)]\,
g_{-}^{\gamma+\lambda-\ones(R\cup S \cup \multisetnarrow{v})}
\} &&
\\[-1mm]
&&&&&&&
+
\{
\sum\nolimits_{v} [(\beta_v{-}R_v) (\nu_v{-}S_v) {-} (\alpha_v{-}R_v) (\mu_v{-}S_v)]\,
E(\gamma{+}\lambda^{\dagger})\,
g_{-}^{\Theta(\gamma+\lambda^{\dagger})-\ones(R\cup S \cup \multisetnarrow{v})}
\}; &&
\\
& \text{(g)}\; && [g_{-}^{\gamma-\ones(R)}, g_{+}^{\lambda-\ones(S)}] && \upto{e}\; &&
+
\{
\sum\nolimits_{v} [(\alpha_v{-}R_v) (\nu_v{-}S_v) {-} (\beta_v{-}R_v) (\mu_v{-}S_v)]\,
g_{+}^{\gamma+\lambda-\ones(R\cup S \cup \multisetnarrow{v})}
\} &&
\\[-1mm]
&&&&&&&
-
\{
\sum\nolimits_{v} [(\beta_v{-}R_v) (\nu_v{-}S_v) {-} (\alpha_v{-}R_v) (\mu_v{-}S_v)]\,
g_{+}^{\Theta(\gamma+\lambda^{\dagger})-\ones(R\cup S \cup \multisetnarrow{v})}
\}; &&
\\
& \text{(h)}\; && [g_{+}^{\gamma-\ones(R)}, g_{-}^{\gamma-\ones(S)}] && \upto{f}\; &&
+
\{
\sum\nolimits_{v} [(\alpha_v{-}R_v) (\beta_v{-}S_v) {-} (\beta_v{-}R_v) (\alpha_v{-}S_v)]\,
g_{+}^{2\gamma-\ones(R\cup S \cup \multisetnarrow{v})}
\} &&
\\[-1mm]
&&&&&&&
+
\{
\sum\nolimits_{v} [(\beta_v{-}R_v) (\beta_v{-}S_v) {-} (\alpha_v{-}R_v) (\alpha_v{-}S_v)]\,
g_{+}^{\gamma+\gamma^{\dagger}-\ones(R\cup S \cup \multisetnarrow{v})}
\}; &&
\\
& \text{(i)}\; && [g_{-}^{\gamma-\ones(R)}, g_{+}^{\gamma-\ones(S)}] && \upto{f}\; &&
+
\{
\sum\nolimits_{v} [(\alpha_v{-}R_v) (\beta_v{-}S_v) {-} (\beta_v{-}R_v) (\alpha_v{-}S_v)]\,
g_{+}^{2\gamma-\ones(R\cup S \cup \multisetnarrow{v})}
\} &&
\\[-1mm]
&&&&&&&
-
\{
\sum\nolimits_{v} [(\beta_v{-}R_v) (\beta_v{-}S_v) {-} (\alpha_v{-}R_v) (\alpha_v{-}S_v)]\,
g_{+}^{\gamma+\gamma^{\dagger}-\ones(R\cup S \cup \multisetnarrow{v})}
\}.
\end{alignedat}
&&
\end{flalign*}
\endgroup
\end{lemma}
\begin{proof}
Here we prove (b) explicitly, while (a), (c), and (d)  are similar. All other cases are immediate adaptations or consequences of (a)-(c) presented.
Let us compute $[g_{+}^{\gamma}, g_{+}^{\lambda}]$ and the proof of (b) is completed by
\begin{align*}
 & [g_{+}^{\gamma}, g_{+}^{\lambda}]  =[g_{+}^{\gamma}g_{+}^{\lambda}-g_{+}^{\lambda}g_{+}^{\gamma}]
 =-[a^{\gamma^\dag},a^{\lambda^\dag}]+[a^{\gamma^\dag},a^{\lambda^\dag}]^\dag
 +[a^{\gamma^\dag},a^{\lambda}]-[a^{\gamma^\dag},a^{\lambda}]^\dag\\
&\upto{d}-[\sum\nolimits_{v}(\alpha_v\nu_v{-}\beta_v\mu_v)a^{\gamma^\dag+\lambda^\dag-\tau_v}
-\sum\nolimits_{v}(\alpha_v\nu_v{-}\beta_v\mu_v)a^{\gamma+\lambda-\tau_v}]
 -[\sum\nolimits_{v}(\alpha_v\mu_v{-}\beta_v\nu_v)a^{\gamma^\dag+\lambda-\tau_v}
 -\sum\nolimits_{v}(\alpha_v\mu_v{-}\beta_v\nu_v)a^{\gamma+\lambda^\dag-\tau_v}]\\
& \upto{d} -\sum\nolimits_{v}(\alpha_v\nu_v{-}\beta_v\mu_v)g_-^{\gamma+\lambda-\tau_v}
+E(\gamma{+}\lambda^\dag) \sum\nolimits_{v}(\beta_v\nu_v{-}\alpha_v\mu_v)g_{-}^{\Theta(\gamma+\lambda^\dag)-\tau_v}. \qedhere
\end{align*}
\end{proof}

\section{Commutator Chains in the skew-hermitian Weyl algebra\label{commutator:chain:section}}
The skew-hermitian Weyl algebra $\hat{A}_n$ is infinite dimensional.
We characterize $\hat{A}_n$ and its infinitely many elements by studying
sequences of 
elements.
We introduce two important sequences of elements in the skew-hermitian 
Weyl algebra that exemplify its structure via nested commutators given by
\emph{Commutator Chains of Type I  and II}. Both sequences
are constructed starting from two initial elements that are related in a meaningful way, 
which we call \emph{seed elements}, and subsequent elements are obtained via nested
commutators. In this sense, elements in these chains are
concatenated with previous prior ones which motives
the terminology ``chain.''  

The degree of the elements in a chain plays a key role in our work.
These chains will allow us to determine if the generated Lie algebra is finite-dimensional or not.
First, chains build Lie-algebra elements from Lie-algebra elements via nested commutators.
Therefore, if the seed
elements belong to a Lie algebra $\g$, the whole chain belongs to $\g$. We will focus on distinguishing
between whether
the elements of a chain have an ever-increasing degree or the degree is
bounded by a finite value for all elements in the chain. The key aspect will be to obtain a
set of conditions distinguish these cases.
We anticipate that each chain element admits a natural complementary one as defined in
Definition~\ref{definition:complementary:elements}. Then, 
any chain is naturally associated to a complementary one.

\subsection{Commutator Chains of Type I}
We start by introducing the first chain of elements. 

\begin{definition}\label{Commutator:Chain:Type:I}
Consider the (seed)
elements $g_\sigma^{\gamma}, g_+^{\tilde{\gamma}} \in \hat{A}_n$ 
with $\gamma=(\alpha,\beta)$, $\tilde{\gamma}=(\tilde{\alpha},\tilde{\alpha})$, and $\sigma \in\{+,-\}$.
The elements
$$C_\sigma^I(\gamma,\tilde{\gamma}):=\{z^{(\ell)}_\sigma\in\hat{A}_n:\,\,z^{(\ell+1)}_\sigma
:=[g_+^{\tilde{\gamma}},z^{(\ell)}_\sigma]\,\,\text{for all integers}\,\,\ell\geq0,\,\,\text{with}\,\,z^{(0)}_\sigma:=g_\sigma^{\gamma}\}$$
form a \emph{Commutator Chain of Type I} generated by $g_\sigma^{\gamma}$ and $g_+^{\tilde{\gamma}}$.
\end{definition}
In order to visualize how the Chain $C_\sigma^I(\gamma,\tilde{\gamma})$ is constructed we refer to the pictorial representation 
in Figure~\ref{Figure:Simple_Commutator_Chain}.
We now proceed to collect and prove important properties of this chain. Recall that we are particularly interested in obtaining information
about the degree of its elements, and this will be a priority in the following discussion.

\begin{figure}[t]
\includegraphics[width=0.9\linewidth]{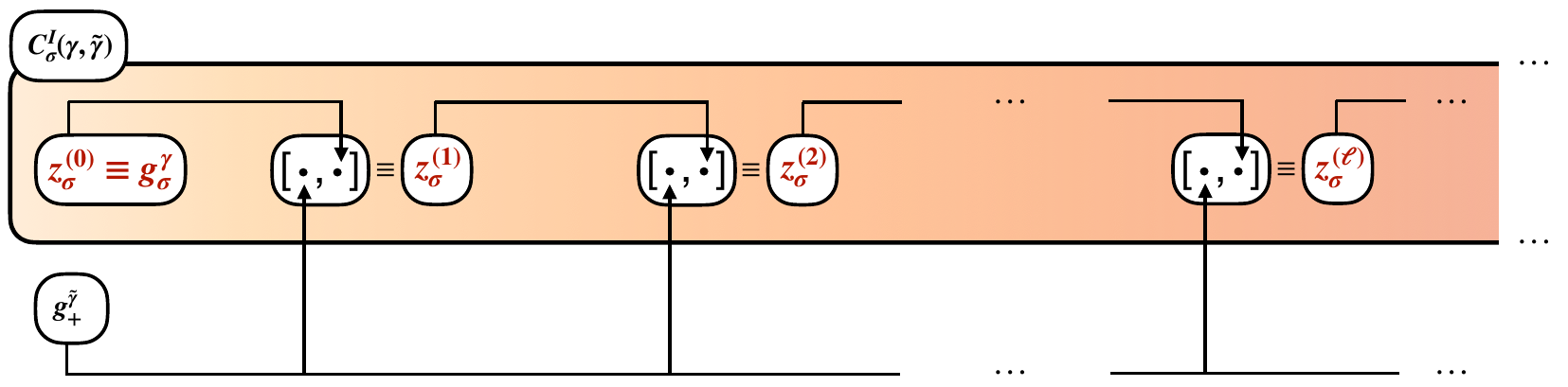}
\caption{\textbf{Commutator Chains of Type I}: Visualization of how a Commutator Chain of
Type I, which is denoted by $C_\sigma^I(\gamma,\tilde{\gamma})$ and has elements 
$z^{(\ell)}_\sigma$ (highlighted in red), is constructed starting from $z^{(0)}_\sigma:=g_\sigma^{\gamma}$ and
$g_+^{\tilde{\gamma}}$ with $\gamma=(\alpha,\beta)$,
$\tilde{\gamma}=(\tilde{\alpha},\tilde{\alpha})$,
and $\sigma \in\{+,-\}$.
The arrows indicate the location of the respective element in a commutator.
\label{Figure:Simple_Commutator_Chain}}
\end{figure}

\begin{lemma}\label{max:length:chain:lemma}
Let $C_\sigma^I(\gamma,\tilde{\gamma})$ be a
Commutator Chain of Type I with elements $z^{(\ell)}_\sigma$ and generated by $g_\sigma^{\gamma}, g_+^{\tilde{\gamma}} \in \hat{A}_n$, where $\gamma=(\alpha,\beta)$,
$\tilde{\gamma}=(\tilde{\alpha},\tilde{\alpha})$,
and $\sigma \in\{+,-\}$.
Then,
$\deg(z^{(\ell)}_\sigma)\leq2 \ell(|\tilde{\alpha}|{-}1)+|\alpha|+|\beta|$.
\end{lemma}

\begin{proof}
We start by noting, following Theorem~\ref{theorem:commutator:monomials}, that the commutator of two skew-hermitian operators
$x,x'\in\hat{A}_n$ always enjoys the property $\deg([x,x'])\leq\deg(x)+\deg(x')-2$. Therefore it is immediate to see that
$\deg(z^{(1)}_\sigma)\leq2\abs{\tilde{\alpha}}+\abs{\alpha}+\abs{\beta}-2$, since it is obtained via a single commutator. However,
$2\abs{\tilde{\alpha}}+\abs{\alpha}+\abs{\beta}-2=2\cdot1\cdot(\abs{\tilde{\alpha}}{-}1)+\abs{\alpha}+\abs{\beta}$,
which supports the claim for the case of $\ell=1$.
We now apply an induction argument with respect to the generic index $\ell$. We have already proven that the claim is true for $\ell=1$.
We then assume that
$\deg(z^{(\ell)}_\sigma)\leq2 \ell(\abs{\tilde{\alpha}}{-}1)+\abs{\alpha}+\abs{\beta}$. Therefore,
$\deg(z^{(\ell+1)}_\sigma)=\deg([g_+^{\tilde{\gamma}},z^{(\ell)}_\sigma])\leq2 \ell(\abs{\tilde{\alpha}}{-}1)
+\abs{\alpha}+\abs{\beta}+2\abs{\tilde{\alpha}}-2$,
and thus $\deg(z^{(\ell+1)}_\sigma)\leq2 (\ell{+}1)(\abs{\tilde{\alpha}}{-}1)+\abs{\alpha}+\abs{\beta}$.
\end{proof}

\begin{lemma}\label{small:algebra:property}
Let $C_\sigma^I(\gamma,\tilde{\gamma})$ be a
Commutator Chain of Type I with elements $z^{(\ell)}_\sigma$ and generated by $g_\sigma^{\gamma}, g_+^{\tilde{\gamma}} \in \hat{A}_n$, where $\gamma=(\alpha,\beta)$,
$\tilde{\gamma}=(\tilde{\alpha},\tilde{\alpha})$,
and $\sigma \in\{+,-\}$. 
Then: (a) $[i a^\dag_k a_k,z^{(\ell)}_\sigma]=\sigma(\alpha_k{-}\beta_k)z^{(\ell)}_{-\sigma}$ for all $k$; (b) If there exists at least one $k$ such that $\alpha_k\neq\beta_k$ then   $z^{(\ell)}_+$ and $z^{(\ell)}_-$ are complementary operators for all $\ell$.
\end{lemma}

\begin{proof}
The proof follows by induction from a straightforward application of Jacobi's identity. 
For the case of $\ell=1$,
\begin{align*}
[i a^\dag_k a_k,z^{(1)}_\sigma] &=[i a^\dag_k a_k,[g_+^{\tilde{\gamma}},g_\sigma^{\gamma}]]
=-[g_+^{\tilde{\gamma}},[g_\sigma^{\gamma},i a^\dag_k a_k]]-[g_\sigma^{\gamma},[i a^\dag_k a_k,g_+^{\tilde{\gamma}}]]\\
&=[g_+^{\tilde{\gamma}},[i a^\dag_k a_k,g_\sigma^{\gamma}]]
=[g_+^{\tilde{\gamma}},\sigma(\alpha_k{-}\beta_k)g_{-\sigma}^{\gamma}]
=\sigma(\alpha_k{-}\beta_k)z^{(1)}_{-\sigma}
\end{align*}
 due to Corollary~\ref{theorem:commutator:monomials:one:diagonal}.
We now assume $[i a^\dag_k a_k,z^{(\ell)}_\sigma]=\sigma(\alpha_k{-}\beta_k)z^{(\ell)}_{-\sigma}$
for $\ell$ and set to prove it for $\ell{+}1$. We have 
\begin{align*}
[i a^\dag_k a_k,z^{(\ell+1)}_\sigma]&=[i a^\dag_k a_k,[g_+^{\tilde{\gamma}},z^{(\ell)}_\sigma]]
=-[g_+^{\tilde{\gamma}},[z^{(\ell)}_\sigma,i a^\dag_k a_k]]-[z^{(\ell)}_\sigma,[i a^\dag_k a_k,g_+^{\tilde{\gamma}}]]
=[g_+^{\tilde{\gamma}},[i a^\dag_k a_k,z^{(\ell)}_\sigma]]\\
&=[g_+^{\tilde{\gamma}},\sigma(\alpha_k{-}\beta_k)z^{(\ell)}_{-\sigma}]
=\sigma(\alpha_k{-}\beta_k)[g_\sigma^{\gamma},z^{(\ell)}_{-\sigma}]
=\sigma(\alpha_k{-}\beta_k)z^{(\ell+1)}_{-\sigma},
\end{align*}
which proves (a).
Finally, for all $k$ such that $\alpha_k\neq\beta_k$, we then define $\mu_k:=\alpha_k-\beta_k\neq0$ and the requirements in Definition~\ref{definition:complementary:elements} are fulfilled. Therefore $z^{(\ell)}_+$ and $z^{(\ell)}_-$ are complementary operators for all $\ell$, which proves (b).
\end{proof}

We here provide the most important property of the Commutator Chain of Type I, namely that all of the
elements $y^{(\ell)}_\sigma$ have highest possible degree when there exists at least one index $k$ such that
$\alpha_k\neq\beta_k$ and $\tilde{\gamma}_k\neq0$ (see Theorem~\ref{Second:Chant:Type:II:lemma:MIAO} below).
We start with the following examples which
highlight the structure of the elements in the chain:

\begin{tcolorbox}[colback=orange!3!white,colframe=orange!85!black,title=EXAMPLES: Commutator Chains of Type I and their structure]
\label{Example:Simple:Commutator:Chain}
\begin{example}
Let us consider $z^{(0)}_+:=g_+^\gamma=i[a b^2+a^\dag(b^\dag)^2]$ and $z^{(0)}_-:=g_-^\gamma=a b^2-a^\dag(b^\dag)^2$, and
$g_+^{\tilde{\gamma}}=2 i a^\dag a b^\dag b$. Note that we have 
$\gamma=(\alpha,\beta)=(1,2,0,0)$ 
and $\tilde{\gamma}=(\tilde{\alpha},\tilde{\alpha})=(1,1,1,1)$.
We compute 
\begin{equation*}
z^{(1)}_+=[g_+^{\tilde{\gamma}},z^{(0)}_+]\simeq 2[a b^\dag b^3-a^\dag (b^\dag)^3 b]+4[a^\dag a^2 b^2-(a^\dag)^2 a (b^\dag)^2],
\end{equation*}
where the lower-degree term $[a b^2- a^\dag (b^\dag)^2]$ has been removed. It is clear here that
$\deg(z^{(1)}_+)=5=\deg(g_+^{\tilde{\gamma}})+\deg(z^{(0)}_+)-2$, and this is the highest possible degree. We proceed to compute
$$z^{(2)}_+=[g_+^{\tilde{\gamma}},z^{(1)}_+]\simeq
-4i [a \bd{2}b^4+a^\dag \bd{4} b^2]-16i[a^\dag a^2 b^\dag b^3+\ad{2} a \bd{3} b]-16i[\ad{2} a^3 b^2+\ad{3} a^2 \bd{2}].$$
Also, in this case, the degree is maximal as $\deg(z^{(2)}_+)=7=\deg(g_+^{\tilde{\gamma}})+\deg(z^{(1)}_+)-2$,.
Furthermore we see that, in the expansion of $z^{(2)}_+$, the highest-degree terms have the generic form
$g_+^{(2\tilde{\alpha}+\alpha-2\oneshalf_a,2\tilde{\alpha}+\beta-2\oneshalf_a)}$,
$g_+^{(2\tilde{\alpha}+\alpha-\oneshalf_a-\oneshalf_b,2\tilde{\alpha}+\beta-\oneshalf_a-\oneshalf_b)}$ and
$g_+^{(2\tilde{\alpha}+\alpha-2\oneshalf_b,2\tilde{\alpha}+\beta-2\oneshalf_b)}$, and they are linearly independent.
They are the only possible independent generators in $\hat{A}_2$ that can be obtained starting from a generator of the form
$g_\sigma^{(\alpha,\beta)}$ where two arbitrary operators $a_k$ and two arbitrary
operators $a^\dag_k$ are removed for indices $k$
such that $\alpha_k\neq\beta_k$.
\end{example}
\begin{example}
We now consider $z^{(0)}_+:=g_+^\gamma=ia^\dag a( b{+}b^\dag)$ and $z^{(0)}_-:=g_-^\gamma=a^\dag a( b{-}b^\dag)$, and
$g_+^{\tilde{\gamma}}=2 i a^\dag a b^\dag b$. Note that we have 
$\gamma=(\alpha,\beta)=(1,1,1,0)$ 
and $\tilde{\gamma}=(\tilde{\alpha},\tilde{\alpha})=(1,1,1,1)$. It is then immediate to verify that 
$z^{(1)}_+\simeq
-2\ad{2} a^2( b^\dag{-}b)$ and
$z^{(2)}_+\simeq
-4i \ad{3} a^3(b^\dag{+}b)$. As above, the highest-degree elements (which is a single element in this case) are the only possible
independent generators in $\hat{A}_2$ that can be obtained starting from a generator of the form
$g_\sigma^{(\alpha,\beta)}$ where
two arbitrary operators $a_k$ and two arbitrary operators $a^\dag_k$ are removed for indices $k$ such that
$\alpha_k\neq\beta_k$. In fact, we see that not a single operator $a,a^\dag$ has been removed as $\alpha_1=\beta_1$ for this case.
\end{example}
\end{tcolorbox}

To streamline the presentation, we employ the notation of Definition~\ref{definition:useful:notation}. 
We start by detailing the induction step in the Commutator Chain of of Type I:

\begin{proposition}\label{max:length:simple:chain:theorem}
Let $\gamma = (\alpha,\beta),\lambda^{\ell}=(\mu^{\ell},\nu^{\ell}) \in \Np{2n}$ with $\alpha \geq \beta$
and $\mu^{\ell} \geq \nu^{\ell}$; let $\tilde{\gamma} = (\tilde{\alpha},\tilde{\alpha})$;
$\sigma \in \{+,-\}$ and $\ell \in \Np{}$. Below $u \in \{1,\ldots,n\}$, $S$ are multisets,
$\mathcal{I}^{\ell}$ is a set of multisets such that $S \in \mathcal{I}^{\ell}$ implies $\abs{S}=k^{(\ell)}$,
and $d^{(\ell)}=\abs{\lambda^{\ell}}-2k^{(\ell)}$.
Consider the expansion
\begin{align}\label{eq:chain:two:induction:start:good}
z_{\sigma}^{(\ell)} \upto{d^{(\ell)}}\;
\sigma^{\ell}\, (-1)^{\ell(\ell+1)/2}\, 2^\ell
\sum_{S \in \mathcal{I}^{\ell}} c_S^{(\ell)} g_{(-1)^{\ell} \sigma}^{\lambda^{\ell}-\ones(S)}
\end{align}
with $c_S^{(\ell)}\in \R$ and $\lambda^{\ell}-\ones(S) \in \Np{2n}$ for $S \in  \mathcal{I}^{\ell}$.
The base case is given by $\lambda^{0}=\gamma$, $z_{\pm}^{(0)} = g_{\pm}^{\gamma}$,
$\mathcal{I}^{0}=\{ \emptymultiset \}$, and $c_{\emptymultiset}^{(0)}=1$. We introduce the coefficient
$c(\ell,u)=\tilde{\alpha}_u(\nu_u^{\ell}{-}\mu_u^{\ell})$ and obtain
\begin{align}\label{eq:chain:two:induction:finish:good}
z_{\sigma}^{(\ell+1)} & = [ g_+^{\tilde{\gamma}},z_{\sigma}^{(\ell)}]
\upto{d^{(\ell+1)}}\, \sigma^{\ell+1}\, (-1)^{(\ell+1)(\ell+2)/2}\, 2^{\ell+1}
\sum_{S,u} c_{S}^{(\ell)}c(\ell,u)\, g_{(-1)^{\ell+1} \sigma}^{\lambda^{\ell}+\tilde{\gamma}-\ones(S\cup\multisetnarrow{u})}
\nonumber
\\
&\upto{d^{(\ell+1)}}\;
\sigma^{\ell+1}\, (-1)^{(\ell+1)(\ell+2)/2}\, 2^{\ell+1}\sum_{U\in\mathcal{I}^{\ell+1}}
c_U^{(\ell+1)}\, g_{(-1)^{\ell+1}\sigma}^{\lambda^{\ell+1}-\ones(U)} \;\text{ with }\;
c_U^{(\ell+1)}
=
\hspace{-2mm} \sum_{\genfrac{}{}{0pt}{}{S,u}{U = S \cup\multisetnarrow{u}}} \hspace{-2mm}
c_{S}^{(\ell)} c(\ell,u) \in \R,
\end{align}
where $\lambda^{\ell+1}=\lambda^{\ell}+\tilde{\gamma}$,
$\mathcal{I}^{\ell+1}$ contains all multisets $U = S \cup\multiset{u}$ such that $\lambda^{\ell+1}-\ones(U) \in \Np{2n}$,
$U \in \mathcal{I}^{\ell+1}$ implies $\abs{U}=k^{(\ell+1)}$, $d^{(\ell+1)}=\abs{\lambda^{\ell}}+\abs{\tilde{\gamma}}-2k^{(\ell)}-2=\abs{\lambda^{\ell+1}}-2k^{(\ell+1)}$, and $k^{(\ell+1)}=k^{(\ell)}+1$.
\end{proposition}

\begin{proof}
We assume that $z^{(\ell)}_{\sigma}$ is given by \eqref{eq:chain:two:induction:start:good} and we compute
	$z^{(\ell+1)}_{\sigma}$.
We apply Lemma~\ref{lem:omnibus}(e) and (f), and we obtain
\begin{align*}
	z^{(\ell+1)}_{\sigma} & =[g_+^{\tilde{\gamma}},z^{(\ell)}_{\sigma}]
\upto{d^{(\ell+1)}} \sigma^{\ell}\, (-1)^{\ell(\ell+1)/2}\, 2^\ell
	\sum_{S \in \mathcal{I}^{\ell}} c_S^{(\ell)}\bigl[g_+^{\tilde{\gamma}}, g_{(-1)^{\ell} \sigma}^{\lambda^{\ell}-\ones(S)}\bigr]\\
	&\upto{d^{(\ell+1)}}\, \sigma^{\ell} (-1)^{\ell(\ell+1)/2}\, 2^{\ell}
	\sum_{S,u} c_S^{(\ell)}
	(-1)^{\ell+1}\sigma\, 2
	\tilde{\alpha}_u(\nu_u^\ell{-}\mu_u^\ell)\,g_{(-1)^{\ell+1} \sigma}^{\lambda^{\ell}+\tilde{\gamma}-\ones(S\cup\multiset{u})}\\	
	&\upto{d^{(\ell+1)}} \sigma^{\ell+1} (-1)^{(\ell+1)(\ell+2)/2}\, 2^{\ell+1}
	\sum_{S,u} c_S^{(\ell)}\tilde{\alpha}_u(\nu_u^\ell{-}\mu_u^\ell)\,g_{(-1)^{\ell+1} 
	\sigma}^{\lambda^{\ell}+\tilde{\gamma}-\ones(S\cup\multiset{u})}.
\end{align*}
The expressions for $\lambda^{\ell+1}$, $U$, $k^{(\ell+1)}$ and $d^{(\ell+1)}$ can be read off directly from the equations above.
\end{proof}

We summarize a few immediate consequences relevant for the proof of our main results
in Theorem~\ref{main:theorem_final} below:

\begin{proposition}\label{small:helping:lemma:coefficientsMIAO}
We employ the notation of Proposition~\ref{max:length:simple:chain:theorem}. We obtain
that (a)~$k^{(\ell)}=\ell$; (b)~$\mu^\ell=k^{(\ell)}\tilde{\alpha} + \alpha$; (c)~$\nu^\ell=k^{(\ell)}\tilde{\alpha} + \beta$.
(d) The expansion of $z_{\sigma}^{(\ell)}$
is given by  Eq.~\eqref{eq:chain:two:induction:start:good} for suitable sets of multisets $\mathcal{I}^{\ell}$.
\end{proposition}
\begin{proof}
Given the recursive relation $k^{(\ell+1)}=k^{(\ell)}+1$ and the initial condition $k^{(0)}=0$ for the quantities $k^{(\ell)}$, 
we verify (a) by induction. 
Similarly, the relation $\lambda^{\ell+1}=\lambda^{\ell}+\tilde{\gamma}$
and the initial condition $\lambda^0=\gamma$ immediately imply (b) and (c).
The claim (d) follows by induction from Proposition~\ref{max:length:simple:chain:theorem}.
\end{proof}

We identify a condition under which the elements 
$z_{\sigma}^{(\ell)}$ in the chain are always nonzero, and $z_{\sigma}^{(\ell)}$ has highest degree:

\begin{theorem}\label{Second:Chant:Type:II:lemma:MIAO} 
We employ the notation of Proposition~\ref{max:length:simple:chain:theorem}. We assume that there is at least one index
$k$ such that $\alpha_k\neq\beta_k$ and $\tilde{\alpha}_k\neq0$. 
Then: (a) the coefficient
$c_{\multiset{k^{\times \ell}}}^{(\ell)}=(\tilde{\alpha}_k(\alpha_k{-}\beta_k))^\ell\neq 0$ and the elements
$z_{\sigma}^{(\ell)}$ contain a unique term of the form $g_{(-1)^\ell\sigma}^{\lambda^{\ell}-\ell \tau_k}$ in their expansion
(where $S_k=\abs{S}=k^{(\ell)}=\ell \geq1$). (b)~We also have
$\deg(z_{\sigma}^{(\ell)})=2\ell(\abs{\tilde{\alpha}}{-}1)+\abs{\alpha}+\abs{\beta}$.
\end{theorem}
\begin{proof}
We prove (a) by induction.
If $k$ is an index such that such that $\alpha_k\neq\beta_k$ and $\tilde{\alpha}_k\neq0$, then
$c_{\multiset{k^{\times 1}}}^{(1)}=\tilde{\alpha}_k(\nu_k^0{-}\mu_k^0)=\tilde{\alpha}_k(\alpha_k{-}\beta_k)\neq 0$.
Let us assume that $c_{\multiset{k^{\times \ell}}}^{(\ell)}=(\tilde{\alpha}_k(\alpha_k{-}\beta_k))^\ell$.
In general, $c_U^{(\ell+1)}=\sum_{S,u} c_{S}^{(\ell)} c(\ell,u)$, with $U = S \cup\multisetnarrow{u}$.
We focus on the
coefficient $c_{\multiset{k^{\times (\ell+1)}}}^{(\ell+1)}$, with $|U|=U_k=S_k{+}1=\ell{+1}$.
Then, we have 
$c_{\multiset{k^{\times (\ell+1)}}}^{(\ell+1)}= c_{\multiset{k^{\times \ell}}}^{(\ell)} c(\ell,u_k)=
(\tilde{\alpha}_k(\alpha_k{-}\beta_k))^\ell c(\ell,u_k)=(\tilde{\alpha}_k(\alpha_k{-}\beta_k))^{\ell+1}\neq 0$ as claimed.
This settles (a) as there is a unique nonzero coefficient $c_{\multiset{k^{\times \ell}}}^{(\ell)}$
and term $g_{(-1)^\ell\sigma}^{\lambda^{\ell}-\ell\tau_k}$
for the multiset $S=\multiset{k^{\times \ell}}$. 
Proposition~\ref{small:helping:lemma:coefficientsMIAO} implies
$\deg(z_{\sigma}^{(\ell)})=\deg(g_{(-1)^\ell\sigma}^{\lambda^{\ell}-\ell\tau_k})=\abs{\lambda^{\ell}}-2\ell
=\abs{\mu^\ell}+
\abs{\nu^\ell}-2\ell=
\ell \abs{\tilde{\alpha}}+
\abs{\alpha}+\ell \abs{\tilde{\alpha}}+\abs{\beta}-2\ell=
2\ell(\abs{\tilde{\alpha}}{-}1)+\abs{\alpha}+\abs{\beta}$, which settles (b).
\end{proof}

\subsection{Commutator Chains of Type II}
Similarly as in the previous subsection, we now introduce
\emph{Commutator Chains of Type II} which also consists of elements of the skew-hermitian Weyl algebra.

\begin{definition}\label{Commutator:Chain:Type:II}
Consider the (seed) elements
$g_+^{\gamma}, g_-^{\gamma} \in \hat{A}_n$ 
with $\gamma=(\alpha,\beta)$. The elements
$$C_\sigma^{II}(\gamma):=\{y^{(\ell)}_\sigma\in\hat{A}_n:\,\,y^{(\ell+1)}_\sigma
:=[y^{(\ell)}_\sigma,[y^{(\ell)}_+,y^{(\ell)}_-]]\,\,\text{for all integers}\,\,\ell\geq0,\,\,\text{with}\,\,y^{(0)}_\sigma:=g_\sigma^{\gamma}\}$$
form a \emph{Commutator Chain of Type II} generated from $g_{+}^{\gamma}$ and $g_{-}^{\gamma}$.
\end{definition}
The commutator structure of the chain $C_\sigma^{II}(\gamma)$ is visualized
in Figure~\ref{Figure:Commutator_Chain}. We now collect useful properties. As before,
we determine the degree of its elements.

\begin{figure}[t]
\includegraphics[width=0.8\linewidth]{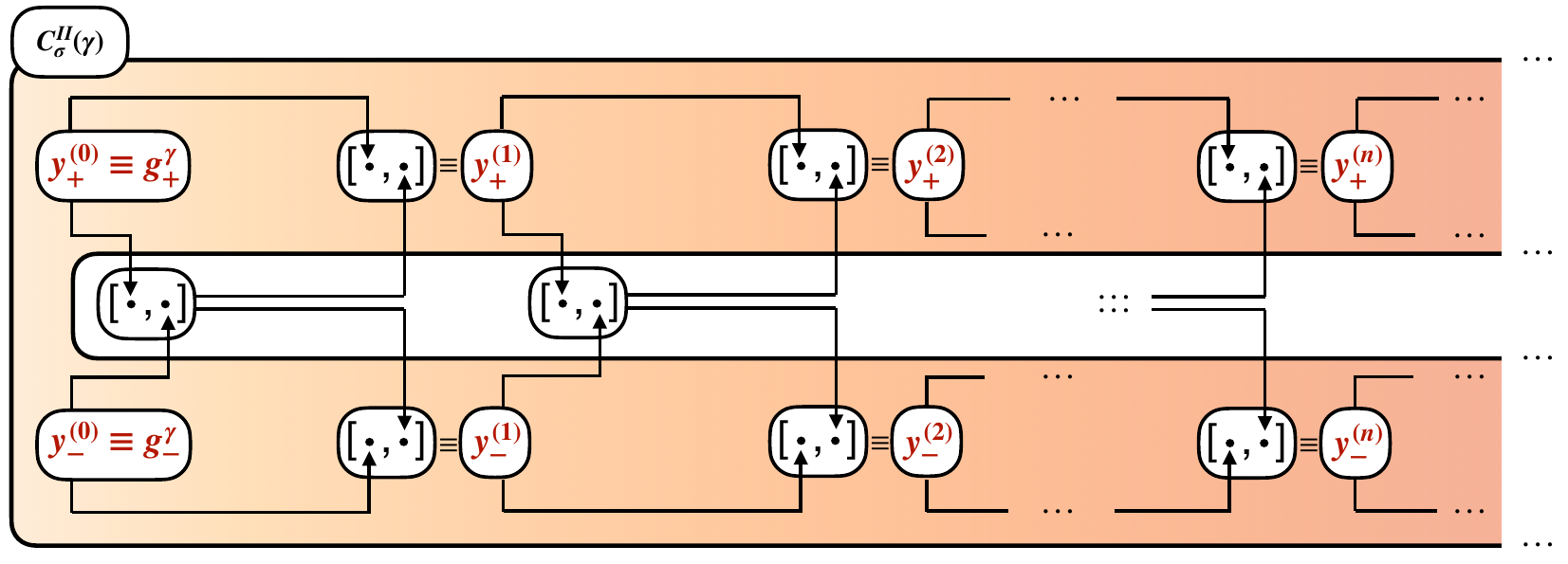}
\caption{\textbf{Commutator Chains of Type II}: Visualization of a Commutator Chain of
Type II which is denoted by $C_\sigma^{II}(\gamma)$ and has elements $y^{(\ell)}_\sigma$ (highlighted in red).
It is constructed starting from $y^{(0)}_\sigma:=g_\sigma^{\gamma} \in \hat{A}_n$
with $\gamma=(\alpha,\beta)$ and $\sigma \in\{+,-\}$.
\label{Figure:Commutator_Chain}}
\end{figure}

\begin{lemma}\label{chain:lemma}
Let $C_\sigma^{II}(\gamma)$ be a
Commutator Chain of Type II with elements $y^{(\ell)}_\sigma$ and generated by $g_+^{\gamma}, g_-^\gamma \in \hat{A}_n$, where $\gamma=(\alpha,\beta)$ and $\sigma \in\{+,-\}$. Then: (a) $[i a^\dag_k a_k,y^{(\ell)}_\sigma]=\sigma(\alpha_k{-}\beta_k)y^{(\ell)}_{-\sigma}$ for all $k$;  (b) If there exists at least one $k$ such that $\alpha_k\neq\beta_k$ then  $y^{(\ell)}_+$ and $y^{(\ell)}_-$ are complementary operators for all $\ell$.
\end{lemma}
\begin{proof}
We know from Theorem~\ref{important:one:element:and:diagonal:implies:other:two:mode:new} that
$[i a^\dag_k a_k,g_\sigma^{\gamma}]=\sigma(\alpha_k{-}\beta_k)g_{-\sigma}^{\gamma}$.
We then check that the claim is true for $y^{(1)}_\sigma:=[y^{(0)}_\sigma,[y^{(0)}_+,y^{(0)}_-]]$.
In fact, we use Jacobi's identity to show that 
\begin{align*}
[i a^\dag_k a_k, y^{(1)}_\sigma]&=[i a^\dag_k a_k, [y^{(0)}_\sigma,[y^{(0)}_+,y^{(0)}_-]]]
=[i a^\dag_k a_k, [g_{\sigma}^{\gamma},[g_{+}^{\gamma},g_{-}^{\gamma}]]]
=-[g_{\sigma}^{\gamma}, [[g_{+}^{\gamma},g_{-}^{\gamma}],i a^\dag_k a_k]]-[[g_{+}^{\gamma},g_{-}^{\gamma}],
[i a^\dag_k a_k,g_{\sigma}^{\gamma}]]\\
&=-\sigma(\alpha_k{-}\beta_k)[[g_{+}^{\gamma},g_{-}^{\gamma}], g_{-\sigma}^{\gamma}]
=\sigma(\alpha_k{-}\beta_k)y^{(1)}_{-\sigma}, 
\end{align*}
where the term $[g_{\sigma}^{\gamma}, [[g_{+}^{\gamma},g_{-}^{\gamma}],i a^\dag_k a_k]]$ vanishes due to Lemma~\ref{pre:chain:lemma}. 
We have shown the case of $\ell=1$. Now, we assume that
$[i a^\dag_k a_k, y^{(\ell)}_\sigma]=\sigma(\alpha_k-\beta_k)y^{(\ell)}_{-\sigma}$ is true for $\ell$. We then use Jacobi's identity and get
\begin{align*}
[i a^\dag_k a_k, y^{(n+1)}_\sigma]&=[i a^\dag_k a_k, [y^{(\ell)}_\sigma,[y^{(\ell)}_+,y^{(\ell)}_-]]]
=-[y^{(\ell)}_\sigma, [[y^{(\ell)}_+,y^{(\ell)}_-],i a^\dag_k a_k]]-[[y^{(\ell)}_+,y^{(\ell)}_-], [i a^\dag_k a_k,y^{(\ell)}_\sigma]]\\
&=[y^{(\ell)}_\sigma, [[y^{(\ell)}_-,i a^\dag_k a_k],y^{(\ell)}_+]]+[y^{(\ell)}_\sigma, [[i a^\dag_k a_k,y^{(\ell)}_+],y^{(\ell)}_-]]-
\sigma(\alpha_k{-}\beta_k)[[y^{(\ell)}_+,y^{(\ell)}_-], y^{(\ell)}_{-\sigma}]\\
&=(\alpha_k{-}\beta_k)[y^{(\ell)}_\sigma, [y^{(\ell)}_+,y^{(\ell)}_+]]+(\alpha_k{-}\beta_k)[y^{(\ell)}_\sigma, [y^{(\ell)}_-,y^{(\ell)}_-]]-
\sigma(\alpha_k{-}\beta_k)[[y^{(\ell)}_+,y^{(\ell)}_-], y^{(\ell)}_{-\sigma}]\\
&=-\sigma(\alpha_k{-}\beta_k)[[y^{(\ell)}_+,y^{(\ell)}_-], y^{(\ell)}_{-\sigma}]
=\sigma(\alpha_k{-}\beta_k)y^{(\ell+1)}_{-\sigma}. 
\end{align*}
This settles (a).
Finally, for all $k$ such that $\alpha_k\neq\beta_k$, we then define $\mu_k:=\alpha_k-\beta_k\neq0$ and the requirements in Definition~\ref{definition:complementary:elements} are fulfilled. Therefore $y^{(\ell)}_+$ and $y^{(\ell)}_-$ are complementary operators for all $\ell$, which proves (b).
\end{proof}

\begin{lemma}\label{max:length:chain:lemma:II}
Let $C_\sigma^{II}(\gamma)$ be a
Commutator Chain of Type II with elements $y^{(\ell)}_\sigma$ and generated by $g_+^{\gamma}, g_-^\gamma \in \hat{A}_n$, where $\gamma=(\alpha,\beta)$ and $\sigma \in\{+,-\}$.
Then, $\deg(y^{(\ell)}_\sigma)\leq3^\ell(\abs{\gamma}{-}2)+2$ for all integers $\ell\geq0$.
\end{lemma}

\begin{proof}
The case of $\ell=0$ is immediate.
Following Theorem~\ref{theorem:commutator:monomials}, the commutator of two
skew-hermitian operators $x,x'\in\hat{A}_n$ always observes $\deg([x,x'])\leq\deg(x)+\deg(x')-2$.
Thus $\deg(y^{(1)}_\sigma)\leq3\abs{\gamma}-4$, since $y^{(1)}_\sigma$ is obtained via two commutators. However,
$3\abs{\gamma}-4=3^1(\abs{\gamma}{-}2)+2$, which settles $\ell=1$.
We assume that the statement holds for $\ell$, i.e., $\deg(y^{(\ell)}_\sigma)\leq3^\ell(\abs{\gamma}{-}2)+2$.
Therefore, $\deg(y^{(\ell+1)}_\sigma)=\deg([y^{(\ell)}_\sigma,[y^{(\ell)}_+,y^{(\ell)}_-]])\leq
3\deg(y^{(\ell)}_\sigma)-4\leq3(3^\ell(\abs{\gamma}{-}2)+2)-4=
3^{\ell+1}(\abs{\gamma}{-}2)+6-4=3^{\ell+1}(\abs{\gamma}{-}2)+2$, which proves the statement.
\end{proof}

It is instructive how the bound $\deg(y^{(\ell)}_\sigma)\leq3^\ell(\abs{\gamma}-2)+2$ can be found. 
We directly see that $\deg(y^{(1)}_\sigma)\leq3\abs{\gamma}-4$, $\deg(y^{(2)}_\sigma)\leq9\abs{\gamma}-16$,
and $\deg(y^{(3)}_\sigma)\leq27\abs{\gamma}-52$.
Thus we conjecture that $\deg(y^{(\ell)}_\sigma)\leq3^\ell|\gamma|-b_\ell$, where the coefficients $b_\ell$ satisfy the
recurrence relation $b_\ell=3b_{\ell-1}+4$ with $b_0=0$. This implies that
$b_\ell=4\sum_{k=0}^{\ell-1}3^k=2(3^\ell{-}1)$.
We put all of this together and we obtain
$\deg(y^{(\ell)}_\sigma)\leq3^\ell \abs{\gamma}-b_\ell=3^\ell \abs{\gamma}-2(3^\ell{-}1)=3^\ell(\abs{\gamma}{-}2)+2$, as claimed.

Lemma~\ref{max:length:chain:lemma:II} 
shows that $\abs{\gamma}>2$ is necessary 
for a Commutator Chain of Type II to contain elements of ever-increasing degree.
This is consistent with the fact that the subalgebra $\hat{A}_n^2\subseteq\hat{A}_n$ generated by
all quadratic operators coincides with the Lie algebra $\mathfrak{sp}(2n,\mathbb{R})$ of the (finite-dimensional) symplectic group
$\mathrm{Sp}(2n,\mathbb{R})$ of dimension $n(2n{+}1)$, see
\cite{Adesso:Ragy:2014}. 
We now analyze the following recursive expressions, which forms the basis to build steps from 
$\ell$ to $\ell{+}1$ in a Commutator Chain of Type II.

\begin{proposition}\label{eq:chain:two:induction:iteration}
Let $\gamma = (\alpha,\beta),\lambda^{\ell}=(\mu^{\ell},\nu^{\ell}) \in \Np{2n}$ with $\alpha \geq \beta$,
$\mu^{\ell} \geq \nu^{\ell}$,
$\sigma \in \{+,-\}$, and $\ell \in \Np{}$. Below $u,v \in \{1,\ldots,n\}$, $R,S,T,U$ are multisets,
$\mathcal{I}^{\ell}$ is a set of multisets such that $R \in \mathcal{I}^{\ell}$ implies $\abs{R}=k^{(\ell)}$,
and $d^{(\ell)}=\abs{\lambda^{\ell}}-2k^{(\ell)}$.
Consider 
\begin{align}\label{eq:chain:two:induction:start}
y_{\sigma}^{(\ell)} \upto{d^{(\ell)}}\;
\sigma^{\ell}
\sum_{S \in \mathcal{I}^{\ell}} c_S^{(\ell)} g_{(-1)^{\ell} \sigma}^{\lambda^{\ell}-\ones(S)}
\end{align}
with $c_S^{(\ell)}\in \R$ and $\lambda^{\ell}-\ones(S) \in \Np{2n}$ for $S \in  \mathcal{I}^{\ell}$.
The base case is given by $\lambda^{0}=\gamma$, $y_{\pm}^{(0)} = g_{\pm}^{\gamma}$,
$\mathcal{I}^{0}=\{ \emptymultiset \}$, and $c_{\emptymultiset}^{(0)}=1$. We introduce the coefficient
$c(\ell,R,S,T,u,v)=2 (\nu_u^{\ell}{-}\mu_u^{\ell})[\mu_u^{\ell}{+}\nu_u^{\ell}{-}(R{\cup}S)_u]
(\nu_v^{\ell}{-}\mu_v^{\ell})[\mu_v^{\ell}{+}\nu_v^{\ell}{-}(R{\cup}S)_v{-}\delta_{u,v}]$ and obtain
\begin{align}\label{eq:chain:two:induction:finish}
y_{\sigma}^{(\ell+1)} & = [y_{\sigma}^{(\ell)}, [ y_{+}^{(\ell)}, y_{-}^{(\ell)}]]
\upto{d^{(\ell+1)}}\; \sigma^{\ell+1}\sum_{R,S,T \in \mathcal{I}^{\ell}} \sum_{u,v} (-1)^{\ell+1}
c_{T}^{(\ell)}c_{R}^{(\ell)}c_{S}^{(\ell)} c(\ell,R,S,T,u,v)\,
g_{(-1)^{\ell+1}\sigma}^{2\lambda^{\ell}+[\lambda^{\ell}]^{\dagger}-\ones(T\cup R\cup S \cup \multisetnarrow{u,v})}
\nonumber
\\
&\upto{d^{(\ell+1)}}\;
\sigma^{\ell+1}\sum_{U\in\mathcal{I}^{\ell+1}}
c_U^{(\ell+1)}\, g_{(-1)^{\ell+1}\sigma}^{\lambda^{\ell+1}-\ones(U)} \;\text{ with }\;
c_U^{(\ell+1)}
=
(-1)^{\ell+1} \hspace{-4mm}
\sum_{\genfrac{}{}{0pt}{}{R,S,T,u,v}{U = T\cup R \cup S \cup\multisetnarrow{u,v}}} \hspace{-4mm}
c_{T}^{(\ell)}c_{R}^{(\ell)}c_{S}^{(\ell)} c(\ell,R,S,T,u,v) \in \R,
\end{align}
where $\lambda^{\ell+1}=2\lambda^{\ell}+[\lambda^{\ell}]^{\dagger}$,
$\mathcal{I}^{\ell+1}$ contains all multisets $U = T\cup R \cup S \cup\multiset{u,v}$ such that $\lambda^{\ell+1}-\ones(U) \in \Np{2n}$,
$U \in \mathcal{I}^{\ell+1}$ implies $\abs{U}=k^{(\ell+1)}$, $d^{(\ell+1)}=3\abs{\lambda^{\ell}}-6k^{(\ell)}-4=\abs{\lambda^{\ell+1}}-2k^{(\ell+1)}$, and $k^{(\ell+1)}=3k^{(\ell)}+2$.
\end{proposition}

\begin{proof}
We apply Lemma~\ref{lem:omnibus}(a) to $[g_{+}^{\gamma}, g_{-}^{\gamma}]$ and obtain $[g_{+}^{\gamma}, g_{-}^{\gamma}] \upto{d}\; 
\sum\nolimits_{v} (\beta_v^2 {-} \alpha_v^2)\,
g_{+}^{\gamma+\gamma^{\dagger}-\ones_v}$ where $d:=2\abs{\gamma}-2$.
We consider $[y_{+}^{(\ell)},y_{-}^{(\ell)}]$ and 
set $f:= 2 d^{(\ell)} -2 = 2 \abs{\lambda^{\ell}}-4k^{(\ell)}-2$ to apply first 
Prop.~\ref{app:polynomials}(a4) and then
Lemma~\ref{lem:omnibus}(h) or (i) to find
\begin{subequations}
\begin{alignat}{10}
&[y_{+}^{(\ell)},y_{-}^{(\ell)}] && \upto{f}\; &&
(-1)^\ell\sum_{R \in \mathcal{I}^{\ell}}\sum_{S \in \mathcal{I}^{\ell}}
c_R^{(\ell)}c_S^{(\ell)} \left[g_{(-1)^{\ell}}^{\lambda^{\ell}-\ones(R)}, g_{(-1)^{\ell+1}}^{\lambda^{\ell}-\ones(S)}\right]\\
\label{specialline}
&&& \upto{f}\; &&
(-1)^\ell\sum_{R \in \mathcal{I}^{\ell}}\sum_{S \in \mathcal{I}^{\ell}} c_R^{(\ell)}c_S^{(\ell)}
\left\{
\sum\nolimits_{v} [(\mu^\ell_v{-}R_v) (\nu^\ell_v{-}S_v) {-} (\nu^\ell_v{-}R_v) (\mu^\ell_v{-}S_v)]\,
g_{+}^{2\lambda^\ell-\ones(R\cup S \cup \multisetnarrow{v})}
\right.\\
&&&&&+(-1)^\ell
\left.
\sum\nolimits_{v} [(\nu^\ell_v{-}R_v) (\nu^\ell_v{-}S_v) {-} (\mu^\ell_v{-}R_v) (\mu^\ell_v{-}S_v)]\,
g_{+}^{\lambda^\ell+(\lambda^\ell)^{\dagger}-\ones(R\cup S \cup \multisetnarrow{v})}
\right\}\\
&&& \upto{f}\; &&
\sum_{R \in \mathcal{I}^{\ell}}\sum_{S \in \mathcal{I}^{\ell}} c_R^{(\ell)}c_S^{(\ell)}
\sum\nolimits_{v} (\nu^\ell_v{-}\mu^\ell_v)(\nu^\ell_v{+}\mu^\ell_v{-}(R{\cup} S)_v)\,
g_{+}^{\lambda^\ell+(\lambda^\ell)^{\dagger}-\ones(R\cup S \cup \multisetnarrow{v})}.
\end{alignat}
\end{subequations}
Here the contribution in \eqref{specialline} vanishes identically due to the anti-symmetry in the coefficients
and the symmetry in (the exponent of) the basis elements when the sets $R$ and $S$ are exchanged.
Next we set $e:=\abs{\lambda^{\ell}} +\abs{\lambda^{\ell} {+} (\lambda^{\ell})^{\dagger}} -2 k^{(\ell)}
-2 (2 k^{(\ell)}{+}1) -2 = 3\abs{\lambda^{\ell}}-6k^{(\ell)}-4$ and apply first 
Prop.~\ref{app:polynomials}(a4) and then
Lemma~\ref{lem:omnibus}(f) or (g) to finally obtain our result
\begin{alignat*}{10}
& y_{\sigma}^{(\ell+1)} && \upto{e}\; && [y_{\sigma}^{(\ell)}, [ y_{+}^{(\ell)}, y_{-}^{(\ell)}]]\\
&&& \upto{e}\; && \sum_{R \in \mathcal{I}^{\ell}}\sum_{S \in \mathcal{I}^{\ell}} c_R^{(\ell)}c_S^{(\ell)}
\sum\nolimits_{v} (\nu^\ell_v{-}\mu^\ell_v)(\nu^\ell_v{+}\mu^\ell_v{-}(R{\cup} S)_v)\,\left[y_{\sigma}^{(\ell)}, 
g_{+}^{\lambda^\ell+(\lambda^\ell)^{\dagger}-\ones(R\cup S \cup \multisetnarrow{v})}\right]\\
&&& \upto{e}\; && \sigma^\ell\sum_{R \in \mathcal{I}^{\ell}}\sum_{S \in \mathcal{I}^{\ell}}
\sum_{T \in \mathcal{I}^{\ell}} c_R^{(\ell)}c_S^{(\ell)}c_T^{(\ell)}
\sum\nolimits_{v} (\nu^\ell_v{-}\mu^\ell_v)(\nu^\ell_v{+}\mu^\ell_v{-}(R{\cup} S)_v)\,\left[g_{(-1)^{\ell} \sigma}^{\lambda^{\ell}
-\ones(T)}, g_{+}^{\lambda^\ell+(\lambda^\ell)^{\dagger}-\ones(R\cup S \cup \multisetnarrow{v})}\right]\\
&&&\upto{e}\; && (-1)^\ell\sigma^\ell\sum_{R \in \mathcal{I}^{\ell}}
\sum_{S \in \mathcal{I}^{\ell}}\sum_{T \in \mathcal{I}^{\ell}} c_R^{(\ell)}c_S^{(\ell)}c_T^{(\ell)}
\sum\nolimits_{v} (\nu^\ell_v{-}\mu^\ell_v)(\nu^\ell_v{+}\mu^\ell_v{-}(R{\cup} S)_v)
\\
&&&&&
\times\left(
-\sigma
\{
\sum\nolimits_{u} [(\mu^\ell_u{-}T_u) 
(\mu^\ell_u{+}\nu^\ell_u{-}(R{\cup} S {\cup} \multisetnarrow{v})_u) {-} 
(\nu^\ell_u{-}T_u) (\mu^\ell_u{+}\nu^\ell_u{-}(R{\cup} S {\cup} \multisetnarrow{v})_u)]\,
g_{(-1)^{\ell+1} \sigma}^{2\lambda^\ell+(\lambda^\ell)^{\dagger}-\ones(R{\cup} S {\cup} T\cup\multisetnarrow{v,u})}
\} 
\right.
\\
&&&&&
\left.
+\sigma
\{
\sum\nolimits_{u} [(\nu^\ell_u{-}T_u) 
(\mu^\ell_u{+}\nu^\ell_u{-}(R{\cup} S {\cup} \multisetnarrow{v})_u) {-} 
(\mu^\ell_u{-}T_u) (\mu^\ell_u{+}\nu^\ell_u{-}(R{\cup} S {\cup} \multisetnarrow{v})_u)] 
\right.\\
&&&&&
\left. 
\times
E(2\lambda^\ell{+}(\lambda^\ell)^{\dagger})\,
g_{(-1)^{\ell+1} \sigma}^{\Theta(2\lambda^\ell+(\lambda^\ell)^{\dagger})-\ones(R{\cup} S {\cup} T\cup\multisetnarrow{v,u})}
\}
\right)
\\
&&& \upto{e}\; && \sigma^{\ell+1}\sum_{R \in \mathcal{I}^{\ell}}
\sum_{S \in \mathcal{I}^{\ell}}\sum_{T \in \mathcal{I}^{\ell}} c_R^{(\ell)}c_S^{(\ell)}c_T^{(\ell)}
\sum\nolimits_{v,u} 2(-1)^{\ell+1}(\nu^\ell_v{-}\mu^\ell_v)(\nu^\ell_v{+}\mu^\ell_v{-}(R{\cup} S)_v)
\\
&&&&&
\times 
(\mu^\ell_u{-}\nu^\ell_u)(\mu^\ell_u{+}\nu^\ell_u{-}(R{\cup} S)_u {-} \delta_{uv})\,
g_{(-1)^{\ell+1} \sigma}^{2\lambda^\ell+(\lambda^\ell)^{\dagger}-\ones(R{\cup} S {\cup} T\cup\multisetnarrow{v,u})}.
\end{alignat*}
We note that $e = d^{(\ell+1)}$ which completes the proof.
\end{proof}

\begin{proposition}\label{small:helping:lemma:coefficients}
We employ the notation of Proposition~\ref{eq:chain:two:induction:iteration}. 
It follows that: (a)~$k^{(\ell)}=3^\ell-1$; (b)~$\mu^\ell=\frac{1}{2}[ (k^{(\ell)}{+}2)\alpha + k^{(\ell)}\beta]$;
(c)~$\nu^\ell=\frac{1}{2}[k^{(\ell)}\alpha + (k^{(\ell)}{+}2)\beta]$; (d) the expansion of $y_{\sigma}^{(\ell)}$
is given by  Eq.~\eqref{eq:chain:two:induction:start} for suitable sets of multisets $\mathcal{I}^{\ell}$.
\end{proposition}
\begin{proof}
We know from Proposition~\ref{eq:chain:two:induction:iteration} that $k^{(\ell+1)}=3k^{(\ell)}+2$ and that $k^{(0)}=0$.
For $\ell=1$, we have $k^{(1)}=3\times0+2=2=3^1-1$.
We assume that $k^{(\ell)}=3^\ell{-}1$ and determine $k^{(\ell+1)}$ as
$k^{(\ell+1)}=3k^{(\ell)}+2=3(3^\ell{-}1)+2=3^{\ell+1}-1$. This proves (a).

We then address (b) and (c). Proposition~\ref{eq:chain:two:induction:iteration} implies
$\lambda^{\ell+1}=2\lambda^{\ell}+[\lambda^{\ell}]^{\dagger}$. For $\ell=1$,
$\lambda^{1}=2\lambda^{0}+[\lambda^{0}]^{\dagger}=2\gamma+\gamma^\dag=(2\alpha+\beta, \alpha+2\beta)$.
Thus, $\mu^1=2\alpha+\beta=[ (k^{(1)}{+}2)\alpha + k^{(1)}\beta]/2$ and
$\nu^1=[k^{(1)}\alpha + (k^{(1)}{+}2)\beta]/2$ as claimed.
We assume (b) and (c) for $\ell$ and show that they hold for $\ell+1$:
\begin{align*}
\lambda^{\ell+1}&=2\lambda^{\ell}+[\lambda^{\ell}]^{\dagger}
=2(\tfrac{1}{2}[ (k^{(\ell)}{+}2)\alpha + k^{(\ell)}\beta],\tfrac{1}{2}[k^{(\ell)}\alpha + 
(k^{(\ell)}{+}2)\beta])+(\tfrac{1}{2}[k^{(\ell)}\alpha + (k^{(\ell)}{+}2)\beta],\tfrac{1}{2}[ (k^{(\ell)}{+}2)\alpha + k^{(\ell)}\beta])\\
&=(\tfrac{1}{2} [(2k^{(\ell)}{+}4)\alpha+2k^{(\ell)}\beta+k^{(\ell)}\alpha+(k^{(\ell)}{+}2)\beta],\tfrac{1}{2}[ 2k^{(\ell)}\alpha +
(2k^{(\ell)}{+}4)\beta+(k^{(\ell)}{+}2)\alpha + k^{(\ell)}\beta])\\
&=(\tfrac{1}{2} [(k^{(\ell+1)}{+}2)\alpha+ k^{(\ell+1)}\beta],\tfrac{1}{2} [k^{(\ell+1)}\alpha+ (k^{(\ell+1)}{+}2)\beta]),
\end{align*}
where we have used the algebraic relations $3k^{(\ell)}+2=k^{(\ell+1)}$ and $3k^{(\ell)}+4=k^{(\ell+1)}+2$.
We have therefore shown that $\mu^{\ell+1}=[(k^{(\ell+1)}{+}2)\alpha+ k^{(\ell+1)}\beta]/2$ and
$\nu^{\ell+1}= [k^{(\ell+1)}\alpha+ (k^{(\ell+1)}{+}2)\beta]/2$. By induction, the claims (b) and (c) follow.
The claim (d) follows by induction from Proposition~\ref{eq:chain:two:induction:iteration}.
\end{proof}

\begin{proposition}\label{Second:Chant:Type:II:lemma} 
We employ the notation of Proposition~\ref{eq:chain:two:induction:iteration}. Let us consider the two conditions
\begin{align*}
&\text{(a) there is at least one index }
\hat{k} \text{ such that } \alpha_{\hat{k}}+\beta_{\hat{k}}=p\geq2 \text{ and }  \alpha_{\hat{k}}\neq\beta_{\hat{k}}
\text{ and}\hspace{2mm}
&
\text{(cf.\ condition (D) in Lemma~\ref{lemma:A_purp})}\\
&\text{(b) there
are at least two indices } \hat{k},\hat{k}'  \text{ such that } \alpha_{\hat{k}}+\beta_{\hat{k}}=\alpha_{\hat{k}'}+\beta_{\hat{k}'}=1.
&
\text{(cf.\ condition (E) in Lemma~\ref{lemma:A_purp})}
\end{align*}
Then we have that:\\
(1a) If (a) holds,
$y_{\sigma}^{(\ell)}$ contains a term for a multiset $S$ such that $\abs{S}=S_{\hat{k}}$ for an index
$\hat{k}$ and $c_S^{(\ell)}\neq0$.\\
(1b) If (b) holds,
$y_{\sigma}^{(\ell)}$ contains a term for
a multiset $S$ such that $S_{\hat{k}}=S_{\hat{k}'}$, $|S|=2S_{\hat{k}}$, and $c_S^{(\ell)}\neq0$
for two indices $k\neq \hat{k}$.\\
(2) If (a) or (b) holds, then $\deg(y_{\sigma}^{(\ell)})=3^\ell(\abs{\gamma}{-}2)+2$.
\end{proposition}
\begin{proof}
We start with the proof of (1). Proposition~\ref{small:helping:lemma:coefficients}(d) implies that 
the expansion of $y_{\sigma}^{(\ell)}$ is given by Eq.~\eqref{eq:chain:two:induction:start}.
For $\ell=0$, we apply Eq.~\eqref{eq:chain:two:induction:finish}
(where $c_{T}^{(0)}=c_{R}^{(0)}=c_{S}^{(0)}=1$ and $\abs{R}=\abs{S}=\abs{T}=1$) and obtain
\begin{align*}
y_{\sigma}^{(1)} :=  [y_{\sigma}^{(0)}, [ y_{+}^{(0)}, y_{-}^{(0)}]]
\upto{d^{(1)}}\; \sigma \sum_{u,v}
 c(1,u,v)\,
g_{-\sigma}^{2\gamma+\gamma^{\dagger}-\ones(\multisetnarrow{u,v})},
\end{align*}
where $ c(1,u,v)=2 (\beta_u{-}\alpha_u)[\alpha_u{+}\beta_u](\beta_v{-}\alpha_v)[\alpha_v{+}\beta_v{-}\delta_{u,v}]$.
If condition (a) holds, then
\begin{align*}
c(1,\hat{k},\hat{k})=2 (\beta_{\hat{k}}{-}\alpha_{\hat{k}})^2(\alpha_{\hat{k}}{+}\beta_{\hat{k}})(\alpha_{\hat{k}}{+}\beta_{\hat{k}}{-}1)
=2 (\beta_{\hat{k}}{-}\alpha_{\hat{k}})^2p(p{-}1) \neq 0.
\end{align*}
The corresponding term $g_{-\sigma}^{2\gamma+\gamma^{\dagger}-\ones(\multisetnarrow{\hat{k},\hat{k}})}$
in the expansion of $y_{\sigma}^{(1)} $ appears only once and is linearly independent from all other terms.
If condition (b) holds, then
\begin{align*}
c(1,\hat{k},\hat{k}')=c(1,\hat{k}',\hat{k})=&2 (\beta^2_{\hat{k}}{-}\alpha^2_{\hat{k}})(\beta^2_{\hat{k}'}{-}\alpha^2_{\hat{k}'}) \neq 0.
\end{align*}
Moreover, the terms
$g_{-\sigma}^{2\gamma+\gamma^{\dagger}-\ones(\multisetnarrow{\hat{k},\hat{k}'})}$ and $g_{-\sigma}^{2\gamma+\gamma^{\dagger}-\ones(\multisetnarrow{\hat{k}',\hat{k}})}$
in the expansion of $y_{\sigma}^{(1)} $ are identical and do not cancel.

For $\ell >1$, we now assume that (1a) and (1b) hold for $\ell$ in order to prove
(1a) and (1b) for $\ell{+}1$.
We recall the coefficient $c_U^{(\ell+1)}$ 
from Eq.~\eqref{eq:chain:two:induction:finish} in Proposition~\ref{eq:chain:two:induction:iteration}.
Due to (a) and (b) being valid for $\ell$, the coefficients $c_{T}^{(\ell)}$, $c_{R}^{(\ell)}$, and $c_{S}^{(\ell)}$
can be assumed to be nonzero for suitable chosen sets $T$, $R$, and $S$.
If (a) holds for $\ell$,
we see that $c_{T}^{(\ell)}\neq0$ for $|T|=T_{\hat{k}}=k^{(\ell)}$, and analogously for $c_{R}^{(\ell)},c_{S}^{(\ell)}$.
We focus on elements in expansion of $c_U^{(\ell+1)}$
for which $\abs{U}=3S_{\hat{k}}+2=U_{\hat{k}}=k^{(\ell+1)}$ and  $T=R=S$.
For this element,
$c_{T}^{(\ell)}c_{R}^{(\ell)}c_{S}^{(\ell)} c(\ell,R,S,T,\hat{k},\hat{k})$ vanishes if and only if
$c(\ell,R,S,T,\hat{k},\hat{k})$ vanishes. We have
\begin{align*}
c(\ell,R,S,T,\hat{k},\hat{k})&=2 (\nu_{\hat{k}}^{\ell}{-}\mu_{\hat{k}}^{\ell})[\mu_{\hat{k}}^{\ell}{+}\nu_{\hat{k}}^{\ell}{-}(R{\cup}S)_{\hat{k}}]
(\nu_{\hat{k}}^{\ell}{-}\mu_{\hat{k}}^{\ell})[\mu_{\hat{k}}^{\ell}{+}\nu_{\hat{k}}^{\ell}{-}(R{\cup}S)_{\hat{k}}{-}\delta_{\hat{k},\hat{k}}]\\
&=2 (\nu_{\hat{k}}^{\ell}{-}\mu_{\hat{k}}^{\ell})^2[\mu_{\hat{k}}^{\ell}{+}\nu_{\hat{k}}^{\ell}{-}2
k^{(\ell)}][\mu_{\hat{k}}^{\ell}{+}\nu_{\hat{k}}^{\ell}{-}2 k^{(\ell)}{-}1] \neq 0
\end{align*}
as $\alpha_{\hat{k}}+\beta_{\hat{k}}=p\geq2$ and Proposition~\ref{small:helping:lemma:coefficients} imply
$\mu_{\hat{k}}^{\ell}{+}\nu_{\hat{k}}^{\ell}{-}2 k^{(\ell)}{-}1=k^{(\ell)}(p{-}2)+ p-1\geq1>0$.
The corresponding term
in the expansion of $y_{\sigma}^{(\ell+1)}$ 
for $U = R{\cup} S{\cup} T {\cup}\multisetnarrow{\hat{k},\hat{k}}$
appears only once and is linearly independent from all other terms.
If (b) holds, then
$c_{T}^{(\ell)}\neq0$ for
$|T|=T_{\hat{k}}+T_{\hat{k}'}=2T_{\hat{k}}=2T_{\hat{k}'}=k^{(\ell)}$ and analogously  for $R$ and $S$.
The element
$c_{T}^{(\ell)}c_{R}^{(\ell)}c_{S}^{(\ell)} c(\ell,R,S,T,\hat{k},\hat{k}')$ 
in the expansion of  $c_U^{(\ell+1)}$ vanishes if and only if
$c(\ell,R,S,T,\hat{k},\hat{k}')$ vanishes (and similarly for 
$c_{T}^{(\ell)}c_{R}^{(\ell)}c_{S}^{(\ell)} c(\ell,R,S,T,\hat{k}',\hat{k})$
and $c(\ell,R,S,T,\hat{k}',\hat{k})$).
However,
\begin{align*}
c(\ell,R,S,T,\hat{k},\hat{k}')=c(\ell,R,S,T,\hat{k}',\hat{k})&=
2 (\nu_{\hat{k}}^{\ell}{-}\mu_{\hat{k}}^{\ell})[\mu_{\hat{k}}^{\ell}{+}\nu_{\hat{k}}^{\ell}{-}(R{\cup}S)_{\hat{k}}]
(\nu_{\hat{k}'}^{\ell}{-}\mu_{\hat{k}'}^{\ell})[\mu_{\hat{k}'}^{\ell}{+}\nu_{\hat{k}'}^{\ell}{-}(R{\cup}S)_{\hat{k}'}]\\
&=2 (\nu_{\hat{k}}^{\ell}{-}\mu_{\hat{k}}^{\ell})(\nu_{\hat{k}'}^{\ell}{-}\mu_{\hat{k}'}^{\ell}) 
=2 ( \beta_{\hat{k}}{-}\alpha_{\hat{k}}) (\beta_{\hat{k}'}{-}\alpha_{\hat{k}'})
\neq 0,
\end{align*}
where we applied Proposition~\ref{small:helping:lemma:coefficients}; note that
$\mu_{\hat{k}}^{\ell}{+}\nu_{\hat{k}}^{\ell}{-}(R{\cup}S)_{\hat{k}}=
\mu_{\hat{k}'}^{\ell}{+}\nu_{\hat{k}'}^{\ell}{-}(R{\cup}S)_{\hat{k}'}=1$.
The terms
in the expansion
of $y_{\sigma}^{(\ell+1)}$ for $U = T{\cup}R{\cup}S {\cup}\multisetnarrow{\hat{k},\hat{k}'}$ or
$U = T{\cup}R{\cup}S {\cup}\multisetnarrow{\hat{k}',\hat{k}}$ are identical under permutations of
$T$, $R$, and $S$, and they do not cancel.

We now proceed to prove (2). The claims (1a) and (1b) verify that
$y_{\sigma}^{(\ell)}$ has always the highest possible degree.
Thus, $\deg(y_{\sigma}^{(\ell)})=\abs{\lambda^\ell}-2\abs{S}=
(k^{(\ell)}{+}1)\abs{\gamma}-2k^{(\ell)}=3^\ell\abs{\gamma}-2(3^\ell{-}1)
=3^\ell(\abs{\gamma}{-}2)+2$
as claimed.
\end{proof}

\begin{theorem}\label{commutator:chain:final}
We employ the notation of Proposition~\ref{eq:chain:two:induction:iteration}. Let us assume also that
$y_{\sigma}^{(0)} = g_{\sigma}^{\gamma}\in\hat{A}_n^\perp$. Then, $\deg(y_{\sigma}^{(\ell)})=3^\ell(\abs{\gamma}{-}2)+2\geq3^\ell+2$.
\end{theorem}
\begin{proof}
Assuming $y_{\sigma}^{(0)} = g_{\sigma}^{\gamma}\in\hat{A}_n^\perp$, Lemma~\ref{lemma:A_purp}
guarantees that one of the two conditions (a) or (b) of Proposition~\ref{Second:Chant:Type:II:lemma} holds.
Thus, Proposition~\ref{Second:Chant:Type:II:lemma} implies that
the expansion of $y_{\sigma}^{(\ell)}$ always contains  an element of maximal degree. 
We obtain $\abs{\gamma}\geq3$ due to Lemma~\ref{lemma:A_purp}
and $\deg(y_{\sigma}^{(\ell)})=3^\ell(\abs{\gamma}{-}2)+2$  from Proposition~\ref{Second:Chant:Type:II:lemma}.
Consequently, $\deg(y_{\sigma}^{(\ell)})=3^\ell(\abs{\gamma}{-}2)+2\geq3^\ell+2$.
\end{proof}

\section{Finiteness of skew-hermitian bosonic Lie algebras with independent free Hamiltonian generators}\label{main:results:section}
We are now ready to state and prove our main result. We are assuming here that the basis elements
$i a^\dag_k a_k$ belong to the set $\mathcal{G}$ that generates the Lie algebra $\g$ and, therefore, appear 
as elements of the generated Lie algebra.
Furthermore, we will exclude the presence of elements of $\hat{A}_n^1$ in $\mathcal{G}$.
All together, this means that the following theorem is the first step in the direction of the full classification of 
bosonic Lie algebras.

\begin{theorem}[Main Result]\label{main:theorem_final}
Consider the Lie subalgebra $\mathfrak{g}=\lie{ \mathcal{G}}\subseteq \hat{A}_n$ of the skew-hermitian Weyl algebra
$\hat{A}_n$ that is generated by a set of generators
\begin{equation*}
\mathcal{G}:=\{ i{a}^\dag_k{a}^{\phantom{\dagger}}_k   \text{ for } k\in \{1,\ldots,n\} \text{ and }
g_{\sigma_p}^{\gamma_p} \text{ for } p \in \mathcal{R} \text{ with } \abs{\mathcal{R}} < \infty,\;
g_{\sigma_p}^{\gamma_p}\notin\hat{A}_n^1, \text{ and } \gamma_p=(\alpha^{(p)},\beta^{(p)})\in\Np{2n}\}.
\end{equation*}
Then, $\mathfrak{g}$ is finite-dimensional if and only if all of the conditions (i)~$\mathcal{G}^\perp=\emptyset$,
(ii)~$\mathcal{G}^{=}\subseteq\gc_\mathcal{G}$,
(iii)~$[\mathcal{G}^2,\mathcal{G}^{\text{om}}]\subseteq\hat{A}_n^{\text{om}}$,
and (iv)~$\PP{[\mathcal{G}^{\text{om}},\mathcal{G}^{\text{om}}]}{\hat{A}_n^\perp}=\emptyset$
are observed.
\end{theorem}

Let us briefly discuss the cases covered by the conditions (i)-(iv) posed above. We divided the set $\mathcal{G}$ of generators
into the five subsets $\mathcal{G}^0$, $\mathcal{G}^{=}$, $\mathcal{G}^2$, $\mathcal{G}^{\text{om}}$, and $\mathcal{G}^\perp$
(refer to Section~\ref{sec:decomp:weyl}).
We could expect $15$ possible conditions on these sets due to the fact that a Lie algebra is obtained from its generating set
$\mathcal{G}$ by commutators, and there are $15$ possible commutators involving the subsets of $\mathcal{G}$ (recall that the
commutator of one subset with itself might \emph{not} vanish).
Nevertheless, we see that if we require that
$\mathcal{G}^{=}\subseteq\gc_\mathcal{G}$ it follows that also $\mathcal{G}^0\subseteq\gc_\mathcal{G}$,
and this removes $5+4=9$ commutators. Furthermore, the condition $\mathcal{G}^\perp=\emptyset$ removes another
$3$ commutators.
Therefore we are left with only $15-9-3=3$ commutators to be fixed: namely, $[\mathcal{G}^2,\mathcal{G}^2]$,
$[\mathcal{G}^2,\mathcal{G}^{\text{om}}]$, $[\mathcal{G}^{\text{om}},\mathcal{G}^{\text{om}}]$. Finally,
$[\mathcal{G}^2,\mathcal{G}^2]\subseteq \mathfrak{sp}(2n,\mathbb{R})$
which has a finite dimension.
Thus, this commutator does not add any information on the overall dimensionality of the algebra and we are left with $2$
commutators, namely $[\mathcal{G}^2,\mathcal{G}^{\text{om}}]$, $[\mathcal{G}^{\text{om}},\mathcal{G}^{\text{om}}]$.
These are precisely the remaining conditions required in the Theorem~\ref{main:theorem_final},
and therefore all possible cases are covered.
This provides insights into the general commutator structure.

We proceed with the proof. One direction of the proof is divided into
four steps. In each step we show that, by negating a specific condition,
the proof is reduced to building a Commutator Chain of Type I or II with the property that 
its elements $z^{(\ell)}_\sigma$ or $y^{(\ell)}_\sigma$ all belong to the generated Lie algebra $\g$ and have 
a monotonically increasing degree as a function of $\ell$.
This, in turn, proves that the Lie algebra $\g$ is not finite as each chain is an unbounded sequence of linearly independent elements.
Figure~\ref{Figure:One} provides a visual representation of the structure of one direction in the proof.

\begin{figure}[t]
\includegraphics[width=0.9\linewidth]{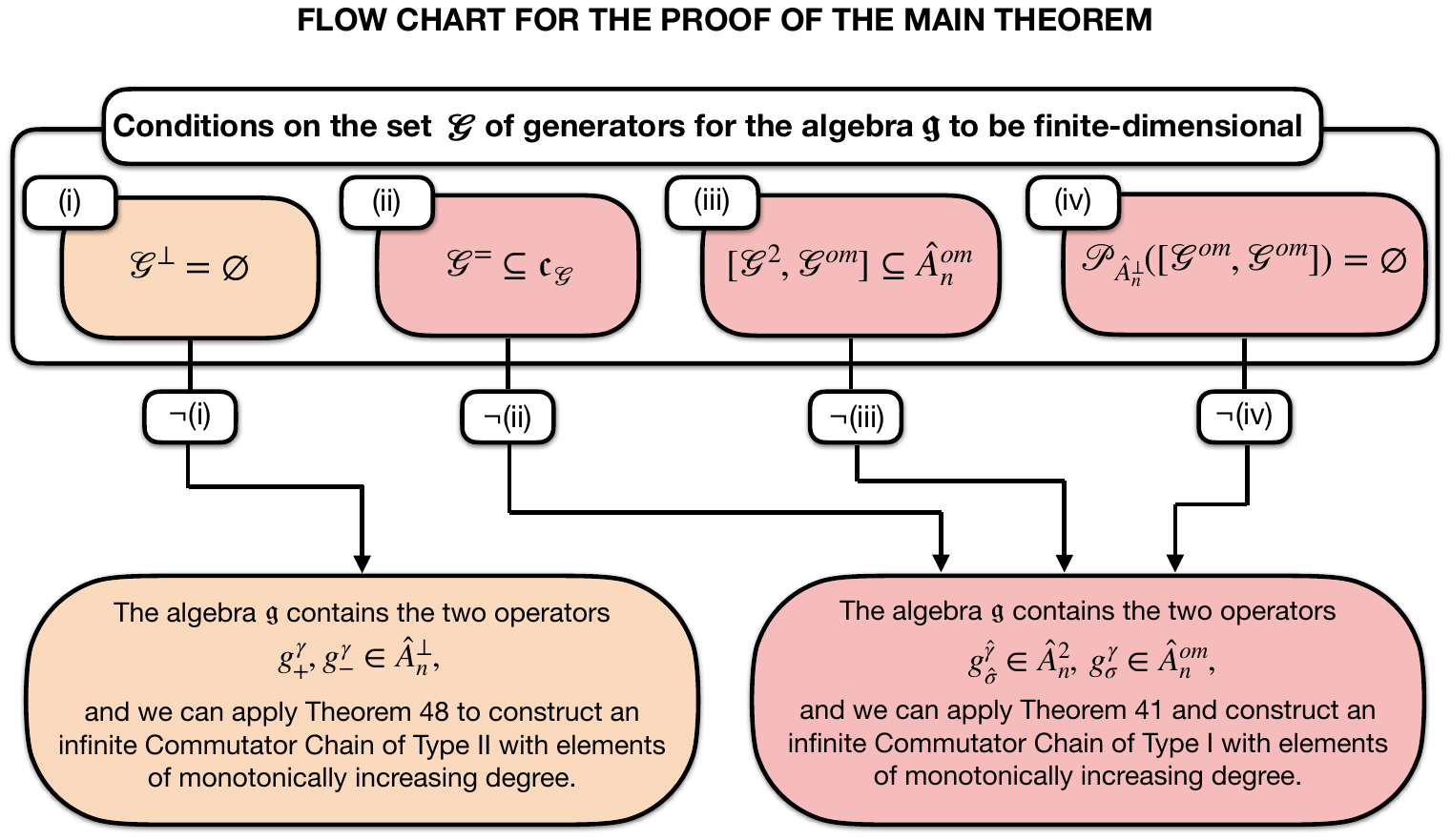}
\caption{\textbf{Flow chart for one direction in the proof of our main result as given by Theorem~\ref{main:theorem_final}}.
We proceed by showing that negating any condition as outlined in the claim results
in a infinite-dimensional  Lie algebra $\g$. The reverse direction is not depicted as a more direct proof is possible.}\label{Figure:One}
\end{figure}

\begin{proof}[Proof of Theorem~\ref{main:theorem_final}]
Below we assume in turn that one of the conditions (i) to (iv) is false and we then show that, in each of these cases, the 
generated Lie algebra $\g$ would be infinite-dimensional. The possible cases are:

First, assume that (i) does not hold, i.e., $\mathcal{G}^\perp\neq\emptyset$.
Thus, $g_{\sigma}^\gamma\in \mathcal{G}^\perp \subseteq \mathcal{G}$ with $\gamma=(\alpha,\beta)$ and $\abs{\gamma}=\abs{\alpha}+\abs{\beta}\geq 3$.
Clearly, $\alpha\neq\beta$ and $g_{\sigma}^\gamma \in \g$.
Theorem~\ref{important:one:element:and:diagonal:implies:other:two:mode:new} 
together with $\alpha\neq\beta$ implies $g_{-\sigma}^\gamma\in\mathfrak{g}$.
Theorem~\ref{commutator:chain:final} then shows that $\g$ is infinite-dimensional.

Second, assume that (ii) does not hold. This implies that $\mathcal{G}^{=}\neq\emptyset$ and that there exist two basis elements
$g_+^{\tilde{\gamma}}=2i a^{\tilde{\gamma}} \in \mathcal{G}^{=} \subseteq \mathcal{G}$
(with $\tilde{\gamma}=(\tilde{\alpha},\tilde{\alpha})$ and $\abs{\tilde{\alpha}}\geq 2$)
and $g_{\sigma}^\gamma\in\mathcal{G}$ such that
$[g_+^{\tilde{\gamma}}, g_{\sigma}^\gamma] \neq 0$.
Corollary~\ref{theorem:commutator:monomials:one:diagonal}(c) implies that 
$g_{\sigma}^\gamma\notin(\mathcal{G}^0{\cup}\mathcal{G}^{=})$, with $\gamma=(\alpha,\beta)$ and
$\alpha\neq\beta$. We assume that $g_{\sigma}^\gamma\notin\mathcal{G}^\perp$
as otherwise we could apply the argument from the first case. 
The conditions for Proposition~\ref{max:length:simple:chain:theorem} apply and we construct a
Commutator Chain of Type~I containing the elements $z^{(\ell)}_\sigma$ generated by
$g_\sigma^\gamma$  and $g_+^{\tilde{\gamma}}$.
Moreover, Corollary~\ref{theorem:commutator:monomials:one:diagonal}(c) also verifies that
there is at least one index $k$ such that $\alpha_k\neq\beta_k$ and $\tilde{\alpha}_k\neq0$.
Theorem~\ref{Second:Chant:Type:II:lemma:MIAO} thus guarantees that 
$\deg(z^{(\ell)}_\sigma)=2\ell(\abs{\tilde{\alpha}}{-}1)+\abs{\alpha}+\abs{\beta}> 2\ell +\abs{\alpha}+\abs{\beta}$.
Thus $\g$ is not finite-dimensional.

In the third part of the proof, we refer to
the notation of Section~\ref{sec:decomp:weyl}.
Here we consider two elements $g_{\hat{\sigma}}^{\hat{\gamma}}\in\mathcal{G}^2$ and
$g_\sigma^{\gamma}\in\mathcal{G}^{\text{om}}$.
As $g_\sigma^{\gamma}\in\mathcal{G}^{\text{om}}$
implies $\alpha \neq \beta$, Theorem~\ref{important:one:element:and:diagonal:implies:other:two:mode:new}(b)
shows that $g_{-\sigma}^\gamma\in \g$.
We now note that since $g_\sigma^{\gamma}\in\mathcal{G}^{\text{om}}$ there exists a unique $k$ such that
$\alpha_k+\beta_k=1$, while $\alpha_p=\beta_p$ for all $p\neq k$ and $|\gamma|\geq3$.
Before proceeding, however, let us comment on the cases where (iii) \emph{does} indeed hold.
This will help us find the conditions for the remaining cases, which are those of interest for this part of the proof.
It is easy to see that the following cases occur
if (iii) \emph{does} hold: (a) If
$g_{\hat{\sigma}}^{\hat{\gamma}}= g_{\hat{\sigma}}^\text{S}(k)$ with
$\hat{\alpha}_k+\hat{\beta}_k=2$ and $\hat{\alpha}_k\neq\hat{\beta}_k$ as well as $\hat{\alpha}_p=\hat{\beta}_p=0$ for $p\neq k$,
then $0\neq [g_{\hat{\sigma}}^{\hat{\gamma}},g_\sigma^{\gamma}] \in \hat{A}_n^{\text{om}}$.
(b) If $g_{\hat{\sigma}}^{\hat{\gamma}}= g_{\hat{\sigma}}^\text{B}(k,k')$ or
$g_{\hat{\sigma}}^{\hat{\gamma}}= g_{\hat{\sigma}}^\text{T}(k,k')$
with $\hat{\alpha}_k+\hat{\beta}_k=1$ and $\hat{\alpha}_{k'}+\hat{\beta}_{k'}=1$ and
$\hat{\alpha}_p=\hat{\beta}_p=0$ for $p\notin \{k,k'\}$ (assuming $k\neq k'$), but also $\alpha_{k'}=\beta_{k'}=0$,
then $0\neq [g_{\hat{\sigma}}^{\hat{\gamma}},g_\sigma^{\gamma}]\in \hat{A}_n^{\text{om}}$.

We are now ready to proceed and assume that (iii) does not hold.
This implies that there exists $g_{\hat{\sigma}}^{\hat{\gamma}}\in\mathcal{G}^2$ and
$g_\sigma^{\gamma}\in\mathcal{G}^{\text{om}}$
such that $0\neq [g_{\hat{\sigma}}^{\hat{\gamma}},g_\sigma^{\gamma}]\notin\hat{A}_n^{\text{om}}$.
We have:
(a') $\hat{\alpha}_{k'}+\hat{\beta}_{k'}=2$ and $\hat{\alpha}_{k'}\neq \hat{\beta}_{k'}$
for $k' \neq k$
as well as $\hat{\alpha}_p=\hat{\beta}_p=0$
for $p\neq k'$
and $\alpha_{k'} = \beta_{k'} \neq 0$ (as otherwise they would commute); or (b')
there is at least one index $k'\neq k$ for which
$\hat{\alpha}_{k'}+\hat{\beta}_{k'}=1$, and 
we must have $\alpha_{k'}=\beta_{k'}\neq0$
in both cases $g_{\hat{\sigma}}^{\hat{\gamma}}= g_{\hat{\sigma}}^\text{B}(k,k')$ and
$g_{\hat{\sigma}}^{\hat{\gamma}}= g_{\hat{\sigma}}^\text{T}(k,k')$.
Regardless of the choice of scenario (a') or (b'), we now employ Lemma~\ref{lem:omnibus}(a) and start by computing
$[g_+^{\gamma},g_-^{\gamma}]\upto{d}
\sum\nolimits_{v} (\beta_v^2 {-} \alpha_v^2)\,g_{+}^{\gamma+\gamma^{\dagger}-\ones_v}
=(\beta_{k}^2 {-} \alpha_{k}^2)\,g_{+}^{\gamma+\gamma^{\dagger}-\ones_{k}}$, 
where the first term in Lemma~\ref{lem:omnibus}(a) vanishes because $\gamma=\lambda$,
the only remaining term in the sum over $v$ is obtained for $v=k$ since $\alpha_p=\beta_p$ for all $p\neq k$, where
$d=2\abs{\gamma}-2$. Crucially, $\abs{\beta_{k}^2 {-} \alpha_{k}^2}=1\neq0$ and we define
$g_+^{\tilde{\gamma}}:=g_{+}^{\gamma+\gamma^{\dagger}-\ones_{k}}$, which has
$\tilde{\gamma}:=\gamma+\gamma^{\dagger}-\ones_{k}$, and therefore $\tilde{\alpha}_p=\tilde{\beta}_p$ for all
$p$ and, in particular, $\tilde{\alpha}_{k}=\tilde{\beta}_{k}=0$. 
Considering either case (a') or (b'), the property $\alpha_{k'}=\beta_{k'}\neq0$
implies $\tilde{\alpha}_{k'}=\tilde{\beta}_{k'}\neq0$.
Note $|\tilde{\gamma}|\geq4$ and that 
$g_+^{\tilde{\gamma}}$ might not be contained in $\g$ itself.
Therefore, since we have that $\tilde{\alpha}_{k'}=\tilde{\beta}_{k'}\neq0$, and
either $\hat{\alpha}_{k'}+\hat{\beta}_{k'}=2$ and 
$\hat{\alpha}_{k'}\neq\hat{\beta}_{k'}$
or $\hat{\alpha}_{k'}+\hat{\beta}_{k'}=1$
with $k'\neq k$ (i.e.\ $\hat{\alpha}_{k'}\neq\hat{\beta}_{k'}$ for $k'\neq k$),
we see that $[[g_+^{\gamma},g_-^{\gamma}],g_\sigma^{\hat{\gamma}}]\upto{d'}[g_+^{\tilde{\gamma}},
g_\sigma^{\hat{\gamma}}]$ with $d'=|\gamma|+|\tilde{\gamma}|-2$. 
Thus, we can essentially construct a Commutator Chain of Type I with seed operators $g_+^{\tilde{\gamma}}\upto{d}[g_+^{\gamma},g_-^{\gamma}]$ and $g_\sigma^{\hat{\gamma}}$
(for which we have either $\hat{\alpha}_{k'}+\hat{\beta}_{k'}=2$ with $\hat{\alpha}_{k'}\neq\hat{\beta}_{k'}$ or $\hat{\alpha}_{k'}+\hat{\beta}_{k'}=1$), where
$z_\sigma^{\ell+1}:=[[g_+^{\gamma},g_-^{\gamma}],z_\sigma^\ell]\upto{d^{(\ell+1)}}[g_+^{\tilde{\gamma}},z_\sigma^\ell]$
for $d^{(\ell+1)}=2(\ell{+}1)(\abs{\tilde{\alpha}}{-}1)+\abs{\alpha}+\abs{\beta}$. 
We obtain that $g_+^{\tilde{\gamma}},g_\sigma^{\hat{\gamma}}$
satisfy the conditions of Theorem~\ref{Second:Chant:Type:II:lemma:MIAO}
since
$\tilde{\alpha}_{k'}=\tilde{\beta}_{k'}\neq0$, $\hat{\alpha}_{k'}+\hat{\beta}_{k'}=2$
or $\hat{\alpha}_{k'}+\hat{\beta}_{k'}=1$, and $|\tilde{\alpha}|\geq2$.  
We conclude that $\deg(z_\sigma^{\ell+1})=\deg([g_+^{\tilde{\gamma}},z_\sigma^\ell])$, which is maximal, and it follows that $\g$
contains all elements $z_\sigma^\ell$ of monotonically increasing order as a function of $\ell$. Therefore, it is 
infinite-dimensional.

Fourth, assume that (iv) does not hold.
This means that there exist two operators
$g_\sigma^{\gamma},g_{\hat{\sigma}}^{\hat{\gamma}}\in\mathcal{G}^{\text{om}}$
with $\gamma=(\alpha,\beta)$ and $\hat{\gamma}=(\hat{\alpha},\hat{\beta})$ such that $[g_{\hat{\sigma}}^{\hat{\gamma}},g_\sigma^{\gamma}]\neq0$ since
$\PP{[g_{\hat{\sigma}}^{\hat{\gamma}},g_\sigma^{\gamma}]}{\hat{A}_n^\perp}\neq\emptyset$. As
$g_{\hat{\sigma}}^{\hat{\gamma}},g_\sigma^{\gamma}\in\mathcal{G}^{\text{om}}$
implies $\alpha \neq \beta$ and $\hat{\alpha}\neq\hat{\beta}$,
Theorem~\ref{important:one:element:and:diagonal:implies:other:two:mode:new}(b)
shows that $g_{-\sigma}^{\gamma},g_{-\hat{\sigma}}^{\hat{\gamma}}\in \g$. Furthermore, we
know that there are unique $k,\hat{k}$ such that $\alpha_k+\beta_k=1$ and
$\hat{\alpha}_{\hat{k}}+\hat{\beta}_{\hat{k}}=1$, while $\alpha_j=\beta_j$ and $\hat{\alpha}_p=\hat{\beta}_p$,
for all $j\neq k$, $p\neq \hat{k}$ respectively. Crucially, it is immediate to verify using
Lemma~\ref{lem:omnibus} that if $k=\hat{k}$ then
$[g_{\hat{\sigma}}^{\hat{\gamma}},g_\sigma^{\gamma}]\in\hat{A}_n^=$ and
therefore $\PP{[g_{\hat{\sigma}}^{\hat{\gamma}},g_\sigma^{\gamma}]}{\hat{A}_n^\perp}=\emptyset$
against the hypothesis.
Thus, we also have $k\neq\hat{k}$.
We can use 
Lemma~\ref{lem:omnibus}(a) 
and obtain 
\begin{align*}[g^{\hat{\gamma}}_{+},g_{-}^{\gamma}]
&\upto{d}\sum\nolimits_{v}
(\alpha_v \hat{\beta}_v {-} \beta_v \hat{\alpha}_v)\, g_{+}^{\hat{\gamma}+\gamma-\ones_v}
+\sum\nolimits_{v} (\beta_v \hat{\beta}_v {-} \alpha_v \hat{\alpha}_v)\,g_{+}^{\Theta(\hat{\gamma}
+\gamma^{\dagger})-\ones_v}\\
&\upto{d}(\alpha_k \hat{\beta}_k {-} \beta_k \hat{\alpha}_k)\, g_{+}^{\hat{\gamma}+\gamma-\ones_k}
+(\beta_k \hat{\beta}_k {-} \alpha_k \hat{\alpha}_k)\,g_{+}^{\Theta(\hat{\gamma}+\gamma^{\dagger})
-\ones_k}+(\alpha_{\hat{k}} \hat{\beta}_{\hat{k}} {-} \beta_{\hat{k}} \hat{\alpha}_{\hat{k}})\, g_{+}^{\hat{\gamma}+\gamma-\ones_{\hat{k}}}
+(\beta_{\hat{k}} \hat{\beta}_{\hat{k}} {-} \alpha_{\hat{k}} \hat{\alpha}_{\hat{k}})\,g_{+}^{\Theta(\hat{\gamma}+\gamma^{\dagger})
-\ones_{\hat{k}}}\\
&\upto{d}(\alpha_k  {-} \beta_k )\hat{\alpha}_k(g_{+}^{\hat{\gamma}+
\gamma-\ones_k}-g_{+}^{\Theta(\hat{\gamma}+\gamma^{\dagger})-\ones_k})
+\alpha_{\hat{k}}( \hat{\beta}_{\hat{k}} {-} \hat{\alpha}_{\hat{k}})\,( g_{+}^{\hat{\gamma}+\gamma-\ones_{\hat{k}}}
+\,g_{+}^{\Theta(\hat{\gamma}+\gamma^{\dagger})-\ones_{\hat{k}}}).
\end{align*}
This expression vanishes when $\alpha_{\hat{k}}=\beta_{\hat{k}}=\hat{\alpha}_k=\hat{\beta}_k=0$, and 
it is immediate to see that this also
implies $[g_{+}^{\hat{\gamma}},g_{-}^{\gamma}]=0$ as well as $[g_{\hat{\sigma}}^{\hat{\gamma}},g_\sigma^{\gamma}]=0$ (in fact, in this case
$\alpha=\beta$ and $\hat{\alpha}=\hat{\beta}$), which is against the hypothesis.
Consequently, at least either
$\alpha_{\hat{k}}=\beta_{\hat{k}}\neq0$, $\hat{\alpha}_k=\hat{\beta}_k\neq0$, or both. We assume here that
$\hat{\alpha}_k=\hat{\beta}_k\neq0$.
We now employ Lemma~\ref{lem:omnibus}(a) and compute
$[g_+^{\hat{\gamma}},g_-^{\hat{\gamma}}]\upto{d}
\sum\nolimits_{v} (\hat{\beta}_v^2 {-} \hat{\alpha}_v^2)\,g_{+}^{\hat{\gamma}+\hat{\gamma}^{\dagger}-\ones_v}
=(\hat{\beta}_{\hat{k}}^2 {-} \hat{\alpha}_{\hat{k}}^2)\,g_{+}^{\hat{\gamma}+\hat{\gamma}^{\dagger}-\ones_{\hat{k}}}$, 
where the first term in Lemma~\ref{lem:omnibus}(a) vanishes because $\hat{\gamma}=\hat{\lambda}$,
the only remaining term in the sum over $v$ is obtained for $v=\hat{k}$ since $\hat{\alpha}_p=\hat{\beta}_p$ for all
$p\neq\hat{k}$ and $d=2|\hat{\gamma}|-2$. Crucially,
$|\hat{\beta}_{\hat{k}}^2 {-} \hat{\alpha}_{\hat{k}}^2|=1\neq0$ and we define
$g_+^{\tilde{\gamma}}:=g_{+}^{\hat{\gamma}+\hat{\gamma}^{\dagger}-\ones_{\hat{k}}}$ with
$\tilde{\gamma}:=\hat{\gamma}+\hat{\gamma}^{\dagger}-\ones_{\hat{k}}$, and therefore
$\tilde{\alpha}_p=\tilde{\beta}_p$ for all $p$ and, in particular,
$\tilde{\alpha}_{\hat{k}}=\tilde{\beta}_{\hat{k}}=0$. Moreover, $\hat{\alpha}_k=\hat{\beta}_k\neq0$
implies $\tilde{\alpha}_k=\tilde{\beta}_k\neq0$.
Finally, we see that $[[g_+^{\hat{\gamma}},g_-^{\hat{\gamma}}],g_\sigma^{\gamma}]\upto{d'}[g_+^{\tilde{\gamma}},
g_\sigma^{\gamma}]$ with $d'=|\gamma|+|\tilde{\gamma}|-2$.
Thus, we can essentially construct a Commutator Chain of Type I starting from operators $g_\sigma^{\gamma}$
(where $\alpha_k+\beta_k=1$ implies $\alpha_k\neq\beta_k$)
and $g_+^{\tilde{\gamma}}\upto{d}[g_+^{\hat{\gamma}},g_-^{\hat{\gamma}}]$, where
$z_\sigma^{\ell+1}:=[[g_+^{\hat{\gamma}},g_-^{\hat{\gamma}}],z_\sigma^\ell]\upto{d^{(\ell+1)}}[g_+^{\tilde{\gamma}},z_\sigma^\ell]$
for $d^{(\ell+1)}=2(\ell{+}1)(\abs{\tilde{\alpha}}{-}1)+\abs{\alpha}+\abs{\beta}$. It follows that, since
$\tilde{\alpha}_k=\tilde{\beta}_k\neq0$ and $\alpha_k+\beta_k=1$, we have that $g_+^{\tilde{\gamma}},g_\sigma^{\gamma}$
satisfy the conditions of Theorem~\ref{Second:Chant:Type:II:lemma:MIAO}. 
We conclude that $\deg(z_\sigma^{\ell+1})=\deg([g_+^{\tilde{\gamma}},z_\sigma^\ell])$, which is maximal, and it follows that $\g$
contains all elements $z_\sigma^\ell$ of monotonically increasing order as a function of $\ell$. Therefore it is 
infinite-dimensional.

We now prove the other direction, i.e., that conditions (i)-(iv), together with the finiteness of $\mathcal{G}$,
imply that $\g$ is finite-dimensional.
The idea behind this part of the proof is the following: given a set $\mathcal{G}$ of generators,
we construct the full Lie algebra by means of linear combinations and commutators. 
The key aspects to consider here are the commutators of the generators,
and how the dimensionality of the Lie algebra can be enlarged using commutators.
We assume the conditions (i)-(iv), and that $\mathcal{G}$ is finite-dimensional. 
We need to consider the five possible subsets $\mathcal{G}^{\text{K}}$ of $\mathcal{G}$.
By assumption $\mathcal{G}^\perp$ is empty and $\mathcal{G}^=\subseteq\gc_\mathcal{G}\subseteq\gc$. 
This immediately implies that the finite-dimensionality of $\g$ is independent of these sets.
We next note that $\mathcal{G}^0$ does not pose a concern because we know that for each $g_\sigma^\gamma\in\mathcal{G}$,
$[\mathcal{G}^0,g_\sigma^\gamma] = r g_{-\sigma}^\gamma$ for $r\in \R$, and thus $g_{-\sigma}^\gamma\in\g$. 
Thus commutators of elements of a finite-dimensional subset of $\g$ with elements of $\mathcal{G}^0$
can only add a finite set of elements. 
We therefore are left with two sets to consider, namely $\mathcal{G}^2$ and $\mathcal{G}^{\text{om}}$. 
We start by observing that $\mathcal{G}^0{\cup}\mathcal{G}^2$ generates a finite-dimensional
Lie algebra that is a subset of the finite-dimensional Lie algebra $\mathfrak{sp}(2n,\mathbb{R})$
of the symplectic group $\mathrm{Sp}(2n,\mathbb{R})$.
It is easy to verify using Lemma~\ref{lem:omnibus} that condition (iv) implies that
$[\mathcal{G}^{\text{om}},\mathcal{G}^{\text{om}}]\subseteq\hat{A}^=_n$. This observation, together
with condition (iii) and Lemma~\ref{lem:omnibus}, immediately implies that
$[\mathcal{G}^2,[\mathcal{G}^{\text{om}},\mathcal{G}^{\text{om}}]]=0$. 
To see this last statement, we recall the proof of the fourth part above, 
and assume that there exist at least two different operators $g_\sigma^{\gamma},g_{\hat{\sigma}}^{\hat{\gamma}}\in\mathcal{G}^{\text{om}}$
with $\gamma=(\alpha,\beta)$ and $\hat{\gamma}=(\hat{\alpha},\hat{\beta})$. 
The case where there is only one $g_\pm^{\gamma}\in\mathcal{G}^{\text{om}}$ is trivially true. 
For each pair $g_\sigma^{\gamma},g_{\hat{\sigma}}^{\hat{\gamma}}\in\mathcal{G}^{\text{om}}$ we have that $[g_\sigma^{\gamma},g_{\hat{\sigma}}^{\hat{\gamma}}]\in\hat{A}_n^=$ due to condition (iv). Using explicit expressions as discussed in the proof of part four above, it is immediate to see that we must have two index sets $S^=$ an $S^1$ such that $S^=\cup S^1=\{1,...,n\}$ and $S^=\cap S^1=\emptyset$, which implies the following: for each operator $g_\sigma^{\gamma}\in\mathcal{G}^{\text{om}}$, one has that there is a unique $k\in S^1$ such that $\alpha_k+\beta_k=1$, while $\alpha_p=\beta_p$ for all $p\in S^=$. Clearly $p\neq k$ for all $p$ since $S^=\cap S^1=\emptyset$. A key consequence at this point is that $[g_\sigma^{\gamma},g_{\hat{\sigma}}^{\hat{\gamma}}]=\sum_p g_{\sigma_p}^{\gamma_p}\in\hat{A}_n^=$ with $\gamma_p=(\alpha^{(p)},\beta^{(p)})$, and $\alpha_q^{(p)},\beta_q^{(p)}$ are non vanishing only for $p\in S^=$. Thus, 
using the logical steps outlined in the proof of part three, condition (iii) implies that for all elements $g_{\sigma'}^{\gamma'}\in\mathcal{G}_n^2$ one has $\alpha_p\neq0$ and $\beta_p\neq0$ only for $p\in S^1$. Consequently, the elements in $\mathcal{G}^2$ have non zero support only in $S^1$ while $[\mathcal{G}^{\text{om}},\mathcal{G}^{\text{om}}]$ has non zero support only in $S^=$, thus $[\mathcal{G}^2,[\mathcal{G}^{\text{om}},\mathcal{G}^{\text{om}}]]=0$.
All together this gives a finite-dimensional
Lie algebra $\g$, since there are only finitely many independent operators in the commutators
$[\mathcal{G}^2,\mathcal{G}^{\text{om}}]$, namely at most $(n{-}1)\times\abs{\mathcal{G}^{\text{om}}}$.
This is implied as each element $g_\sigma^\gamma\in \hat{A}_n^{\text{om}}$
has a unique $k$ such that $\alpha_k+\beta_k=1$ while $\alpha_p=\beta_p$ for all $p\neq k$
and $\abs{\gamma}\geq3$. Thus, we assume that there are $m\leq n-2$ indices $p_v$ such that
$\alpha_{p_v}=\beta_{p_v}=0$ for $v\in \{1,\ldots,m\}$. We construct $m+1$ operators
$g_\sigma^{\gamma^{(u)}}\in \hat{A}_n^{\text{om}}$ such that $\alpha^{(u)}_j+\beta^{(u)}_j=1$ for
$j\in\{k,p_1,\ldots,p_m\}$ and $\alpha^{(u)}_p=\beta^{(u)}_p=\alpha_p=\beta_p\neq0$ otherwise.
There are only finitely many independent operators in
$[\mathcal{G}^{\text{om}},\mathcal{G}^{\text{om}}]$, namely at most
$(n{-}1)^2\times\abs{\mathcal{G}^{\text{om}}}^2$. Including all possible commutators with
$\mathcal{G}^0$ completes the proof.
\end{proof}

\section{Discussion and outlook}\label{discussion:outlook:section}
We continue by discussing our results and exploring future directions as an outlook.
Many aspects of interest remain open and have not yet been resolved here in their full
generality.

\subsection{Discussion}
The first natural next step beyond the present work is to consider cases where the free  
terms $i a^\dag_k a_k$ appear among the generators of the Lie algebra of interest only as a  time-independent 
linear combination. This property is common in many concrete cases of physical interest, where a
Hamiltonian $H(t)$ reads $H(t)=H_0+H_{\text{I}}(t)$ and $H_0=\sum_k \omega_k a^\dag_k a_k$, with
$\omega_k$ time-independent real numbers (usually referred to as the \textit{bare frequencies} of the oscillators).
In order to understand the subtleties related to this issue, a simple example helps us to highlight the different possibilities:

\begin{tcolorbox}[colback=orange!3!white,colframe=orange!85!black,title=EXAMPLE]
\begin{example}\label{example:bad}
Let us consider the elements $x_1:=ia^\dag a+ib^\dag b+2ic^\dag c$ and $x_+:=i(a^\dag b^\dag c+a b c^\dag)$, as well as the two sets of generators: $\mathcal{G}_1=\{ia^\dag a, ib^\dag b, ic^\dag c, x_+\}$
and $\mathcal{G}_2=\{x_1, x_+\}$. It is easy to see that commuting the generators in $\mathcal{G}_1$ leads to the additional Lie-algebra elements
$x_-=a^\dag b^\dag c-a b c^\dag$ , $[x_+,x_-]=2i(a^\dag ac^\dag c+b^\dag bc^\dag c-a^\dag ab^\dag b)$, etc.
It turns out that the generators in $\mathcal{G}_1$ give rise to an infinite-dimensional Lie algebra $\g$ (see Main Result above). 
On the other hand, we can
immediately see from $\mathcal{G}_2$ that $[x_1,x_+]=0$, and thus the Lie algebra generated by the elements of $\mathcal{G}_2$ is finite.  
\end{example}
\end{tcolorbox}

It is clear that our main result Theorem~\ref{main:theorem_final} 
cannot be applied to the second set of generators $\mathcal{G}_2$
of Example~\ref{example:bad}, and we will need to develop a more general 
variant in order to address this and similar cases.
In fact, one key aspect in Theorem~\ref{main:theorem_final} is that the elements $i a^\dag_k a_k$
are contained in the set of generators and are consequently contained individually 
in the generated Lie algebra $\g$. For each  $g_\sigma^\gamma\in\g$ that is not an element of $\hat{A}^0_n$ or in$\hat{A}^=_n$,
this then implies that $g_{-\sigma}^\gamma\in\g$, which is a crucial step in our proof.
For $\mathcal{G}_2$, however, a generator
$x_+=i(a^\dag b^\dag c+a b c^\dag)\in\mathcal{G}_2$ with $x_+\in\hat{A}^\perp_3$
occurs, but the elements
$i a^\dag_k a_k$ do not appear individually in $\mathcal{G}_2$ and instead appear only as part of the linear combination $x_1:=ia^\dag a+ib^\dag b+2ic^\dag c$
with the key property that $[x_1,x_+]=0$. Thus, it is not true that $g_+^\gamma= x_2\in\g$ implies that
$g_-^\gamma\in\g$ as required in the proof of our main result. In the context of the commutator chains (see Section~\ref{commutator:chain:section}), 
the existence of complementary elements (see Section~\ref{subsection:complementary}) is not guaranteed
using only the fixed linear combination $x_1=ia^\dag a+ib^\dag b+2ic^\dag c$, and therefore the proof
does not immediately generalize.

Hybrid scenarios might also occur. For example, a subset of the operators $ia_k^{\dagger}a_k$
appears in the set of generators $\mathcal{G}$ only in a particular linear combination $x$, 
while the remaining ones appear individually. In other words, a subset of the operators $ia_k^{\dagger}a_k$ is part of the drift, while the remaining ones are part of the control.
Such situations are quite common whenever the frequency of some of the harmonic oscillators is externally driven.
An example is given by the Hamiltonian
$H(t)=\omega_1 a_1^\dag a_1+\omega_2 a^\dag_2 a_2+g(t) (a^\dag_1{+}a_1)^2
+((a^\dag_2)^2+(a_2)^2)$, which contains a term of the form $g(t) (a^\dag_1{+}a_1)^2$ where $g(t)$ is a
time dependent coupling. It is clear that $g(t) (a^\dag_1{+}a_1)^2=g(t) ((a^\dag_1)^2{+}(a_1)^2)+2g(t)a^\dag_1 a _1+g(t)$.
Thus the contribution $2g(t)a^\dag_1 a _1$ naturally adds to the free Hamiltonian part
$\omega_1 a^\dag_1 a_1$. This gives rise to an overall time-dependent contribution $(\omega_1+2g(t))a_1^\dag a_1$
affecting only the first mode.  

It is also natural to seek simple examples of Hamiltonians that satisfy the conditions outlined in our main result
in Theorem~\ref{main:theorem_final}.
For the one-mode case of $\hat{A}_1$, $\mathcal{G}_{\text{om}}=\emptyset$ always holds.
Thus finite bosonic Lie algebras of $\hat{A}_1$ contain: (i) the Lie algebra $\mathfrak{sp}(2,\R)$
of the symplectic group $\mathrm{Sp}(2,\mathbb{R})$; and
(ii) all finite abelian subalgebras of $\hat{A}^0_1\oplus\hat{A}_1^=$ (since $\hat{A}^0_1\oplus\hat{A}_1^=$ is itself abelian).
The corresponding Lie algebra $\mathfrak{sp}(2,\R)$ is one possible finite-dimensional
Lie subalgebra of the skew-hermitian Weyl algebra
$\hat{A}_1$ and this case is quite similar to classification of the finite-dimensional (complex) Lie subalgebras
of the one-mode Weyl algebra $A_1$ as discussed in the Introduction, or in the literature \cite{TST:2006}.

Which Hamiltonians satisfy the conditions of the main result and therefore result in a finite-dimensional Lie algebra?
We cannot provide a simple list of Hamiltonians with this property because a complete classification
eludes us at the current stage of understanding this topic. Nevertheless, we present a simple nontrivial example
of a Hamiltonian with two modes that satisfies the conditions of Theorem~\ref{main:theorem_final},
and thus leads to a finite-dimensional Lie algebra. It reads
\begin{align*}
H=\omega_a a^\dag a+\omega_b b^\dag b+ A_{a^\dag a}+\lambda_+((b^\dag)^2{+}b^2)
+i\,\lambda_-((b^\dag)^2{-}b^2)+g_+ A^+_{a^\dag a}(b^\dag{+}b)+i\,g_- A^-_{a^\dag a}(b^\dag{-}b),
\end{align*}
where $A_{a^\dag a},A^\pm_{a^\dag a}$ are (potentially time-dependent) polynomials in $a^\dag a$ and
$\lambda_\pm,g_\pm$ are (potentially time-dependent) real coefficients. The finiteness of this Lie algebra
can be verified directly by computing all commutators. With $A_{a^\dag a}=\lambda_\pm=g_-=0$ and
$A^+_{a^\dag a}=a^\dag a$, this is nothing more than the standard \emph{optomechanical Hamiltonian}
$H=\omega_a a^\dag a+\omega_b b^\dag b+ g_+ a^\dag a(b^\dag+b)$ modulo the sign in front of the coupling term \cite{Aspelmeyer:Kippenberg:2014}.
More complicated Hamiltonians can be constructed for higher-dimensional cases, but we leave their discovery to the reader.

We also want to comment on the nature of the two Commutator Chains of Type I and II 
that we have introduced in this work (see Section~\ref{commutator:chain:section}).
It is evident that these are not the only possible chains that can be defined, but only two of many.
It is not clear if there are any criteria to prefer particular commutator chains above others. We do note,
however, that the two chains introduced here seem to arise naturally from the simpler structure of the
skew-hermitian Weyl algebra $\hat{A}_n$. In fact, on the one hand a
Commutator Chain of Type I is the natural extension of
commutators of the form $[i a^\dag_k a_k, g_\sigma^\gamma]=\sigma(\alpha_k-\beta_k) g_{-\sigma}^\gamma$ 
(see Theorem~\ref{important:one:element:and:diagonal:implies:other:two:mode:new}).
As already noted above, Commutator Chains of Type I
contain this simpler case, which is obtained when $|\tilde{\alpha}|=1$ (or, equivalently, $|\tilde{\gamma}|=2$). In that case, the degree of the element
$z^{(\ell)}_\sigma$ of the chain is $\deg(z^{(\ell)}_\sigma)=2\ell(|\tilde{\alpha}|-1)+|\alpha|+|\beta|=|\alpha|+|\beta|$
which is constant as expected. Therefore, replacing $i a^\dag_k a_k$ with a generic element
$g_+^{\tilde{\gamma}}\in\hat{A}^=_n$ gives us a
Commutator Chain of Type I with elements of ever-increasing degree.
We could even say that the operators $i a^\dag_k a_k$ are those such that
Commutator Chains of Type I have only elements of bounded (and constant) degrees.
On the other hand, Commutator Chains of Type II are a natural extension of commutators
of the form $[g_+^\gamma,g_-^\gamma]$, which observe the constraint
$[ia^\dag_k a_k,[g_+^\gamma,g_-^\gamma]]=0$ for all $k$. In this case, as noted before, we have that this chain
contains a simpler case, in particular when $|\gamma|=2$.
Furthermore, in such case the degree of an element $y^{(\ell)}_\sigma$ in the chain is given by
$\deg(y^{(\ell)}_\sigma)=3^\ell(|\gamma|-2)+2=2$ which is constant as expected. 
This is not surprising since we already know that elements of degree two form the finite-dimensional Lie algebra
$\mathfrak{sp}(2n,\mathbb{R})$. There are also chains such that
$z_\sigma^{(\ell)}=0$ or $y_\sigma^{(\ell)}=0$ for all $\ell\geq\ell_*$, where $\ell_*$ is a fixed positive natural number.
This is obviously the case if there is an $\ell_*$ such that $z_\sigma^{(\ell_*)}=0$ or $y_\sigma^{(\ell_*)}=0$, since all
further elements are constructed recursively using commutators only. An example is given by
the Commutator Chains of Type II with $y_+^{(0)}=g_+^\gamma=i (a^\dag a b^\dag+ a^\dag a b)$ and
$y_-^{(0)}=g_-^\gamma=a^\dag a b-a^\dag a b^\dag$. It follows that $y_\sigma^{(\ell)}=0$ for $\ell\geq1$. 
For these reasons, we believe that the chosen chains are natural to the structure of the theory, but more
can be constructed depending on the specific need at hand. Figure~\ref{Figure:Sunny:Side:Up}
visualizes the different cases as just discussed.

\begin{figure}[t]
\includegraphics[width=\linewidth]{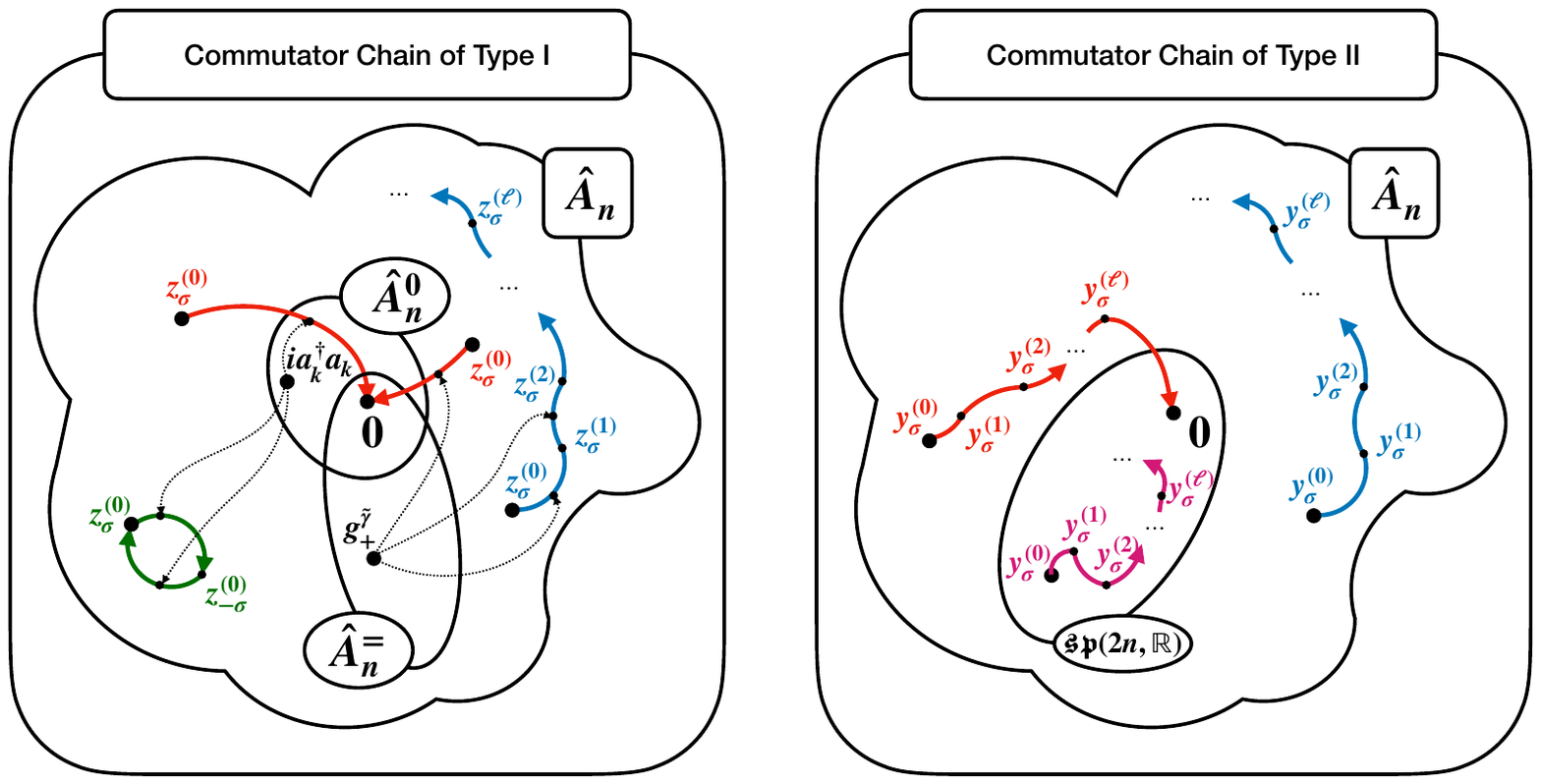}
\caption{\textbf{General properties of Commutator Chains of Type I and II}. Different classes occurring as
commutator chains are highlighted to provide  a qualitative understanding of these mathematical objects.
For both types, there are three major
behaviors: (i) red chains are those for which $z^{(\ell)}_\sigma=0$ or $y^{(\ell)}_\sigma=0$
for $\ell\geq l_*$. For type I, it always occurs that $\ell_*=1$, while 
$\ell_*$ can take other values for type II. (ii) Blue chains are those for which the degrees
$\deg(z^{(\ell)}_\sigma)$ and $\deg(y^{(\ell)}_\sigma)$ increase monotonically and without bound as a function of
$\ell$. (iii) Finally, the last type of chain differentiates between the green chains of
Type I for which $|\tilde{\alpha}|=1$ (or $|\tilde{\gamma}|=2$) and the purple chains of
Type II for which $|\gamma|=2$.
Note that here $\mathfrak{sp}(2n,\mathbb{R})$ is the Lie algebra of the symplectic group.}\label{Figure:Sunny:Side:Up}
\end{figure}

Another aspect to consider is captured by the following question: \emph{how precisely can
the dynamics be obtained or simulated if there is no exact factorization of the dynamics}?
This question does not have a straightforward answer. One reason is that,
already in the cases that admit a finite factorization, the ordering of the terms in the
factorization does play an important role and a particular choice of ordering might not lead to global solutions
(i.e, solutions valid for all times). This globality of solutions needs to be verified independently 
even in favorable finite-dimensional cases (see, e.g., Appendix B in \cite{YZPGK:2015}). Moreover,
finite-dimensional Lie subalgebras of the skew-hermitian Weyl algebra are often not compact Lie algebras
\cite{Joseph:1972}, which can lead to further complications
as the exponential mapping might no longer be surjective (see, e.g., Section~2.8 in \cite{Elliott:2009}).
This suggests that the question of approximating the
dynamics for infinite-dimensional system has a threefold perspective: First, the ordering of the terms
in the factorization $U(t)=\prod_k \exp[-i u_k(t) G_k]$ might influence the globality
of the solution. Second, given the factorization $U(t)=\prod_k \exp[-i u_k(t) G_k]$, the overall contribution to the
dynamics given by the terms $\prod_{k>k_*} \exp[-i u_k(t) G_k]$ might dramatically depend on the choice of
ordering of the first $k_*$ terms. This point is crucial when simulating or approximating
the dynamics and it is at the core of the quest for minimization the effects of unwanted coupling terms. 
Finally, it is unclear how to estimate the multiplicity of a particular
element in the factorization. It can be tempting to think that it should always be possible to write the
desired expression $U(t)=\prod_k \exp[-i u_k(t) G_k]$ with the minimal possible number
of terms, i.e., no repeated operators $G_k$. On the one hand, in the case of finite-dimensional Lie algebras $\g$,
the number of degrees of freedom in the
Hamiltonian $H(t)$ matches the dimension of the Lie algebra, and therefore it also matches
the number of terms in the factorized expression of the time-evolution operator \cite{Wei:Norman:1963}.
Thus, we can expect each operator $G_k$ to appear only once. On the other hand, however, in the case of infinite
factorization, each operator might appear multiple times in different positions. A concrete example would be
$U(t)=\exp[-i u^{(0)}_1(t)(t) G_1]\exp[-i u_2(t) G_2]\exp[-i u^{(1)}_1(t) G_1]\prod_{k\geq3} \exp[-i u_k(t) G_k]$,
where $u^{(1)}_1(t)\neq u^{(0)}_1(t)$ in general. 
A priori, such multiple occurrences in this latter case
are not forbidden. This is yet another aspect that can influence the precision
in a simulation when attempting to truncate a factorization. 

This work relies strongly on the relation developed and refined
in Lemma~\ref{app:lemma:powers}, Corollary~\ref{theorem:commutator:monomials:one:diagonal} and
Theorem~\ref{theorem:commutator:monomials}. This relation gives a simple condition to
verify if the commutator of two monomials or basis elements will have the maximally possible degree.
As discussed in the Introduction, this is similar in spirit as results on the commutator
for the Weyl algebra found in the literature \cite{Dixmier:1968,Igusa:1981a,Joseph:1974a}.
It plays a important
role in our work and allows us to focus on highest degree components of an arbitrary element
$x=\sum_p C_p g_{\sigma_p}^{\gamma_p}\in\hat{A}_n$. We no longer need to take
into account the full structure of such arbitrary elements, which would make tackling the questions posed
here much more difficult, or outright impossible.

\subsection{Outlook}
We now provide a brief overview of the outlook of this work.
As a first next step we mention the extension of our results to Hamiltonians
with  a drift containing free elements $ia^\dag_k a_k$. This implies that the drift will contain time-independent
linear combinations $x=\sum_k c_k ia^\dag_k a_k$ of the operators $i a^\dag_k a_k$, where the $c_k$ are real numbers that are not all zero.
We expect that the number of classes of Hamiltonians with such drift that enjoy a finite factorization of their dynamics will increase.
This can be understood from the fact that, in the case of the drift being absent, each free element $i a^\dag_k a_k$ is independently
an element of the Lie algebra $\g$. Therefore, $\g$ includes the commutators of these elements
with all others. However, when we consider that only a time-independent combination $x=\sum_k c_k ia^\dag_k a_k$
of the operators $i a^\dag_k a_k$ is an element of the Lie algebra, it follows that only
commutators between $x$ and all others must be considered. Let us call
$\mathfrak{g}=\lie{\mathcal{G}}\subseteq \hat{A}_n$ and $\mathfrak{g}'=\lie{ \mathcal{G}'}\subseteq \hat{A}_n$,
where $\mathcal{G}:=\{i a^\dag_k a_k \,\, \text{for all} \,\, k,\,\,
\text{and}\,\,g_{\sigma_p}^{\gamma_p}\,\,\text{where}\,\,p\in\mathcal{P}\,\,
\text{with}\,\,|\mathcal{P}|<\infty\}$ and $\mathcal{G}^{\prime}:=\{x=\sum_k c_k ia^\dag_k a_k,\,\,
\text{and}\,\,g_{\sigma_p}^{\gamma_p}\,\,\text{where}\,\,p\in\mathcal{P}\,\,
\text{with}\,\,|\mathcal{P}|<\infty\}$. It should be evident that $\g$ contains each $i a^\dag_k a_k$ as an element,
while $\g'$ contains the $i a^\dag_k a_k$ only via their
particular linear combination $x$ with constant coefficients $c_k$. Both algebras include the \emph{same}
finite set of generators $ g_{\sigma_p}^{\gamma_p}$. It is clear that, since a Lie algebra includes also all
possible linear combinations of its elements, as well as all commutators, then $\g{}^{\prime}\subseteq\g$.
The equality can be obtained in particular when $n=1$, i.e., there is only one bosonic degree of freedom.
Note that if $\g$ has a finite dimension, the same applies to $\g'$. However, since Lie algebras $\g'$ are contained
in the Lie algebras $\g$ as argued here, we can expect that there will be more
finite-dimensional Lie algebras $\g'$ than finite-dimensional Lie algebras $\g$.
This supports the claim made above and shows the importance of extending
our work to more general classes of generators.

Another important aspect that remains outstanding is how to optimally approximate a given time evolution
operator $U(t)$. We have already argued that the order of the factorized time-evolution operator can greatly
affect the faithfulness of an approximation. We therefore propose that the complementary question to the
ability to factorize is: \emph{how can we optimize an approximate factorized expression of the time-evolution operator}? Intuition leads us to two
observations: first, given a Hamiltonian that has an infinite-dimensional Lie algebra, it is clear
that very high-order interaction terms (i.e., those terms $g_\sigma^\gamma$ with large degree $|\gamma|$)
will in general be driven by a small coupling. The reason is that such terms appear as a consequence of
high-order commutators and therefore the overall strength that drives them comes as a product of many
couplings of lower order interaction terms. Usually such couplings are small and therefore their product
leads to negligible contributions. This can also be seen directly within approaches that employ the
first few orders of the Campbell-Baker-Hausdorff formula. Second, terms of high degree might require more energy to be implemented.
In fact, such terms might have a degree that involves creating dozens of excitations, an operation that will be much
more costly than, say, exciting the lower energy level twice. These two considerations are complementary
and would seem to suggest that interaction terms in a Lie algebra obtained by means of many
commutators of generators will in general have subdominant effects. Further work is necessary to address
this important aspect systematically, and to provide general yet useful bounds to be employed in concrete tasks.

The conditions provided in our main results allow for a fast and
efficient algorithm to verify the dimensionality of the Lie algebras for any set of generators. This, in turn, is an important component in advancing the development of quantum technologies. In fact, once the tasks mentioned here have been accomplished, this body of work will aid in tackling fundamental
questions of quantum computing \cite{Ladd:Jelezko:2010,Nielsen:Chuang:2000}, quantum
simulation \cite{Buluta:Nori:2009,Bloch:2012ty,Georgescu:Ashhab:2013}, quantum annealing \cite{Apolloni:Carvalho:1989,Kadowaki:Nishimori:1998,Ohzeki:Nishimori:2011,Hauke:Katzgraber:2020}, and quantum control \cite{Huang:Tarn:1983,Wu:Tarn:2006,Bagarello:2007,Dong:Petersen:2010}.

\section*{Conclusions}
In this work we have addressed the long-standing question of determining the finite factorization of quantum
dynamics of coupled harmonic oscillators with tunable interactions. As a concrete step towards answering this question we
have introduced tools to classify the elements of the skew-hermitian Weyl algebra in terms of meaningful
subvector spaces. The key result is a set of conditions that
the generators of the Lie algebra, associated to a Hamiltonian
without drift, must satisfy in order to obtain a finite-dimensional Lie algebra. These conditions
require only direct verification of properties of the generators in relation to the subvector spaces introduced
before, or at most the properties of one degree of commutators between them. In this sense, the conditions
determined here provide an economical algorithm for verification of a problem that might seem a priori extremely
complicated, if not impossible to solve. As a side consequence, the conditions also inform on which
physical systems without drift Hamiltonian
can potentially be simulated exactly, i.e., which systems have finite-dimenstional Lie algebras.

Our efforts required us to obtain a specific condition on the pairs of operators that
allows us to predict if the degree of their commutator is maximal.
Furthermore, our efforts also required us to define and study interesting structures called \textit{commutator chains}, 
which are meaningfully constructed sequences of operators
in the full skew-hermitian Weyl algebra that are uniquely associated to pairs of chosen
\textit{seed} operators. In the cases where these chains can be constructed, they can provide information regarding
the dimensionality of the Lie algebra that they belong to by means of the control over the degree of their elements.
Taken all together these tools provide interesting insights in the structure of Lie algebras.

Applications of the results presented can be found in the areas of quantum annealing, quantum computing
and quantum simulations to name a few. Furthermore, they can aid our understanding of the ultimate ability to
simulate physical processes, if simulation of a system is taken to require exact control over its dynamics. Additional work is necessary to provide a comprehensive answer to these important questions.

Our work opens up a programme to study Lie algebras in the context
of coupled bosonic systems with tunable interactions, with the ultimate goal of answering the question of
factorization of their quantum dynamics. We believe that the results presented here have far reaching consequences
for both quantum physics and the theory of Lie algebras.
		
\section*{Acknowledgments}
We thank Niklas Jung, Nicolas Wittler, Leila Khouri, Valente Pranubon, Daniel Klaus Burgarth, J\'ozsef Zsolt Bernad, and Johannes Niediek for useful comments and discussions. 
We thank Tim Heib for carefully reading an earlier version of this work and for diligently
collecting comments and corrections.

\section*{Author Contributions}
DEB and AX conceived the idea of studying finite-dimensional bosonic Lie algebras as a stepping stone
to factorizations of the corresponding bosonic quantum dynamics. DEB developed the general approach
for deciding finiteness and he outlined and detailed initial versions of claims, their proofs, and
most of their supporting results. RZ provided links to the mathematical literature as well as techniques
and supporting notation which were used by DEB and RZ to systematically refine, simplify, sharpen,
generalize, and polish all claims and proofs.
DEB produced the figures. All authors contributed towards the preparation of the manuscript.

\section*{Funding}
DEB
acknowledges support from the joint project No. 13N15685 ``German Quantum Computer based on Superconducting
Qubits (GeQCoS)'' sponsored by the German Federal Ministry of Education and Research (BMBF) under the
\href{https://www.quantentechnologien.de/fileadmin/public/Redaktion/Dokumente/PDF/Publikationen/Federal-Government-Framework-Programme-Quantum-technologies-2018-bf-C1.pdf}{framework programme “Quantum technologies -- from basic research to the market”}, 
as well as the German Federal Ministry of Education and Research via
the \href{https://www.quantentechnologien.de/fileadmin/public/Redaktion/Dokumente/PDF/Publikationen/Federal-Government-Framework-Programme-Quantum-technologies-2018-bf-C1.pdf}{framework programme “Quantum technologies -- from basic research to the market”}
under contract number 13N16210 ``Diamond spin-photon-based quantum computer (SPINNING)''. RZ acknowledges funding
under Horizon Europe programme HORIZON-CL4-2022-QUANTUM-02-SGA via the project 101113690 (PASQuanS2.1).

\section*{Data availability}
Data sharing is not applicable to this work as no datasets are used to support the presented results.

\section*{Conflict of interest}
The authors have no relevant financial or non-financial interests to disclose.

\bibliographystyle{apsrev4-2}
\bibliography{DecHamBiblio}

\end{document}